\documentclass[english,a4paper,12pt,
openright,twoside]{book}
\evensidemargin\oddsidemargin
\global\parskip 6pt
\usepackage[T1]{fontenc}
\usepackage[latin1]{inputenc}
\usepackage{babel}
\usepackage{amsmath}
\usepackage{amsthm}
\usepackage{bbm}
\usepackage{amssymb}
\usepackage{mathrsfs}
\usepackage{color}
\usepackage[pdftex]{graphicx}
\usepackage[bookmarks]{hyperref}
\usepackage{hyperref}
\usepackage{stmaryrd}
\usepackage{mathrsfs}
\usepackage{dsfont}
\usepackage{wasysym}
\usepackage{pdfpages} 

\def\bea{\begin{eqnarray}}
\def\eea{\end{eqnarray}}
\def\be{\begin{equation}}
\def\ee{\end{equation}}
\usepackage{fancyhdr}

\fancypagestyle{plain}{%
\fancyhead{} 
}

\newtheoremstyle{eg}
  {3pt}
  {3pt}
  {\small}
  {\parindent}
  {\scshape}
  {:}
  {.5em}
  {}

\theoremstyle{eg}
\newtheorem{theorem}{Theorem}
\newtheorem{definition}[theorem]{Definition}
\newtheorem{example}[theorem]{Example}
\newtheorem{proposition}[theorem]{Proposition}
\newtheorem{remark}[theorem]{Remark}

\def\M{\mathcal{M}}
\def\P{\mathcal{P}}

\def\X{\mathcal{X}}
\def\U{\mathcal{U}}
\def\H{\mathcal{H}}
\def\A{\mathcal{A}}
\def\F{\mathcal{F}}
\def\W{\mathcal{W}}
\begin{document}
\pagestyle{empty}
\pagenumbering{roman}

 \thispagestyle{empty}

%

\begin{center}

\textsc{\LARGE University of Wroc{\l}aw\\
    Department of Physics and Astronomy\\
    } \medbreak
{\LARGE Institute for Theoretical Physics}
\end{center}
\vspace{3.0cm}
\begin{center}
    
{\Large PhD Dissertation}\\

\end{center}
\vspace{1.5cm}
\begin{center}
    
\textsc{\Huge "$\kappa$- Minkowski spacetime:
 mathematical formalism and applications in Planck scale physics"}
\end{center}
\vspace{2.0cm}

    by

\textsc{\huge Anna Pacho{\l}}\\
\vspace{1.0cm}
\begin{center}
    
\textsc{\huge Supervisor\\ Dr hab. Andrzej Borowiec}\\

\end{center}
\vspace{2.0cm}
\begin{center}
\large
    {Wroc{\l}aw, 2011}
\end{center}

\newpage
\thispagestyle{empty}
{\color{white}.}
\newpage

\thispagestyle{empty}
{\color{white}.}
\vspace{14cm}

\hspace{\stretch{1}} \begin{flushright}
    \emph{ "To those who do not know mathematics\\
 it is difficult to get across a real feeling as to the beauty,\\ the deepest beauty, of nature ...\\ If you want to learn about nature,\\ to appreciate nature, it is necessary\\ to understand the language that she speaks in."
}
\end{flushright}
\begin{flushright}
    Richard Phillips Feynman
\end{flushright}

\vspace{\stretch{1}}

\chapter*{Preface}
\footnotesize{This dissertation is based on research done at the Institute for Theoretical Physics University of Wroc{\l}aw between October 2007 and June 2011 in collaboration with Dr hab. Andrzej Borowiec and in Part II also with Prof. Stjepan Meljanac (Rudjer Boskovic Institute, Zagreb, Croatia) and Prof. Kumar S. Gupta (Theory Division, Saha Institute of Nuclear Physics, Calcutta, India). This thesis contains material which has been published in the following papers:}
\begin{itemize}
\item Andrzej Borowiec and Anna Pacho{\l} \href{http://arxiv.org/abs/1005.4429}{ "kappa-Minkowski spacetimes and DSR algebras: fresh look and old problems"}\\ SIGMA 6, 086 (2010) [arXiv:1005.4429].
\item  Andrzej Borowiec, Kumar S. Gupta, Stjepan Meljanac and Anna Pacho{\l}
\href{http://arxiv.org/abs/0912.3299}{
"Constraints on the quantum gravity scale from kappa - Minkowski spacetime"}\\ EPL 92, 20006 (2010)  [arXiv:0912.3299].
\item Andrzej Borowiec and Anna Pacho{\l}  \href{http://arxiv.org/abs/0903.5251} {"The classical basis for kappa-Poincare algebra and doubly special relativity theories"}\\ J. Phys. A: Math. Theor. 43, 045203 (2010) [arXiv:0903.5251].
\item Andrzej Borowiec and Anna Pacho{\l}  \href{http://arxiv.org/abs/0812.0576}{"kappa-Minkowski spacetime as the result of Jordanian twist deformation"}\\ Phys. Rev. D 79, 045012  (2009) [arXiv:0812.0576].
\item Andrzej Borowiec and Anna Pacho{\l}
 "On Heisenberg doubles of quantized Poincare algebras"\\ Theoretical and Mathematical Physics, 169(2), 1611 (2011)

\end{itemize}
\normalsize
\begin{center}
    \textbf{Acknowledgements}\\

 {\it I am deeply grateful to my supervisor Dr A. Borowiec,
 for his help and guidance,\\
 for the infinitely many discussions and his endless patience.\\
Special thanks go to Prof. S. Meljanac and his group for helpful discussions and their hospitality during my visits at Rudjer Boskovic Institute in Zagreb (Croatia).\newline

\footnotesize{This PhD dissertation has been also supported by The National Centre for Science Research Grant (N N202 238540) and the scholarship within project: "Rozw{\'o}j potencja{\l}u i oferty edukacyjnej Uniwersytetu Wroc{\l}awskiego szansa zwiekszenia konkurencyjnosci Uczelni."}}

\end{center}

\tableofcontents
\pagestyle{fancy}
\fancyhf{}
\fancyhead[LE,RO]{\bfseries\thepage}
\fancyhead[RE,LO]{\textbf{Contents}}

\chapter*{Abstract}
\fancyhf{}

\fancyhead[LE,RO]{\bfseries\thepage}
\fancyhead[RE,LO]{Abstract}
The dissertation presents possibilities of applying noncommutative spacetimes description ,
particularly kappa-deformed Minkowski spacetime and Drinfelds deformation theory, as a mathematical formalism for Doubly Special Relativity theories (DSR), which are thought as phenomenological limit
of quantum gravity theory. Deformed relativistic symmetries are described within Hopf algebra
language. In the case of (quantum) kappa-Minkowski spacetime the symmetry group is described by the (quantum)
kappa-Poincare Hopf algebra. Deformed relativistic symmetries were used to construct the DSR algebra, which unifies noncommutative coordinates with generators of the symmetry algebra. It contains the
deformed Heisenberg-Weyl subalgebra. It was proved that DSR algebra can be obtained by nonlinear
change of generators from undeformed algebra. We show that the possibility of applications in  Planck scale physics is connected with certain realizations of quantum spacetime, which in turn leads to deformed dispersion relations.


\chapter*{Introduction}
\fancyhf{}
\fancyhead[LE,RO]{\bfseries\thepage}
\fancyhead[RE,LO]{\textbf{Introduction}}
\addcontentsline{toc}{chapter}{Introduction}

 One of the most intriguing theoretical problems nowadays is the search for fundamental theory describing Planck scale physics. The so called "Planck scale" is the scale at which gravitational effects are equivalently strong as quantum ones and it relates to either a very big energy scale or equivalently to a tiny size scale. A new theory which would consistently describe this energy region is necessary to understand the real nature of the Universe. It would allow, for example, to give a meaning to the Big Bang and black holes, physical situations were both quantum mechanics and general relativity are relevant, but cannot be applied consistently at the same time. Such theory, in the weak gravity limit, should give back already known Quantum Mechanics and in the limit of Planck's constant going to zero it should reduce to Einstein's Gravity. The quantum field theory description at the Planck scale breaks down due to the nonrenormalisability of Einstein's theory of gravity. Therefore the search for a theory of Quantum Gravity, describing gravitational interactions at the quantum level is one of the most important problems in modern physics. There are few theories trying to solve that problem, like String Theory, Loop Quantum Gravity, Noncommutative Geometry and many more.  The topic of this thesis is connected with noncommutative spacetimes as one of the approaches to the description of Planck scale physics, particularly the description of geometry from the quantum mechanical point of view. At the Planck scale the idea of size or distance in classical terms is not valid any more, because one has to take into account quantum uncertainty. The Compton wavelength of any photon sent to probe the realm at this energy scale will be of the order of Schwarzschild radius. Therefore a photon sent to probe object of the Planck size (Planck length $L_p=1.62×10^{-35} m$) will be massive enough to create a black hole, so it will not be able to carry any information. The research on the structure of spacetime, at the scale where quantum gravity effects take place, is one of the most important questions in fundamental physics. One of the possibilities is to consider noncommutative spacetimes. Since in Quantum Mechanics and Quantum Field Theory the classical variables become noncommutative in the quantization procedure and in the General Theory of Relativity spacetime is a dynamical variable itself, one could assume that noncommutative spacetime will be one of the properties of physics at the Planck scale.  In noncommutative spacetime the quantum gravity effects modify the coordinate relations in the quantum phase space and, besides Heisenberg relations between coordinates and momenta, one has to introduce noncommutativity of coordinates themselves. This leads to new uncertainty relations for coordinates, which exclude the existence of Planck particles (of the size of the Schwarzschild radius) from the theory.\\

The first attempts of constructing a field theory on noncommutative spacetimes are related with Heisenberg's ideas from the 30's. In the 1940's, Snyder proposed the model of Lorentzian invariant discrete spacetime as a first example of noncommutative spacetime \cite{Snyder}. Noncommutative coordinates lead to a modification of relativistic symmetries, which are described by the Quantum Groups and Quantum Lie algebras. However the Theory of Quantum Groups has been developed independently at first. The mathematical formalisms connected with Hopf algebras, were introduced and used in mathematical physics by V. Drinfeld, L.D. Faddeev,  M. Jimbo, S. Woronowicz and Y. Manin in the 80's. The term "quantum group" first appeared in the theory of quantum integrable systems and later was formalized by V. Drinfeld and M. Jimbo as a particular class of Hopf algebras with connection to deformation theory (strictly speaking, as deformation of universal enveloping Lie algebra). Such deformations are classified in terms of classical r-matrix satisfying the classical Yang-Baxter equation. However Noncommutative Geometry as an independent area of mathematics appeared together with papers by Alain Connes and since then it found applications in many theories as String Theory, Noncommutative Quantum Field Theory (NC QFT) and NC Gauge Theories, NC Gravity, as well as in noncommutative version of the Standard Model. Nonetheless the Hopf algebras (Quantum Groups) formalism is strictly connected with noncommutative spacetimes and Quantum Groups are thought as generalizations of symmetry groups of the underlying spacetime. Mathematically speaking such noncommutative spacetime is the Hopf module algebra (a.k.a. covariant quantum space) and stays invariant under Quantum Group transformations. \\

The first example of deformed coordinate space as a background for the unification of gravity and quantum field theory was introduced in \cite{Dop1},\cite{Dop2}. This spacetime quantization effect is induced by classical gravity and it is called theta (or canonical)-noncommutativity. This type of commutation relations was also found in String Theory with a constant background field, where coordinates on the spacetime manifolds at the end of strings (D-branes) do not commute \cite{SW}. Since then it has been also the subject of many investigations concerning NC QFT defined over NC spaces with twisted Poincare symmetry \cite{Chai1,Chai2}, and NC Gauge Theories \cite{ADMSW}, NC Gravity \cite{ABDMSW,ADMW,Szabo2006}.

Another type of noncommutativity is the so-called Lie algebraic one, where coordinate commutation relations are of the Lie algebra type. This deformation was inspired by the $\kappa$-deformed Poincar\'{e} algebra \cite{Luk1,Luk2,MR} as deformed symmetry of the $\kappa$-Minkowski spacetime \cite{Z}. The deformation parameter $\kappa$ is of mass dimension, hence usually connected with Planck mass $M_P$ or more recently $M_{QG}$ quantum gravity scale. $\kappa$- deformation has been used by many authors, as a starting point to construct quantum field theories \cite{kappaQFT}, or modification of particle statistics  \cite{kappaSt} and then to discuss Planck scale physics, since it naturally included 'kappa' (Planck scale) corrections. Moreover the algebraic structure of the $\kappa$- Poincar\'{e} algebra has been and still is intensively analyzed  from mathematical and physical point of view by many authors \cite{Luk1,Luk2,MR,LRZ,KosLuk,Kos,GNbazy}.  Important result concerning bi-covariant differential calculus on $\kappa$-Minkowski spacetime has been shown in \cite{Sitarz}.
Recently also research on mathematical properties of $\kappa$-Minkowski spacetime from C*-algebra point of view experienced growing interest. Firstly the description of $\kappa$-Minkowski algebra as algebra of operators represented in Hilbert space was proposed in \cite{Ago} and then was followed by \cite{LD-GP}, \cite{SitDur}.
The past decade also contributed to new possibilities for applications of $\kappa$-deformation formalism thanks to the the so-called deformed special relativity theories (DSR), originally called doubly special relativity. A chance for physical application of  $\kappa$-deformation appeared when an
extension of special relativity was proposed in Refs. \cite{AC,BACKG}, and another one, showing different point of view, in Refs. \cite{Smolin} \footnote{ For comparison of these two approaches see, e.g., \cite{comparison}.}. The basic idea behind this extension is to consider two invariant parameters, the speed of light as in Special Relativity and a mass (length) parameter, e.g. the Planck Mass (length). Recently, various phenomenological aspects of DSR theories have been studied in, e.g., \cite{phenom}. It was also suggested that DSR can be thought as a phenomenological limit of quantum gravity theories. 
The connection between $\kappa$ -deformation and DSR theory in first formulation (DSR1)
has been shown (see, e.g. \cite{GNbazy}, \cite{BACKG}, \cite{LukDSR}) including the conclusion that the spacetime of DSR must be noncommutative as the result of Hopf structure of
this algebra.  Together with this connection a physical interpretation for deformation parameter $\kappa$, as second invariant scale, in $\kappa$-Minkowski spacetime appears naturally and allows us to interpret deformed dispersion relations as valid at the "$\kappa$-scale" (as Planck scale or Quantum Gravity scale) when quantum gravity corrections become relevant. Recent studies show that DSR theories might be experimentally falsifiable \cite{smolin1,DSRdebate, DSRdebateI,DSRnonloc}, which places $\kappa$-Minkowski spacetime on the frontier between strictly mathematical theoretical example of noncommutative algebra and physical theory describing Planck scale nature. However as it has been noticed \cite{ACdsrmyth} that $\kappa$-Minkowski/$\kappa$-Poincare formalism is not the only one which can be appropriate for description of DSR theories. Nevertheless it seems to be very fruitful and becoming even more promising providing framework for recent suggestions of experimentally testing quantum gravity effects. 

The mathematical framework connected with Quantum Groups and deformation theory in the case of
$\kappa$ -Poincar\'{e} and $\kappa$-Minkowski spacetime constitutes the first part of this thesis. General setup concerning Hopf algebras, twist deformation and smash (crossed) product construction is reminded in Chapter 1. Then $\kappa$-Minkowski spacetime from twist-deformation point of view is introduced in Chapter 2. Two families of twists: Abelian and Jordanian are considered, where the module algebra of functions over the universal enveloping algebra of inhomogeneous general linear algebra is equipped in star-product and can be represented by star-commutators of noncommuting ($\kappa$-Minkowski) coordinates. This approach is connected with the Drinfeld twist deformation technique applied to inhomogeneous general linear algebra. The corresponding family of realizations of $\kappa$-Minkowski spacetime as quantum spaces over deformed $\mathfrak{igl}(n)$ algebra is introduced together with the deformation of the minimal Weyl-Poincare algebra. Chapter 3 is focused on the mathematical properties of the quantum $\kappa$-Poincare group. $\kappa$ -Poincar\'{e} Hopf algebra is  considered as quantum deformation of the Drinfeld-Jimbo type, since its classical r-matrix satisfies the modified Yang-Baxter equation. The Drinfeld-Jimbo deformed Poincare algebra is described together with the deformed d'Alembert operator, which implies a deformation on dispersion relations. Two points of view are considered: a traditional one, with h-adic topology (Section 4.1), and a new one, with a reformulation of the algebraic sector of the corresponding Hopf algebra (the so-called q-analog version, which is no longer of the form of universal enveloping algebra) (Section 4.2). Correspondingly one can consider two non-isomorphic versions of $\kappa$-Minkowski spacetime (i.e. h-adic and q-analog ones), which in both cases is a covariant quantum space (i.e. Hopf module algebra over $\kappa$-Poincare Hopf algebra). Part I of this thesis ends with smash (a.k.a. crossed) product of Hopf algebras with their modules, in the case of $\kappa$-Poincare and $\kappa$-Minkowski, the smash-algebra is called the DSR algebra. It uniquely unifies $\kappa$-Minkowski spacetime coordinates with $\kappa$-Poincare generators. It is also proved that this DSR algebra can be obtained by a non-linear change of the generators from the undeformed one.  The Deformed Phase Space (Weyl) algebra is also described as a subalgebra of the DSR algebra, together with momenta and coordinates realizations. Again, it turns out to be obtained by change of generators in the undeformed one.

Part II concerns $\kappa$-Poincare algebra and $\kappa$-Minkowski spacetime formalism as one of the possible candidate for description of the Planck scale realm. Especially as formalism standing behind 
"Doubly (deformed) Special Relativity" (DSR). Moreover it was recently proposed that the measurement of quantum gravity effects might be possible thanks to the difference in arrival time of high energy photons coming from gamma ray bursts (GRB's). This "time delays" appear due to modified dispersion relations. The dispersion relations following from DSR are consistent with this effect and lead to Lorentz symmetry violation. Deformed dispersion relations which follow from $\kappa$-deformed d'Alambert operator lead to such "time delays" as well. Therefore one of the possibilities of the "time delay" mathematical framework is to consider, e.g. spacetime noncommutativity. In this part of the thesis two versions (Abelian and Jordanian) of deformed dispersion relations were compared and the corresponding time delays were calculated up to the second order accuracy in the quantum gravity scale (deformation parameter). Furthermore
the general framework describing modified dispersion relations and time delays with respect to different noncommutative $\kappa$-Minkowski spacetime realizations was proposed. It was also shown that some of the realizations provide certain bounds on quadratic corrections, i.e. on the quantum gravity scale consistent with $\kappa$-Minkowski noncommutativity.\\


\newpage
 \thispagestyle{empty}
\vspace{2cm}

\begin{center}
    {\LARGE PART I}
\end{center}

\vspace{1.5cm}
\begin{center}
    
\textsc{ \Huge Mathematical formalism}

\end{center}
\vspace{9.0cm}
\hspace{\stretch{1}} \begin{flushright}
     \emph{"Mathematics may be defined as the subject\\ in which we never know what we are talking about,\\ nor whether what we are saying is true."}
\end{flushright}
\begin{flushright}
    Bertrand Russell 
\end{flushright}

\vspace{1.0cm}
\hspace{\stretch{1}} \begin{flushright}
     \emph{"As far as the laws of mathematics refer to reality,\\ they are not certain, and as far as they are certain,\\ they do not refer to reality."}
\end{flushright}
   \begin{flushright}
     Albert Einstein
\end{flushright}

\vspace{5.0cm}

\newpage
\newpage
\thispagestyle{empty}

\chapter{Noncommutative Spacetimes and Their Symmetries}
\pagenumbering{arabic}
\fancyhf{} 
\fancyhead[LE,RO]{\bfseries\thepage}
\fancyhead[LO]{\bfseries\rightmark}
\fancyhead[RE]{\bfseries\leftmark}
Classically Minkowski spacetime is a four-dimensional manifold  (an affine space) equipped with a nondegenerate, symmetric bilinear form with lorentizian signature and it constitutes the mathematical setting in which Einstein's theory of special relativity is most conveniently formulated. A natural quantization of this manifold can be described in the language of non-commutative geometry. Motivation for this fact starts with Israel Gelfand and Mark Naimark in 1943 with their theorem which states that there is a one-to-one correspondence between commutative $C^*$-algebras and topological (locally compact, Hausdorff) spaces \cite{Gelfand}, see also \cite{GarciaNC}. In compact case one has:
\begin{theorem}
Every compact, Hausdorff space, is equivalent
to the spectrum of the $C^*$-algebra of continuous, complex valued
functions on that space.
\end{theorem}
Later on it was proven by Alain Connes that certain commutative algebras with some additional structure (the so-called spectral triples) are in one-to-one relation with compact Riemannian spin manifolds. Therefore one can see that geometry can be "written" in the algebra language. The idea of non-commutative geometry is to consider non-commutative algebras as non-commutative geometric spaces, i.e. to 'algebralize' geometric notions and then generalize them to noncommutative algebras  \cite{Connes}. For example such noncommutative algebras one can obtain via deformation procedure from commutative ones. Quantum deformations, which lead to noncommutative algebras (as noncommutative spacetimes), are strictly connected with the Quantum Groups (Hopf algebras) formalism, which constitutes the description of deformed symmetries. Mathematically speaking such noncommutative spacetime is a Hopf module algebra over Hopf algebra (Quantum Group) of symmetry and quantum noncommutative spacetime stays invariant under quantum group of transformations in analogy to the classical case.
Therefore one can say (physically) that quantum group is noncommutative generalization of a symmetry group of the physical system. Quantum groups are the deformations of some classical structures as groups or Lie algebras, which are made in the category of Hopf algebras. Similarly, quantum spaces are noncommutative generalizations (deformations) of ordinary spaces.

 One of the possible deformations which lead to noncommutative spacetimes is the Drinfeld's quantum deformation technique and it can be applied in quantizing differential manifolds. Vector fields provide an infinite dimensional Lie algebra and act naturally, as derivations, on the algebra of functions on a manifold. However this algebra is simultaneously a Hopf module algebra over the Hopf algebra generated by vector fields and it can be deformed by relevant twisting
element. The deformed algebra of functions is the so called noncommutative spacetime (a.k.a. covariant quantum space) and it contributes to the realization of Connes' idea of noncommutative space (manifold) and might be helpful in quantizing gravity. The quantum Hopf algebra has an undeformed algebraic sector (commutators) and a deformed co-algebraic sector
(coproducts and antipodes). The twisted module algebra is now equipped with twisted star-product and can be represented by deformed (star)-commutators of noncommutative coordinates which replace the commutative ones and play the role of generators of that algebra. The whole deformation depends on a formal parameter which controls classical limit. Simultaneous deformation of coproduct in Hopf algebra and product in Hopf-module algebra with undeformed action, forces us to deform
accordingly the Leibniz rule. In the following this procedure will be explicitly shown on the examples of Lie-algebraic type of noncommutative algebra, i.e. $\kappa$- Minkowski spacetime. Firstly for quantum group of symmetry for $\kappa$- Minkowski spacetime one considers twisted deformation of universal envelope of inhomogeneus linear algebra ($\U_{\mathfrak{igl}}[[h]]^\F$). 
However more suitable candidate for symmetry of $\kappa$- Minkowski spacetime from physical point of view is $\kappa$-Poincare quantum group, since it gives link with relativistic physics.
%
 But before getting to that point let us start with general remarks on Hopf algebras and their modules.

\section{Preliminaries}

Let us begin with the basic notion on a Hopf algebra considered as symmetry algebra (quantum group) of another algebra representing quantum space: the so-called  module algebra or  covariant quantum space. A Hopf algebra is a complex\footnote{We shall work mainly with vector spaces over the f\/ield of complex numbers $\mathbb{C}$. All maps are
$\mathbb{C}$-linear maps.}, unital and associative algebra equipped with additional structures such as a comultiplication, a counit and an antipode.
\begin{definition}\label{HA}
Let $\H$ be a vector space equipped simultaneously with an algebra structure 
$(\H,m,\eta)$ and a coalgebra structure $(\H, \Delta, \epsilon)$.\\
$(\H,m,\eta)$ as an associative algebra with unit and two linear maps $m : \H\otimes\H\shortrightarrow\H$, called 
the multiplication or the product, and $\eta: \mathbb{C}\shortrightarrow\H$, called the unit, such that:
\begin{equation}
m\circ (m \otimes id) = m \circ (id \otimes m);\qquad
m\circ(\eta \otimes id) = id = m \circ (id \otimes \eta) 
\end{equation}
Algebra $\H$ is equipped with product: $LM=m(L\otimes M)$ ($L,M\in\H$) and the map $\eta$ is determined by its value $\eta(1)\in\H$ is the unit element of $\H$.\\
$(\H, \Delta, \epsilon)$ as coalgebra is a vector space $\H$ over $\mathbb{C}$ equipped with two linear mappings\\ $\Delta : \H \shortrightarrow\H \otimes \H$, called the comultiplication or the coproduct, and $\epsilon: \H\shortrightarrow\mathbb{C}$ called the counit, such that: 
\begin{equation}
( \Delta\otimes id) \circ \Delta = (id \otimes \Delta) \circ \Delta;\qquad 
(\epsilon \otimes id) \circ\Delta = id = (id \otimes \epsilon) \circ \Delta
\end{equation}
$\Delta$ and $\epsilon$ are algebra homomorphisms, $m$ and $\eta$ are coalgebra homomorphisms.  
\begin{equation}
\Delta(LM)=\Delta(L)\Delta(M);\qquad\epsilon(LM)=\epsilon(L)\epsilon(M);
\end{equation}
for $L,M\in\H$ and $\Delta(1)=1\otimes 1;\qquad\epsilon(1)=1$.\\
Such bialgebra $\H$ is called \textbf{Hopf algebra} if there exists a linear 
mapping $S :\H\shortrightarrow \H$, called the antipode or the coinverse of $\H$, such that :
\begin{equation}
m\circ (S\otimes id)\circ \Delta =\eta\circ\epsilon = m\circ(id\otimes S)\circ\Delta 
\end{equation}
The antipode $S$ of a Hopf algebra $\H$ is an algebra anti-homomorphism and a coalgebra anti-homomorphism of $\H$.
\end{definition}
 In some sense,
these structures and their axioms reflect the multiplication, the unit element
and the inverse elements of a group and their corresponding properties \cite{Klimyk, Kassel}.
Let us also introduce the so-called Sweedler notation for the comultiplication. 
If $a$ is an element of a coalgebra $\H$, the element $\Delta(a)\in\H\otimes\H$ is a finite sum 
$\Delta(a)=\sum_i a^i_{(1)}\otimes a_{(2)}^i,\quad a^i_{(1)},a^i_{(2)}\in\H$. As a shortcut we shall simply write: $\Delta(a)=a_{(1)}\otimes a_{(2)}$.
\begin{example}\label{example1}
Any Lie algebra $\mathfrak{g}$ provides an example of the (undeformed) Hopf algebra by taking its universal enveloping algebra $\U_{\mathfrak{g}}$
equipped with the primitive coproduct: $\Delta_0(u)=u\otimes 1+1\otimes u$, counit:  $\epsilon(u) = 0$, $\epsilon(1) = 1$ and antipode: $S_0(u) = -u$, $S_0(1) = 1$, for $u\in \mathfrak{g}$, and extending them by multiplicativity property to the entire $\U_{\mathfrak{g}}$. Recall that the universal enveloping algebra is a result of the factor construction
\begin{gather}\label{UEnvelop}
    \U_{\mathfrak{g}}=\frac{T\mathfrak{g}}{J_{\mathfrak{g}}},
\end{gather}
where $T\mathfrak{g}$ denotes tensor (free) algebra of the vector space $\mathfrak{g}$ quotient out by the ideal $J_\mathfrak{g}$ generated by elements $\langle X\otimes Y-Y\otimes X-[X, Y]\rangle$: $X, Y \in \mathfrak{g}$.
\end{example}
\begin{definition}\label{module}
A (left) $\H$-module algebra $\A$ over a Hopf algebra $\H$ is an algebra  $\A$ which is a left
$\H$-module such that $m:\A\otimes\A\shortrightarrow\A$ and $\eta:\mathbb{C}\shortrightarrow\A$ are left $\H$-module homomorphisms.
If $\triangleright:\H\otimes\A\rightarrow\A$ denotes (left) module action 
$L\triangleright f$ of $L\in\H$ on $f\in\A$ the following compatibility
condition is satisfied: 
\begin{gather}\label{genLeibniz}
L\triangleright (f\cdot g)=(L_{(1)}\triangleright f)\cdot (L_{(2)}\triangleright g)
\end{gather}
for $L\in\H$, $f,g\in\A$ and $L\triangleright 1=\epsilon(L)1$, $1\triangleright f=f$ (see, e.g., \cite{Klimyk, Kassel}).\\
And analogously for right $\mathcal{H}$-module algebra $\mathcal{A}$ the condition:
 $$(f\cdot g)\triangleleft L=(f\triangleleft L_{(1)})\cdot (g\triangleleft L_{(2)})$$
is satisfied, with (right)-module action $\triangleleft :\mathcal{A}\otimes 
\mathcal{H}\rightarrow \mathcal{A}$; for $L\in \mathcal{H}$, $f,g\in 
\mathcal{A}$, $1\triangleleft L=\epsilon (L)1$, $f\triangleleft 1=f$.
\end{definition}
In mathematics the module condition plays a role of generalized Leibniz rule and it invokes exactly the Leibniz rule for primitive elements $\Delta(L)=L\otimes 1+1\otimes L$.
Instead, in physically oriented literature it is customary to call this condition a covariance condition and the corresponding algebra $\A$ a covariant quantum space (see, e.g., \cite{ABDMSW,DJMTWW,covQS}) with respect to $\H$. The covariance condition (\ref{genLeibniz}) entitles us also to introduce a new unital and associative algebra, the so-called smash (or crossed) product algebra $\A\rtimes\H$ \cite{Klimyk, Kassel,Majid2,smash}.
\begin{definition}
Smash (crossed) product is determined on the vector space $\A\otimes\H$ by multiplication:
\begin{gather}\label{smash}
(f\otimes L)\rtimes(g\otimes M)=f(L_{(1)}\triangleright g)\otimes L_{(2)}M
\end{gather}
(in the case of $\mathcal{A}$ being left $\mathcal{H}$-module). Initial algebras are canonically embedded, $\A\ni f\shortrightarrow f\otimes 1$ and $\H\ni L\shortrightarrow 1\otimes L$~as subalgebras in $\A\rtimes\H$.\footnote{Further on, in order to simplify notation, we shall identify $(f\otimes 1)\rtimes(1\otimes L)$ with $fL$, therefore~(\ref{smash}) rewrites simply as $(fL)\rtimes(gM)=f(L_{(1)}\triangleright g)L_{(2)}M$ (and respectively for $\H \ltimes \A$).}\\

In the case of right $\mathcal{H}$-module algebra $\mathcal{A}$ the smash product $\mathcal{%
H\ltimes A}$ structure is determined on the vector space $\mathcal{H}\otimes 
\mathcal{A}$ by multiplication: $
(L\otimes f)\ltimes(M\otimes g)=LM_{(1)}\otimes (f\triangleleft M_{(2)})g$.
\end{definition}

 Particularly, the trivial action $L\triangleright g=\epsilon(L)g$ makes $\A\rtimes\H$ isomorphic to the ordinary tensor product algebra $\A\otimes\H$: $(f\otimes L)(g\otimes M)=fg\otimes LM$. It has a canonical Heisenberg representation on the vector space $\A$ which reads as follows:
\begin{gather*}
\hat{f}(g)= f\cdot g, \qquad \hat{L}(g)=L\triangleright g,
\end{gather*}
where $\hat{f}$, $\hat{L}$ are linear operators acting in $\A$, i.e.\ $\hat{f}, \hat{L}\in {\rm End}\, \A$. In other words around, the action~(\ref{genLeibniz}) extend to the action of entire $\A\rtimes\H$ on $\A$: $(f M)\triangleright g=f(M\triangleright g)$ (and respectively for $\H\ltimes \A$).

\begin{remark}\label{bicross}
 Note that $\A\rtimes\H$ (or $\H\ltimes\A$)  are not a Hopf algebras in general. In fact, enactment of the Hopf algebra structure on $\A\rtimes\H$ (or $\H\ltimes\A$) involves some extras (e.g., co-action) and is related with the so-called bicrossproduct  construction~\cite{Majid2} (see also Appendix \ref{app2}). It has been shown in~\cite{MR}  that $\kappa$-Poincar\'{e} Hopf algeb\-ra~\cite{Luk1} admits bicrossproduct construction.
\end{remark}
\begin{example}\label{example2}
Another interesting situation appears when the algebra $\A$ is a universal envelope of some Lie algebra $\mathfrak{h}$, i.e.\ $\A=\U_{\mathfrak{h}}$.
In this case it is enough to determine the Hopf action on generators $a_i\in \mathfrak{h}$ provided that the consistency conditions
\begin{gather}\label{smash2}
(L_{(1)}\triangleright a_i)\left( L_{(2)}\triangleright a_j\right) -(L_{(1)}\triangleright%
a_j)\left( L_{(2)}\triangleright a_i\right) -\imath c_{ij}^k L\triangleright a_k=0
\end{gather}
hold, where $[a_i, a_j]=\imath c_{ij}^ka_k$.
Clearly, (\ref{smash2}) allows to extend the action to entire algebra $\A$. For similar reasons the def\/inition of the action can be reduced to generators $L$ of $\H$.
\end{example}

\begin{example}\label{example3}
One more interesting case appears if $\H=\U_{\mathfrak{g}}$ for the Lie algebra
$\mathfrak{g}$ of some Lie group $G$ provided with the primitive Hopf algebra structure. This undergoes a "geometrization" procedures as follows. Assume one has given $G$-manifold $\M$. Thus $\mathfrak{g}$ acts via vector f\/ields on the (commutative) algebra of smooth functions $\A=C^\infty(\M)$.\footnote{In fact $\A$ can be chosen to be $\mathfrak{g}$-invariant subalgebra of $C^\infty(\M)$.} Therefore the Leibniz rule makes the compatibility conditions~(\ref{smash2}) warranted. The corresponding smash product algebra becomes an algebra of dif\/ferential operators on $\M$ with coef\/f\/icients in $\A$. A deformation of this geometric setting has been recently advocated as an alternative to quantization of gravity (see, e.g., \cite{ABDMSW,ADMW,Szabo2006} and references therein).
\end{example}

\begin{example}\label{example4}
A familiar Weyl algebra can be viewed as a crossed product of an algebra of translations $\mathfrak{T}^{n}$ containing $P_{\mu }$ generators with an algebra $\mathfrak{X}^n$ of spacetime coordinates $x^{\mu }$. More exactly, both algebras are def\/ined as a dual pair of the universal commutative algebras with $n$-generators (polynomial algebras), i.e.\ $\mathfrak{T}^{n}\equiv {\rm Poly}(P_\mu)\equiv\mathbb{C}[P_0,\ldots, P_{n-1}]$ and $\mathfrak{X}^{n}\equiv {\rm Poly}(x^\mu)\equiv\mathbb{C}[x^0,\ldots, x^{n-1}]$.\footnote{Here $n$  denotes a dimension of physical spacetime
which is not yet provided with any metric. Nevertheless, for the sake of future applications, we shall use "relativistic" notation with spacetime indices $\mu$ and $\nu $ running $0,\ldots,n-1$ and space indices $j,k=1,\ldots,n-1$ (see Example~\ref{ortho}).} Alternatively, both algebras are isomorphic to the universal enveloping algebra $\U_{\mathfrak{t}^{n}}\cong \mathfrak{T}^{n}\cong\mathfrak{X}^{n}$ of the $n$-dimensional Abelian Lie algebra $\mathfrak{t}^{n}$. Therefore one can make use of the primitive Hopf algebra structure on $\mathfrak{T}^{n}$ and extend the action implemented by duality map
\begin{gather}\label{actionP}
P_{\mu}\triangleright x^\nu=-\imath \langle P_{\mu}, x^\nu\rangle =-\imath\delta_\mu^\nu, \qquad
P_{\mu}\triangleright 1=0
\end{gather}
to whole algebra  $\mathfrak{X}^{n}$ due to the Leibniz rule, e.g.,
$P_{\mu}\triangleright (x^\nu x^\lambda)=-\imath\delta^\nu_\mu x^\lambda-\imath\delta^\lambda_\mu x^\nu$, induced by primitive coproduct $\Delta(P_\mu)=P_\mu\otimes 1+1\otimes P_\mu$, cf.~(\ref{smash}).
In result one obtains the following standard set of Weyl--Heisenberg commutation relations:
\begin{gather}\label{Weyl}
\left[ P_{\mu },x^{\nu }\right]_\rtimes\equiv\left[ P_{\mu },x^{\nu }\right]=-\imath\delta _{\mu}^\nu\, 1, \qquad
\left[ x^{\mu },x^{\nu }\right]_\rtimes\equiv\left[ x^{\mu },x^{\nu }\right]=
\left[ P_{\mu },P_{\nu }\right]_\rtimes\equiv\left[ P_{\mu },P_{\nu }\right]=0.
\end{gather}
as generating relations\footnote{Hereafter we skip denoting commutators with $\rtimes$ symbol when it is clear that they has been obtained by smash product construction.} for the Weyl algebra $\mathfrak{W}^{n}\equiv\mathfrak{X}^{n}\rtimes\mathfrak{T}^{n}$. In fact the algebra (\ref{Weyl}) represents the Heisenberg double \cite{Klimyk,Kassel,Hdouble}, which will be mentioned later in this Section. It means that $\mathfrak{T}^{n}$ and $\mathfrak{X}^{n}$ are dual pairs of Hopf algebras (with the primitive coproducts) and the action (\ref{actionP}) has the form: $P\triangleright x=\langle P, x_{(1)}\rangle x_{(2)}$.
In the Heisenberg representation $P_{\mu} \triangleright = -\imath\partial_\mu=-\imath\frac{\partial}{\partial x^\mu}$. For this reason  the
Weyl algebra is known as an algebra of dif\/ferential operators with polynomial coef\/f\/icients in $\mathbb{R}^n$.  \end{example}

\begin{remark}
The Weyl algebra as def\/ined above is not an enveloping algebra of any Lie algebra. It is due to the fact that the action (\ref{actionP}) is of ``$0$-order''. Therefore, it makes dif\/f\/icult to determine a Hopf algebra structure on it. The standard way to omit this problem relies on introducing the central element $C$ and replacing the commutation relations (\ref{Weyl}) by the following ones
\begin{gather}\label{Weyl--Heisenberg}
\left[ P_{\mu },x^{\nu }\right]=-\imath \delta _{\mu}^\nu C, \qquad
\left[ x^{\mu },x^{\nu }\right]=\left[C ,x^{\nu }\right]=
\left[ P_{\mu },P_{\nu }\right]=\left[C ,P_{\nu }\right]=0.
\end{gather}
The relations above determine $(2n+1)$-dimensional Heisenberg Lie algebra of rank $n+1$ . Thus Heisenberg algebra can be def\/ined as an enveloping algebra for (\ref{Weyl--Heisenberg}).
 We shall not follow this path here, however it may provide a starting point for Hopf algebraic deformations, see e.g.~\cite{LukMinn}.
\end{remark}

\begin{remark}
Real structure on a complex algebra can be def\/ined by an appropriate Hermitian conjugation~$\dag$. The most convenient way to introduce it, is by the indicating Hermitian generators. Thus the real structure corresponds to real algebras. For the Weyl algebra case, e.g., one can set the generators $(P_\mu, x^\nu)$ to be Hermitian, i.e.\ $(P_\mu)^\dag=P_\mu$, $(x^\nu)^\dag=x^\nu$. We shall not explore this point here. Nevertheless, all commutation relations below will be written in a form adopted to Hermitian realization. \end{remark}
\begin{remark}\label{Hspace}
Obviously, the Hopf action (\ref{actionP}) extends to the full algebra $C^\infty(\mathbb{R}^n)\otimes\mathbb{C}$ of complex valued smooth functions on $\mathbb{R}^n$. Its invariant subspace of compactly supported  functions $C_0^\infty(\mathbb{R}^n)\otimes\mathbb{C}$ form a dense domain in the Hilbert space of square-integrable functions: ${\cal L}^2(\mathbb{R}^n, dx^n)$. Consequently,  the Heisenberg  representation extends  to Hilbert space representation of $\mathfrak{W}^{n}$ by (unbounded) operators. This corresponds to canonical quantization procedure and in the relativistic case leads to St\"{u}ckelberg's version of Relativistic Quantum Mechanics~\cite{Menski}.
\end{remark}

\begin{example}\label{example5}
 Smash product generalizes also the notion of Lie algebra semidirect product. As an example one can consider a semidirect product of $\mathfrak{gl}(n)$ with the algebra of translations $\mathfrak{t}^n$: $\mathfrak{igl}(n)=\mathfrak{gl}(n)\niplus\mathfrak{t}^n$.  Thus $\U_{\mathfrak{igl}(n)}=\mathfrak{T}^{n}\rtimes\U_{\mathfrak{gl}(n)}$. Now the corresponding (left) Hopf action of $\mathfrak{gl}(n)$ 
on $\mathfrak{T}^n$ generators reads
\begin{gather*}
L^{\mu}_{\nu}\triangleright P_{\rho}=\imath\delta^\mu_\rho P_\nu.
\end{gather*}
The resulting algebra is described by a standard set of $\mathfrak{igl}(n)$ commutation relations:
\begin{gather}  \label{igl}
[L_{\nu }^{\mu },L_{\lambda }^{\rho }]=-\imath \delta _{\nu }^{\rho }L_{\lambda
}^{\mu }+\imath\delta _{\lambda }^{\mu }L_{\nu }^{\rho },\qquad    [L_{\nu
}^{\mu },P_{\lambda }]=\imath\delta _{\lambda }^{\mu }P_{\nu }, \qquad [P_\mu, P_\nu]=0.
\end{gather}

Analogously one can consider Weyl extension of $\mathfrak{gl}(n)$ as a
double crossed-product construction $\mathfrak{X}^n\rtimes(\mathfrak{T}^n\rtimes
\U_{\mathfrak{gl}(n)})$ with generating relations (\ref{igl})
supplemented by
\begin{gather*}
[L_{\nu}^{\mu }, x^{\lambda }]=-\imath\delta_{\nu }^{\lambda}x^{\mu } , \qquad
[P_\mu, x^\nu]=-\imath\delta_\mu^\nu , \qquad [x^\mu, x^\nu]=0.
\end{gather*}
The corresponding action is classical, i.e.\ it is implied by Heisenberg dif\/ferential representation (cf.\ formula~(\ref{Heisenberg2}) below):
\begin{gather}\label{action1}
 P_{\mu}\triangleright x^\nu=-\imath\delta_\mu^\nu, \qquad
   L^{\mu}_{\nu}\triangleright x^{\rho}=-\imath\delta_\nu^\rho x^\mu.
\end{gather}
Therefore, the Weyl algebra $\mathfrak{W}^n$ becomes a subalgebra in
$\mathfrak{X}^n\rtimes\left(\mathfrak{T}^n\rtimes\U_{\mathfrak{gl}(n)}\right)$.
Besides this isomorphic embedding one has a surjective algebra homomorphism
$\mathfrak{X}^n\rtimes\left(\mathfrak{T}^n\rtimes\U_{\mathfrak{gl}(n)}\right)\rightarrow\mathfrak{W}^n$
provided by
\begin{gather}\label{Heisenberg1a}
P_\mu \rightarrow P_\mu , \qquad
x^\mu\rightarrow x^\mu , \qquad L^\nu_\mu \rightarrow x^\nu P_\mu.
\end{gather}
We shall call this epimorphism Weyl algebra (or Heisenberg) realization of
$\mathfrak{X}^n\rtimes\left(\mathfrak{T}^n\rtimes\U_{\mathfrak{gl}(n)}\right)$.
Particularly, the map $L^\nu_\mu \rightarrow x^\nu P_\nu$ is a Lie algebra isomorphism.
The Heisenberg realization described above induces Heisenberg representation of $\mathfrak{X}^n\rtimes\left(\mathfrak{T}^n\rtimes\U_{\mathfrak{gl}(n)}\right)$
\begin{gather}\label{Heisenberg2} P_\mu\,\triangleright=-\imath\partial_\mu\equiv-\imath\frac{\partial}{\partial x^\mu}, \qquad
x^\mu=x^\mu ,\qquad L^\nu_\mu
\triangleright=-\imath x^\nu\partial_\mu
\end{gather}
acting in the vector space $\mathfrak{X}^n$. One can notice that (\ref{Heisenberg2}) can be extended to the vector space $C^\infty(\mathbb{R}^n)\otimes\mathbb{C}$ and f\/inally to the Hilbert space representation in ${\cal L}^2(\mathbb{R}^n, dx^n)$.
\end{example}
More examples one can f\/ind noticing that the Lie algebra $\mathfrak{igl}(n)$ contains several interesting subalgebras, e.g., inhomogeneous special linear transformations $\mathfrak{%
isl}(n)=\mathfrak{sl}(n)\niplus\mathfrak{t}^{n}$.

\begin{example}\label{ortho}
Even more interesting and important for physical applications example is provided by the inhomogeneous orthogonal transformations $\mathfrak{io}(g; n)=\mathfrak{o}%
(g; n)\niplus\mathfrak{t}^{n}\subset \mathfrak{isl}(n)$,  which Weyl extension is defined by the following set of commutation relations:
\begin{gather}  \label{isog1}
[M_{\mu \nu },M_{\rho \lambda }]= \imath g_{\mu \rho }M_{\nu
\lambda }-\imath g_{\nu \rho }M_{\mu \lambda }- \imath g_{\mu
\lambda }M_{\nu \rho }
+\imath g_{\nu \lambda }M_{\mu \rho },
\\  \label{isog2}
[M_{\mu \nu },P_{\lambda }]= \imath g_{\mu \lambda}P_{\nu }-\imath g_{\nu \lambda }P_{\mu }
, \\
\label{Weyl_g}
[M_{\mu \nu },x_{\lambda }]= \imath g_{\mu \lambda}x_{\nu }-\imath g_{\nu \lambda }x_{\mu },
\qquad
\left[ P_{\mu },x_{\nu }\right]= -\imath g_{\mu\nu}, \qquad
\left[ x_{\mu },x_{\nu }\right]=
\left[ P_{\mu },P_{\nu }\right]=0,
\end{gather}
where
\begin{gather}\label{isog3}
M_{\mu \nu }=g_{\mu \lambda }L_{\nu }^{\lambda }-g_{\nu \lambda }L_{\mu
}^{\lambda }.
\end{gather}
are defined by means of the (pseudo-Euclidean)\footnote{We shall write $\mathfrak{io}(n-p, p)$ whenever the signature $p$ of the metric will become important. In fact, dif\/ferent metric's signatures lead to dif\/ferent real forms of $\mathfrak{io}(n,\mathbb{C})$.} metric tensor $g_{\mu\nu}$ and $x_\lambda=g_{\lambda\nu}x^\nu$.
The last formula determines together with (\ref{action1}) the classical action of $\mathfrak{io}(n,\mathbb{C})$ on  $\mathfrak{X}^n$.
Thus relations (\ref{isog1})--(\ref{Weyl_g}) determine Weyl extension of the inhomogeneous orthogonal Lie algebra $\mathfrak{X}^n\rtimes(\mathfrak{T}^n\rtimes\U_{\mathfrak{o}(g; n)})$ as subalgebra
in $\mathfrak{X}^n\rtimes(\mathfrak{T}^n\rtimes\U_{\mathfrak{gl}(n)})$. In particular, the case of Poincar\'{e} Lie algebra $\mathfrak{io}(1,3)$ will be shown in more detail later on.
\end{example}

One special class of examples of the above crossed product algebras is the 
\textbf{Heisenberg double}. It takes place if $\mathcal{H}$ and $\mathcal{A}$
from above constitute a dual pair of Hopf algebras\footnote{%
Sometimes, for shortcuts, we denote $\mathcal{A}=\mathcal{H}^{\ast }$ or $%
\mathcal{H}=\mathcal{A}^{\ast }$ and call them dual Hopf algebras %
although it has strict sense only in the finite-dimensional
case.}, i.e. one has a nondegenerate bilinear Hopf
pairing between two Hopf algebras $<\cdot ,\cdot >:$ $\mathcal{H}\times 
\mathcal{A}\rightarrow\mathbb{C}$.
More precisely both Hopf algebra structures are mutually dual, in a sense:
\begin{eqnarray}
<L_{\left( 1\right) },f>< L_{\left( 2\right) }, f^{\prime }>=<L,f\cdot _{%
\mathcal{A}}f^{\prime }>;\\ <L,f_{\left( 1\right)}><L^{\prime
},f_{(2)}>=<f,L\cdot _{\mathcal{H}}L^{\prime }> ; \\
<L,S_{\mathcal{A}}\left( f\right) >=<S_{\mathcal{H}}\left( L\right)
,f>;\\ <1,->=\epsilon _{\mathcal{A}}; <-,1>=\epsilon _{\mathcal{H}}
\end{eqnarray}
where $L,L^{\prime }\in \mathcal{H}\left( m_{\mathcal{H}},\Delta _{\mathcal{H}},\epsilon _{\mathcal{H%
}},S_{\mathcal{H}}\right);\ f, f^{\prime }\in\mathcal{A}\left( m_{\mathcal{A}},\Delta _{%
\mathcal{A}},\epsilon _{\mathcal{A}},S_{\mathcal{A}}\right)$. 

The left regular action of (a quantum group) $\mathcal{H}$ on (covariant
quantum space) $\mathcal{A}=\mathcal{H}^{\ast }$ can be defined in the
following way:%
\begin{equation}  \label{left}
L\triangleright f=<L,f_{\left( 2\right) }>f_{(1)}
\end{equation}
where $\Delta _{\mathcal{A}}\left( f\right) =f_{(1)}\otimes $ $f_{\left(
2\right) };\forall L\in \mathcal{H};f,g,\in \mathcal{A}=\mathcal{H}^{\ast }$%
. In this case the corresponding crossed product algebra $\left( \mathcal{A}%
\rtimes \mathcal{H}\right) $ is called the Heisenberg double of the pair $%
\mathcal{A},\mathcal{H}$ and denoted by $\mathfrak{H}(\mathcal{A}
)$. The product in this algebra becomes:\newline
$(f\otimes L)\rtimes(g\otimes
M)=<L_{\left( 1\right) },g_{\left( 2\right) }>fg_{\left( 1\right) }\otimes
L_{\left( 2\right) }M$.\\ 
For the right regular action we have analogous expression to (\ref{left}) in the following form: $f\triangleleft L=<L,f_{(1)}>f_{\left( 2\right)}$. Thus
\begin{equation*}
<LL^{\prime },a>=<L,L^{\prime }\triangleright a>=<L^{\prime },a\triangleleft
L>=<1,L^{\prime }\triangleright a\triangleleft L>=\epsilon _{\mathcal{A}%
}(L^{\prime }\triangleright a\triangleleft L)
\end{equation*}%
where $L^{\prime }\triangleright a\triangleleft L=<L^{\prime
},a_{(3)}><a_{(1)},L>a_{(2)}$ with symbolic notation:
$\Delta ^{2}(a)=a_{(1)}\otimes a_{(2)}\otimes a_{(3)}$. The crossed product
algebra (Heisenberg double) $\mathfrak{H}(\mathcal{H})=\mathcal{H\ltimes A}$
is equipped with:
\begin{equation}
(L\otimes f)\ltimes (M\otimes g)=LM_{(1)}\otimes <M_{\left( 2\right)
},f_{(1)}>f_{\left( 2\right) }g
\end{equation}
in that case.
\begin{remark}
We are allowed to exchange the roles between Hopf algebra
representing symmetry group and those representing underlying space and
define the (left) action of $\mathcal{A}$ on $\mathcal{H=A}^{\ast }$ as: 
\begin{equation}
f\triangleright L=<L_{(2)},f>L_{\left( 1\right) }
\end{equation}
or similarly the (right) action of $\mathcal{A}$ on $\mathcal{H=A}^{\ast }$
: 
\begin{equation}
L\triangleleft f=<L_{(1)},f>L_{\left( 2\right) }
\end{equation}
\begin{proposition}
One has canonical isomorphism $$\mathfrak{H}(\mathcal{H})=\mathcal{H\ltimes A}%
\backsimeq \mathcal{H}\rtimes \mathcal{A}\qquad\mbox{and}\qquad 
\mathfrak{H}(\mathcal{A})=\mathcal{A\ltimes H}%
\backsimeq \mathcal{A}\rtimes \mathcal{H}$$ 
\end{proposition}
\begin{proof}
Initial algebras are canonically embedded, $\A\ni f\shortrightarrow 1\otimes f$ and $\H\ni L\shortrightarrow L\otimes 1$~as subalgebras in $\H\ltimes\A$ and $\A\ni f\shortrightarrow 1\otimes f$ and $\H\ni L\shortrightarrow L\otimes 1$~as subalgebras in $\H\rtimes\A$ respectively. Since algebras $\mathcal{H\ltimes A}$ and $\mathcal{H}\rtimes \mathcal{A}$ are generated by the same set of cross-commutation relations:
\begin{equation}
\left[ L\otimes 1,1\otimes g\right] _{\ltimes }= \left[ L\otimes 1,1\otimes g\right] _{\rtimes }=L\otimes g-L_{(1)}\otimes
<L_{\left( 2\right) },g_{(1)}>g_{\left( 2\right) }
\end{equation}
this proves that $\mathcal{H\ltimes A}\cong\mathcal{H}\rtimes \mathcal{A}$.
Analogously algebras $\mathcal{A}\rtimes \mathcal{H}$ and $\mathcal{A\ltimes H}$ have initial algebras $\H$ and $\A$ as subalgebras and are generated by the same cross commutation relations:
\begin{equation}
\left[ 1\otimes L,g\otimes 1\right] _{\rtimes }=\left[ 1\otimes L,g\otimes 1\right] _{\ltimes }=<L_{\left( 1\right)
},g_{\left( 2\right) }>g_{\left( 1\right) }\otimes L_{\left( 2\right)
}-g\otimes L
\end{equation}
\end{proof}
Moreover the algebras 
$\mathfrak{H}(\mathcal{A})$ and $\mathfrak{H}(\mathcal{H})$ are isomorphic too \cite{Hdouble}.


Therefore it seems that the separation between symmetry and space algebras becomes irrelevant in the Heisenberg double formulation.
\end{remark}
\begin{example}\label{exampleHDouble}
Heisenberg double construction requires a dual pair of Hopf algebras.
As an illustrative example let us consider
Heisenberg double of Poincare algebra $\mathcal{U}_{\mathfrak{iso}%
(1,3)}$ generated by $(M_{\mu \nu }, P_\nu)$ and $\mathcal{F(}\mathrm{ISO}(1,3))$ (functions on a group
$\mathrm{ISO}(1,3)$) generated by $(\Lambda_\mu^\nu, x^\nu)$.
Regarding the Poincare Lie algebra ${\mathfrak{iso}%
(1,3)}$ one can easily recognise semi-direct product construction between
the Lorentz $\mathfrak{so}(1,3)$ generators $M_{\mu \nu }$, which satisfy (\ref{isog1})
and the algebra of translations $\mathfrak{T}^{n}$ (introduced in example \ref{example4}).
Strictly speaking, on the level of universal enveloping algebras, the Poincare algebra is a smash product: $\mathcal{U}_{\mathfrak{iso}(1,3)}=\mathfrak{T}^{n}\rtimes
\mathcal{U}_{\mathfrak{so}(1,3)}$  which provides (\ref{isog2}) crossed commutation relations.
The algebra $\mathcal{U}_{\mathfrak{iso}(1,3)}$ can be equipped with a
primitive coalgebra structure:
\begin{equation}
\Delta \left( P_{\mu }\right) =\Delta _{0}\left( P_{\mu }\right)
\mbox{\
and\ }\Delta \left( M_{\mu \nu }\right) =\Delta _{0}\left( M_{\mu \nu
}\right)  \label{copPoincare}
\end{equation}%
With
\begin{equation}
\epsilon \left( P_{\mu }\right) =\epsilon \left( M_{\mu \nu }\right) =0%
\mbox{\ and\ }S\left( P_{\mu }\right) =-P_{\mu };S\left( M_{\mu \nu }\right)
=-M_{\mu \nu }  \label{counPoincare}
\end{equation}%
it becomes a (non-quantized) Hopf algebra.

Let us describe its dual counterpart $\mathcal{F(}\mathrm{ISO}(1,3))$  
as an algebra of functions on the Lie group $%
\mathrm{ISO}(1,3)$. One can notice that $\mathcal{F}(\mathrm{ISO}(1,3))$ is
another example of crossed product construction: $\mathcal{F}(\mathrm{ISO}%
(1,3))=\mathcal{F}\left( \mathfrak{X}^{4}\rtimes \mathrm{SO}(1,3)\right) $
between Abelian position algebra\\ $\mathfrak{X}^{4}=\{x^{\mu }: \left[ x^{\mu },x^{\nu }\right] =0\}$ and $\mathfrak{T}^{n}$. The Lorentz sector $$\mathcal{F}(\mathrm{SO}(1,3))=\{\Lambda _{\alpha}^{\beta }:\left[ \Lambda _{\alpha }^{\beta },\Lambda _{\mu }^{\nu }\right]
=0: \Lambda \eta \Lambda^{T} =\eta \}$$ constitutes
the dual pair with $\mathcal{U}_{\mathfrak{so}(1,3)}$ :
\be\label{dualMLam}
\qquad <M_{\alpha \beta },\Lambda _{\nu }^{\mu
}>=\imath \delta _{\alpha }^{\mu }g _{\beta \nu }-\imath \delta _{\beta}^{\mu }g _{\alpha \nu }\ee
The cross commutation relations 
becomes this time trivial: $\left[ x^{\mu },\Lambda _{\alpha }^{\beta }\right] =0$.
The (matrix) coalgebra structure 
\begin{equation}
\Delta \left( \Lambda _{\nu }^{\mu }\right) =\Lambda _{\gamma }^{\mu
}\otimes \Lambda _{\nu }^{\gamma };\mbox{\ \ \ }\Delta \left( x^{\mu
}\right) =\Lambda _{\gamma }^{\mu }\otimes x^{\gamma }+x^{\mu }\otimes 1
\label{copF}
\end{equation}%
together with counit: $\epsilon \left( \Lambda _{\gamma }^{\mu }\right) =\delta
_{\gamma }^{\mu };\epsilon \left( x^{\mu }\right) =0$ and antipode:\\ $%
S\left( \Lambda _{\gamma }^{\mu }\right) =\left( \Lambda ^{-1}\right)
_{\gamma }^{\mu }$; $S\left( x^{\gamma }\right) =-\left( \Lambda ^{-1}\right)
_{\mu }^{\gamma }x^{\mu }$ makes it a Hopf algebra.

The natural pairing (duality)\footnote{For commuting generators pairing extends naturally to the whole algebra.} induced by $<P_{\mu },x^{\nu}>=\imath \delta _{\mu }^{\nu }$ and (\ref{dualMLam})
provides the final necessary ingredient  in order to perform Heisenberg double construction 
\begin{equation*}
\mathfrak{H}(\mathcal{U}_{\mathfrak{iso}(1,3)})=\mathcal{U}%
_{\mathfrak{iso}(1,3)}\rtimes \mathcal{F(}\mathrm{ISO}(1,3)
\end{equation*}%
The left regular action of  $\mathcal{H=F(}\mathrm{ISO}(1,3)$
on $\mathcal{A}=\mathcal{H}^{\ast }=\mathcal{U}_{%
\mathfrak{iso}(1,3)}$ is determined by the duality 
and allows to obtain the usual Weyl algebra as
a subalgebra. The corresponding  smash product $\rtimes $ gives rise
to crossed commutation relations:
\begin{eqnarray}
\lbrack M_{\mu \nu },\Lambda _{\beta }^{\alpha }] &=&-<M_{\mu \nu },\Lambda
_{\beta }^{\delta }>\Lambda _{\delta }^{\alpha };\qquad \left[ P_{\mu
},x^{\nu }\right] =-\imath \delta _{\mu }^{\nu }; \\
\left[ P_{\nu },\Lambda _{\alpha }^{\beta }\right]  &=&0;\qquad \ [M_{\mu
\nu },x^{\alpha }]=-<M_{\mu \nu },\Lambda _{\beta }^{\alpha }>x^{\beta }
\end{eqnarray}%
which supplemented by (\ref{isog1},\ref{isog2}) and
\begin{equation}
\left[ P_{\mu },P_{\nu }\right] =0;\qquad \left[ x^{\mu },x^{\nu }\right]
=0;\qquad \left[ \Lambda _{\alpha }^{\beta },\Lambda _{\mu }^{\nu }\right]
=0;\qquad \left[ \Lambda _{\alpha }^{\beta },x^{\nu }\right] =0
\end{equation}%
describe the Heisenberg double algebra $\mathfrak{H}(\mathcal{U}_{\mathfrak{iso}(1,3)})$.

One can also consider smash product algebra : $\mathfrak{X}%
^{4}\rtimes \mathcal{U}_{\mathfrak{iso}(1,3)}$ determined by the 
classical action (\ref{action1}), provided that (\ref{isog3}). 
It turns out to be  a subalgebra of the  
Heisenberg double $\mathfrak{H}(\mathcal{U}_{\mathfrak{iso}(1,3)})$  
since it generates the same cross-commutation relations between $\left( P_{\mu }, 
 M_{\mu \nu }\right)$ and  $x^{\mu} $:
\begin{equation}
\left[ P_{\mu },x^{\nu }\right] =-\imath \delta _{\mu }^{\nu };\qquad [M_{\mu \nu },x^{\rho }]=\imath
\delta _{\beta }^{\rho }x_{\alpha }-\imath \delta _{\alpha }^{\rho }x_{\beta
}
\end{equation}
Particularly the Weyl algebra $(P_\mu, x^\nu)$ sectors (phase space) in both constructions are the same.\\
\textbf{Remark 1:} Interestingly, considering the equivalent Heisenberg
double $\mathfrak{H}(\mathcal{F(}\mathrm{ISO}(1,3))$, with exchanged roles between $%
\mathcal{H=U}_{\mathrm{iso}(1,3)}$ and $\mathcal{A}=\mathcal{F(}\mathrm{%
ISO}(1,3))$ we obtain slightly different cross commutation relations:
\begin{equation}
\lbrack M_{\mu \nu },\Lambda _{\beta }^{\alpha }]=<M_{\mu \nu },\Lambda
_{\beta }^{\delta }>\Lambda _{\delta }^{\alpha };\qquad \left[ P_{\mu
},x^{\nu }\right] =\imath \Lambda _{\mu }^{\nu };\qquad \lbrack M_{\mu \nu
},x^{\lambda }]=0
\end{equation}%
One can see that this Heisenberg double is isomorphic to the one introduced
before by changing  generators: $x^{\mu }\shortrightarrow \Lambda _{\alpha }^{\mu }x^{\alpha }.$
\end{example}
 In the case of $\kappa-$Poincar\'{e} algebra Heisenberg double was
originally studied in \cite{LukNow}, further developments and more systematic investigations on this subject can be found in \cite{ABAP5}.
\section{Twist-deformation}
Having the Hopf-module algebra structure one can use the deformation theory to construct new (deformed) objects which still possess the same structure.
Assume a Hopf module algebra $\A$ over $\mathcal{H}$.
These can be deformed, by a suitable twisting element $\mathcal{F}$,
to achieve the deformed Hopf module algebra $(\mathcal{A}_{\mathcal{F}},\mathcal{H}^{\mathcal{F}})$, where the algebra $\mathcal{A}_{\mathcal{F}}$ is equipped with a twisted (deformed) star-product (see~\cite{ABDMSW, Oeckl} and references therein)
\begin{gather}  \label{tsp}
x\star y=m \circ\mathcal{F}^{-1}\triangleright (x\otimes y)=(\bar{\mathrm{f}}^\alpha\triangleright x)\cdot(\bar{%
\mathrm{f}}_\alpha\triangleright y)
\end{gather}
while the classical action $\triangleright$ remains unchanged.
Hereafter the twisting element $\mathcal{F}$ is symbolically written in the
following form:
\begin{gather*}
  \mathcal{F}=\mathrm{f}^\alpha\otimes\mathrm{f}_\alpha\in \H\otimes\H \qquad \mbox{and}
\qquad \mathcal{F}^{-1}=\bar{\mathrm{f}}^\alpha\otimes\bar{\mathrm{f}}_\alpha\in \H\otimes\H
\end{gather*}
and belongs to the Hopf algebra $\H$. The corresponding smash product $\mathcal{A}_{\mathcal{F}}\rtimes\mathcal{H}^{\mathcal{F}}$
has deformed cross-commutation relations (\ref{smash}) determined by deformed coproduct $\Delta^{\mathcal{F}}$.
Before proceeding further let us remind in more detail that the quantized Hopf algebra $\H^{\mathcal{F}}$ has non-deformed algebraic sector (commutators),
while coproducts and antipodes are subject of the deformation:
\begin{gather*}\label{twcop}
\Delta^{\mathcal{F}} (\cdot)=\mathcal{F} \Delta (\cdot)\mathcal{F}^{-1}
,\qquad S^{\mathcal{F}}(\cdot)=u S(\cdot) u^{-1},
\end{gather*}
where $u=\mathrm{f}^\alpha S(\mathrm{f}_\alpha)$.
The twisting two-tensor $\mathcal{F}$ is an invertible element in $\H\otimes \H$ which fulf\/ills the 2-cocycle and normalization
conditions \cite{Drinfeld,Chiari}:
\begin{gather*}  
 \mathcal{F}_{12}(\Delta\otimes {\rm id})\mathcal{F}=\mathcal{F}
_{23}({\rm id}\otimes\Delta)\mathcal{F} ,\qquad   (\epsilon\otimes {\rm id})\mathcal{F}
=1=({\rm id}\otimes\epsilon)\mathcal{F},
\end{gather*}
which guarantee co-associativity of the deformed coproduct $\Delta^{\mathcal{F}}$ and associativity of the corresponding twisted star-product~(\ref{tsp}). Moreover, it implies simultaneously that deformed and undeformed smash product are isomorphic:

\begin{proposition}\label{prop1} For any Drinfeld twist $\mathcal{F}$  the twisted smash product
algebra  $\mathcal{A}_{\mathcal{F}}\rtimes\mathcal{H}^{\mathcal{F}}$ is  isomorphic to the initial (undeformed) one $\mathcal{A}\rtimes\mathcal{H}$. In other words the algebra $\mathcal{A}_{\mathcal{F}}\rtimes\mathcal{H}^{\mathcal{F}}$ is twist independent and can be realized by a change of generators in the algebra $\mathcal{A}\rtimes\mathcal{H}$. (cf. \cite{Hdouble2})
\end{proposition}

Notice that subalgebras $\mathcal{A}$ and $\mathcal{A}_{\mathcal{F}}$ are not isomorphic. In the case of commutative $\mathcal{A}$ we are used to call noncommutative algebra $\mathcal{A}_{\mathcal{F}}$ quantization of $\mathcal{A}$. Similarly $\mathcal{H}$ and $\mathcal{H}^{\mathcal{F}}$ are not isomorphic as Hopf algebras.

\begin{proof}\label{proof1}
First of all we notice that the inverse twist $\mathcal{F}^{-1}=\bar{\mathrm{f}}^\alpha\otimes\bar{\mathrm{f}}_\alpha$ satisf\/ies analogical
cocycle condition $(\Delta\otimes {\rm id})(\mathcal{F}^{-1})\mathcal{F}^{-1}_{12}=({\rm id}\otimes\Delta)(\mathcal{F}^{-1})\mathcal{F}_{23}^{-1}$.
It reads as
\[
\bar{\mathrm{f}}^\alpha_{(1)}\bar{\mathrm{f}}^\beta\otimes \bar{\mathrm{f}}^\alpha_{(2)}\bar{\mathrm{f}}_\beta\otimes\bar{\mathrm{f}}_\alpha=
\bar{\mathrm{f}}^\alpha\otimes\bar{\mathrm{f}}_{\alpha (1)} \bar{\mathrm{f}}^\beta\otimes\bar{\mathrm{f}}_{\alpha (2)}\bar{\mathrm{f}}_\beta.
\]
For any element $x\in \mathcal{A}$ one can associate the corresponding element
$\hat x=(\mathrm{\bar{f}}^\alpha\triangleright x)\cdot\bar{\mathrm{f}}_\alpha\in \mathcal{A}\rtimes\mathcal{H}$. Thus~(\ref{tsp}) together with the cocycle condition gives $\hat x\cdot\hat y=(\bar{\mathrm{f}}^\alpha\triangleright (x\star y))\cdot\bar{\mathrm{f}}_\alpha$.
Due to invertibility of twist we can express the original elements $x$ as a functions of the deformed one:
$x=(\mathrm{f}^\alpha\triangleright \hat x)\star\mathrm{f}_\alpha $.
It means that subalgebra generated by the elements $\hat x$ is isomorphic to $\mathcal{A}_{\mathcal{F}}$.
Notice that $\mathcal{A}=\mathcal{A}_{\mathcal{F}}$, $\mathcal{H}=\mathcal{H}^{\mathcal{F}}$
and $\mathcal{A}\rtimes\mathcal{H}=\mathcal{A}_{\mathcal{F}}\rtimes\mathcal{H}^{\mathcal{F}}$
as linear spaces.

The requested isomorphism $\varphi: \mathcal{A}_{\mathcal{F}}\rtimes\mathcal{H}^{\mathcal{F}}\rightarrow\mathcal{A}\rtimes\mathcal{H}$
can be now def\/ined by an invertible mapping
\begin{gather}\label{change}
\mathcal{A}_{\mathcal{F}}\ni x\rightarrow \varphi(x)=\hat x \in\mathcal{A}\rtimes\mathcal{H}\qquad
\hbox{and}\qquad \mathcal{H}^{\mathcal{F}}\ni L\rightarrow \varphi(L)=L\in \mathcal{A}\rtimes\mathcal{H}.
\end{gather}
Utilizing again the cocycle condition one checks that $\varphi(L\star x)=L\cdot\hat x=(L_{(1)}\mathrm{\bar{f}}^\alpha\triangleright x)\cdot L_{(2)}\bar{\mathrm{f}}_\alpha $.

Let $(x^\mu)$ be a set of generators for $\mathcal{A}$ and  $(L^k)$ be a set of generators for $\mathcal{H}$. Then the isomorphism (\ref{change}) can be described as a change of generators (``basis''):
$(x^\mu, L^k)\rightarrow (\hat x^\mu, L^k)$ in $\mathcal{A}\rtimes\mathcal{H}$.
\end{proof}
Dif\/ferent scheme for "unbraiding of braided tensor product" has been presented in~\cite{Fiore1}.

\begin{remark}\label{trivialTwist}
Consider internal automorphism of the algebra  $\mathcal{H}$ given by the similarity transformation $L\rightarrow WLW^{-1}$, where $W$ is some invertible element in $\mathcal{H}$.
This automorphism induces the corresponding isomorphism of Hopf algebras, which can be equivalently described as coalgebra isomorphism $(\mathcal{H}, \Delta)\rightarrow (\mathcal{H}, \Delta^{\mathcal{F}_W})$, where $\Delta^{\mathcal{F}_W}$ denotes twisted coproduct. Here $\mathcal{F}_W=W^{-1}\otimes W^{-1}\Delta(W)$ denotes the so called trivial (coboundary) twist. Of course, the twisted module algebra $\mathcal{A}_{\mathcal{F}_W}$ becomes isomorphic to the undeformed $\mathcal{A}$. Substituting\\
$W=\exp u$ one gets explicit form of the twisting element $$\mathcal{F}_W=\exp{(-u\otimes 1-1\otimes u)}\exp{\Delta(u)}=\exp{(-\Delta_0(u))}\exp{\Delta(u)}$$
\end{remark}
As it is well-known from the general framework \cite{Drinfeld}, a twisted deformation of Lie algebra $\mathfrak{g}$ requires a topological extension of the corresponding enveloping algebra $\U_{\mathfrak{g}}$ into an algebra of formal power series $%
\U_{\mathfrak{g}}[[h]]$ in the formal parameter $h$ (see Appendix A and, e.g.,~\cite{Klimyk,Kassel,Chiari,Bonneau} for deeper exposition)\footnote{This is mainly due to the fact that twisting element has to be invertible.}.
This provides the so-called $h$-adic topology. There is a correspondence between twisting element, which can be now  rewritten as a power series expansion
\[
\mathcal{F}=1\otimes 1+
\sum_{m=1}^\infty h^m\, \mathrm{f}^{(m)}\otimes\mathrm{f}_{(m)}\qquad\hbox{and}\qquad
\mathcal{F}^{-1}=1\otimes 1+
\sum_{m=1}^\infty h^m\, \bar{\mathrm{f}}^{(m)}\otimes\bar{\mathrm{f}}_{(m)},
\]
$\mathrm{f}^{(m)}, \mathrm{f}_{(m)}, \bar{\mathrm{f}}^{(m)}, \bar{\mathrm{f}}_{(m)}\in \U_{\mathfrak{g}}$, classical $r$-matrix $\mathfrak{r}\in \mathfrak{g}\wedge\mathfrak{g}$ satisfying  classical Yang--Baxter equation and universal (quantum) $r$-matrix $\mathcal{R}$:
\begin{gather*} 
\mathcal{R}=\mathcal{F}^{21}\mathcal{F}^{-1}=1+h \mathfrak{r}
\mod \big(h^2\big)
\end{gather*}
satisfying quantum Yang--Baxter equation. Moreover, classical $r$-matrices classify non-equivalent deformations.
Accordingly, the Hopf module algebra $\mathcal{A}$ has to be extended by $h$-adic topology to $\mathcal{A}[[h]]$ as well\footnote{The operator algebra setting for quantum groups and quantum spaces is admittedly much more heavy. Unfortunately, the
passage from quantized Lie algebra level to the function algebra level is not very straightforward and sometimes even not possible, see e.g.~\cite{Op}.}.

\begin{remark}\label{Twist}
In general, there is no constructive way to obtain twist for a given classical $r$-matrix. Few examples are known in the literature, e.g., Abelian~\cite{Reshetikhin}, Jordanian~\cite{Ogievetsky}  extended Jordanian twists \cite{Kulish} (see also \cite{Bonneau,VNT1}) as well as some of their combinations~\cite{Varna}. Twisted deformations of relativistic symmetries have been studied for a long time, see, e.g., \cite{BLT2} and references therein.
Almost complete classif\/ication of classical $r$-matrices for Lorentz an Poincar\'{e} group
has been given in~\cite{Zakrzewski} (see also~\cite{Lya}).
\end{remark}

\begin{example}\label{example6}
Let us consider the twist deformation of the algebra from Example~\ref{example5}. $h$-adic extension $\U_{\mathfrak{igl}(n)}\shortrightarrow\U_{\mathfrak{igl}(n)}[[h]]$ forces us to extend polynomial algebra:
$\mathfrak{X}^{n}\shortrightarrow\mathfrak{X}^{n}[[h]]$, which remains to be (undeformed) module algebra under the $\mathbb{C}[[h]]$-extended Hopf action $\triangleright$. Therefore their smash product contains $h$-adic extension of the Weyl algebra: $\mathfrak{W}^{n}[[h]]=\mathfrak{X}^{n}[[h]]\rtimes\mathfrak{T}^{n}[[h]]\subset \mathfrak{X}^{n}[[h]]\rtimes\left(\mathfrak{T}^{n}[[h]]\rtimes \U_{\mathfrak{gl}(n)}[[h]]\right)$.
After  having done $\mathfrak{igl}(n)$-twist $\mathcal{F}$ we can now deform simultaneously  both structures:  $U_{\mathfrak{igl}(n)}[[h]]\mapsto U_{\mathfrak{igl}(n)}[[h]]^{\mathcal{F}}$ and $\mathfrak{X}^{n}[[h]]\mapsto(\mathfrak{X}^{n}[[h]])_{\mathcal{F}}$ keeping the Hopf action $\triangleright$ unchanged. Deformed algebra $\mathfrak{X}^n[[h]]_\mathcal{F}$ has deformed star multipication $\star$ and can be represented by deformed $\star$-commutation relations
\begin{gather}\label{tcr}
  [x^{\mu },x^{\nu }]_{\star }\equiv x^{\mu }\star x^{\nu }-x^{\nu
}\star x^{\mu }=\imath\,h \theta^{\mu\nu}(x)\equiv
\imath h\big(\theta^{\mu\nu}+\theta^{\mu\nu}_\lambda x^\lambda+\theta^{\mu\nu}_{\lambda\rho} x^\lambda x^\rho+\cdots\big)
\end{gather}
replacing the undeformed (commutative) one $[x^{\mu },x^{\nu }]=0$,
where the coordinate functions $(x^\mu)$ play a role of generators for the
corresponding algebras:  deformed and undeformed one.
We will see on the examples below that many dif\/ferent twisted star products may lead to the same commutation relations~(\ref{tcr}). Particularly, one can def\/ine deformed Weyl algebra
$\mathfrak{W}^{n}[[h]]^\mathcal{F}=\mathfrak{X}^{n}[[h]]_\mathcal{F}
\rtimes\mathfrak{T}^{n}[[h]]^\mathcal{F}$, where $\mathfrak{T}^{n}[[h]]^\mathcal{F}$
denotes the corresponding Hopf subalgebra of deformed momenta in $\left(\mathfrak{T}^{n}[[h]]\rtimes \U_{\mathfrak{gl}(n)}[[h]]\right)^\mathcal{F}$.
\end{example}

\begin{remark}\label{Poisson}
The deformed algebra $\mathfrak{X}^{n}[[h]]_\mathcal{F}$ provides a deformation quantization of $\mathbb{R}^n$ equip\-ped with the Poisson structure (brackets)~\cite{Kon,BFFLS}
\begin{gather}\label{Pstr}
\{x^{\mu },x^{\nu }\}=
\theta^{\mu\nu}(x)\equiv
\theta^{\mu\nu}+\theta^{\mu\nu}_\lambda x^\lambda+\theta^{\mu\nu}_{\lambda\rho}x^\lambda x^\rho+\cdots,
\end{gather}
represented by Poisson bivector $\theta=\theta^{\mu\nu}(x)\partial_\mu\wedge\partial_\nu$. It is assumed that $\theta^{\mu\nu}(x)$ are polynomial functions, i.e.\ the sum in (\ref{Pstr}) is f\/inite where $\theta^{\mu\nu},\theta^{\mu\nu}_\lambda,\theta^{\mu\nu}_{\lambda\rho},\ldots$ are real numbers.
\end{remark}
\begin{proposition}\label{prop2}
Proposition {\rm \ref{prop1}} implies that $ \mathfrak{X}^{n}[[h]]\rtimes \U_{\mathfrak{igl}(n)}[[h]]$ is $\mathbb{C}[[h]]$ isomorphic to $ \mathfrak{X}^{n}[[h]]_\mathcal{F}\rtimes( \U_{\mathfrak{igl}(n)}[[h]])^\mathcal{F}$. Moreover, this isomorphism is congruent to the identity map mo\-dulo~$h$.
\end{proposition}

Replacing elements $\overline{\mathrm{f}}^{(m)}$,
$\overline{\mathrm{f}}_{(m)}\in\U_{\mathfrak{igl}(n)}$
in the formulae
\begin{gather}\label{Heis3}
\hat{x}^\mu=x^\mu+\sum_{m=1}^\infty h^m  \big(\overline{\mathrm{f}}^{(m)}
\triangleright x^\mu\big)\cdot\overline{\mathrm{f}}_{(m)}
\end{gather}
from Proposition~\ref{prop1}
by using Heisenberg realizations (\ref{Heisenberg1a}) one gets
\begin{proposition}\label{prop3}
All $\mathfrak{igl}(n)$-twist deformed Weyl algebras $\mathfrak{W}^{n}[[h]]^\mathcal{F}$
are $\mathbb{C}[[h]]$-isomorphic to undeformed $h$-adic extended Weyl algebra $\mathfrak{W}^{n}[[h]]$.
In this sense we can say that $\mathfrak{W}^{n}[[h]]^\mathcal{F}$
is a~pseudo-deformation of  $\mathfrak{W}^{n}[[h]]$ since the latter one
can be obtained by (nonlinear and invertible) change of generators from the first one\footnote{Cf.~\cite{Fiore2} for dif\/ferent approach to deformation of Clif\/ford and Weyl algebras.}.
\end{proposition}

\begin{remark}\label{h-repr}
Replacing again elements $\overline{\mathrm{f}}^{(m)}, \overline{\mathrm{f}}_{(m)}\in\U_{\mathfrak{igl}(n)}$ in (\ref{Heis3}) by their Heisenberg
representation: $P_\mu=-\imath\partial_\mu$ (cf.\ Example \ref{example5}) one obtains Heisenberg representation of $\mathfrak{W}^{n}[[h]]^\mathcal{F}$ and entire algebra $ \mathfrak{X}^{n}[[h]]_\mathcal{F}\rtimes( \U_{\mathfrak{igl}(n)}[[h]])^\mathcal{F}$ in the vector space $\mathfrak{X}^{n}[[h]]$.
Moreover, one can extend  Hilbert space representation from Remark \ref{Hspace} to the representation acting in ${\cal L}^2(\mathbb{R}^n, dx^n)[[h]]$. This is due to general theory of representations for $h$-adic quantum groups, see, e.g., \cite{Chiari}.
\end{remark}

\begin{remark}
Moreover, the twisted star product enables us to introduce two operator realizations of the algebra $\mathfrak{X}^n[[h]]_\mathcal{F}$ (or in general $\A_\F$) in terms of (formal) differential operators (Heisenberg differential representation) of $\U_{\mathfrak{igl}(n)}[[h]]$ ($\H^\F$). The so-called left-handed and right-handed realizations are
naturally defined by
\begin{equation}  \label{lhr-rhr1}
\hat{x}^{\mu }_L (f)= x^{\mu }\star f \qquad \mbox{and} \qquad \hat x^{\mu
}_R (f)= f\star x^{\mu }
\end{equation}
with $\hat{x}^{\mu }_L $, $\hat{x}^{\mu }_R$ $\in U_{\mathfrak{igl}(n)}[[h]]$ ($\H^\F$)
satisfying the operator commutation relations
\begin{equation}  \label{lhr-rhl2}
[\hat{x}^{\mu }_L, \hat{x}^{\nu }_L ]=\imath h\theta^{\mu\nu}(\hat{x})\qquad %
\mbox{and} \qquad [\hat{x}^{\mu }_R, \hat{x}^{\nu }_R ]=-\imath h\theta^{\mu\nu}(%
\hat{x})
\end{equation}
correspondingly. In other words, the above formulae describe embedding of $\mathfrak{X}^n[[h]]_\mathcal{F}$ ($\A_\F$) into $U_{\mathfrak{igl}(n)}[[h]]$ ($\H^\F$). These operator
realizations are particular cases of the so-called Weyl map (see, e.g., \cite{Weyl} and references therein for more details). They allow us to calculate the commutator:
\begin{equation}  \label{cc}
[x^\mu, f]_\star=\hat{x}^{\mu }_L (f)-\hat{x}^{\mu }_R(f)
\end{equation}
It has been argued in Ref. \cite{JSSW} (see also \cite{BMS}) that any Lie-algebraic star product (when generators satisfy the Lie algebra structure)
\begin{equation}  \label{c1}
\ [x^{\mu },x^{\nu }]_{\star }=\imath h\theta_{\lambda }^{\mu \nu }x^{\lambda }
\end{equation}
can be obtained by twisting an element in the form
\begin{equation}  \label{c2}
\ \mathcal{F}= exp(\frac{\imath}{2}x^\lambda g_\lambda(\imath\frac{\partial}{%
\partial y},\imath\frac{\partial}{\partial z}))\ \mid_{y\shortrightarrow x;\
z\shortrightarrow x}
\end{equation}
The star product
\begin{equation}  \label{c3}
\ f(x)\star g(x)=exp(\frac{\imath}{2}x^\lambda g_\lambda(\imath\frac{\partial%
}{\partial y},\imath\frac{\partial}{\partial z}))f(y)g(z) \mid_{y%
\shortrightarrow x; z\shortrightarrow x }
\end{equation}
would imply
\begin{equation}  \label{c4}
[x^\mu, f]_\star=\imath h\theta_{\lambda }^{\mu \nu }x^{\lambda }\partial_\nu f
\end{equation}
i.e., a vector field action on $f$. The last formula is particularly
important for obtaining a Seiberg-Witten map for noncommutative gauge
theories (see \cite{JSSW,BMS}). Later on it will be shown on explicit
examples that this formula is not satisfied for an arbitrary twist in
the form of (\ref{c2}). However, we shall find an explicit twist for
$\kappa-$deformed Minkowski spacetime which belongs to the class
described by (\ref{c4}).
\end{remark}
\begin{remark}
Before proceeding further, let us comment on some important
differences between the canonical and Lie-algebraic cases. In the
former, related to the Moyal product case
\begin{equation}  \label{rc1}
[x^\mu, x^\nu]_{\star, M}=\imath\theta^{\mu \nu }
\end{equation}
with a constant antisymmetric matrix $\theta^{\mu \nu }$; cf.
(\ref{tcr}), one finds that
\begin{equation}  \label{rc2}
[g, f]_{\star, M}=\partial_\nu\zeta^{\nu }(f,g)
\end{equation}
is a total divergence, e.g.,
\begin{equation}  \label{rc3}
[x^\mu, f]_{\star, M}=\partial_\nu\left(\imath\theta^{\mu \nu}f
\right)
\end{equation}
This further implies the following tracial property of the integral:
\begin{equation}  \label{rc4}
\int d^n\,x\, [g, f]_{\star, M}= 0
\end{equation}
which is rather crucial for  a variational derivation of Yang-Mills
field equations (see \cite{ADMSW}). In contrast to (\ref{rc3}),
eq. (\ref{c4}) rewritten under the form
\begin{equation}  \label{rc5}
[x^\mu, f]_\star=  \partial_\nu\left(\imath\theta^{\mu \nu}_\lambda
x^\lambda f \right)- \imath\theta_{\nu }^{\mu \nu } f
\end{equation}
indicate  obstructions to the tracial property (\ref{rc4}) provided
that $\theta_{\nu }^{\mu \nu }\neq 0$  
\footnote{For  the
$\kappa$ -deformation [see(\ref{c5}) and  (\ref{kM}) below], one has
$\theta_{\nu }^{\mu \nu }=(n-1)a$.}.
\end{remark}
In the next chapter we shall present explicitly  two one-parameter families of twists
corresponding to twisted star-product realization  of the $\kappa $-deformed Minkowski spacetime algebra \cite{MR,Z}.

\chapter{$\kappa$-Minkowski spacetime from Drinfeld twist}
The first idea of non commuting coordinates was suggested as long ago as the 1940 by Snyder \cite{Snyder}. More recently, deformed coordinate spaces based on algebraic relations $[x^\mu,x^\nu]=i\theta^{\mu\nu}$ (with $\theta^{\mu\nu}$ -constant) have been considered as
providing uncertainty relations which restrict the accuracy up to which the coordinates of an event can be measured \cite{Dop1}-\cite{Dop2}.  
These presently known as "canonical" spacetime commutation relations were and still are the subject of many investigations in quantum field and gauge theories and even in string theory.

The next step in difficulty is to consider the Lie-algebraic type of noncommutativity, i.e.: $[x^\mu,x^\nu]=\imath\theta^{\mu\nu}_\lambda x^\lambda$. The $\kappa$-Minkowski noncommutative spacetime is an example of this type of noncommutativity. Usually known in the form: $[x^0,x^k]=\frac{\imath}{\kappa} x^k;\qquad [x^i,x^k]=0$, however more general: $[x^\mu,x^\nu]=\imath (a^{\mu}x^\nu-a^\nu x^\mu)$ (where $a^{\mu}$ is fixed four-vector) has also been used by many authors. It was firstly proposed in \cite{LLM}, and investigated in three cases (with $\eta_{\mu\nu}=(+,-,-,-)$):\\
i) $a_{\mu}a^{\mu}=1$ for time-like  $\kappa$-deformations;\\
ii) $a_{\mu}a^{\mu}=-1$ for tachyonic one;\\
iii) $a_{\mu}a^{\mu}=0$ providing light-cone $\kappa$-deformation.\\
Such algebras described by the above mentioned commutation relations constitute Hopf module algebras with symmetry group described in Quantum Group formalism. We will focus here on the Lie-algebraic quantum deformations (\ref{c1}) which are the most physically appealing from a bigger set of quantum deformations (\ref{tcr}) in the Hopf-algebraic framework of
quantum groups  \cite{Drinfeld}, \cite{Jimbo}-\cite{Fadeev}. It appears that quantum deformations of Lie algebras are controlled by classical $r$ -matrices satisfying the classical
Yang-Baxter (YB) equation: homogeneous or inhomogeneous.
Particularly, an effective tool is provided by the so-called
twisted deformations \cite{Drinfeld} which classical $r$ -matrices satisfy the
homogeneous YB equation and can be applied to both Hopf algebra
(coproduct) as well as related Hopf module algebra ($*$ -product) \cite{Oeckl}.
Two types of explicit examples of twisting two tensors are the best known and investigated in
the literature: Abelian \cite{Reshetikhin} and Jordanian \cite{Ogievetsky} as
well as the extended Jordanian one \cite{Kulish} (see also \cite{Bonneau,VNT1}). The $\kappa$ -deformation
of Poincar\'{e} algebra is characterized by the inhomogeneous classical YB equation, which
implies that one should not expect to get $\kappa-$Minkowski space from a
Poincar\'{e} twist. However, twists belonging to extensions of the Poincar\'{e} algebra are not
excluded \cite{Bu,MSSKG}. An interesting result has been found in Ref. 
\cite{Ballesteros}, 
where starting from nonstandard (Jordanian) deformed $D=4$ conformal
algebra the $\kappa-$Minkowski space has been obtained in two ways: by
applying the Fadeev-Reshetikin-Takhtajan technique \cite{Fadeev} or exploiting a bialgebra structure
generated by a classical $\mathfrak{so}(4,2)$ $r$ -matrix.
Moreover the classification of quantum
deformations strongly relies on classification of the classical $r$ -%
matrices. In a case of relativistic Lorentz and Poincar\'{e} symmetries, such classification 
has been performed some time ago in Ref. \cite{Zakrzewski}. A passage from the
classical $r$ -matrix to twisting two-tensor and the corresponding Hopf-algebraic deformation is a nontrivial task. Explicit twists for
Zakrzewski's list have been provided in Ref. \cite{Varna} as well as their
superization \cite{VNT1,VNT2}. 
Systematic quantizations of the corresponding Lorentz and
Poincar\'{e} Hopf algebras are carried on in Ref. \cite{BLT2}. In particular, the
noncommutative spacetimes described by three types of Abelian Poincar\'{e}
twists have been calculated in Ref. \cite{LW}. 

Below we shall present explicitly  two one-parameter families of $\mathfrak{igl}(n)$-twists
corresponding to twisted star-product realization  of the well-known $\kappa $-deformed Minkowski spacetime algebra \cite{MR,Z}. These are Abelian and Jordanian families. The Abelian
family has been previously investigated \cite{Bu}. One finds out, in an
explicit form, the operator realization of the $\kappa$-Minkowski coordinates $%
\hat x^\mu$ for each value of twist parameters. In this section we also
point out the smallest possible subalgebras, containing Poincar\'{e} algebra, to
which one can reduce the deformation procedure.

Algebra of $\kappa$-Minkowski coordinates (cf. example \ref{example6} and cf. (\ref{c1})) $(\mathfrak{X}^{n}[[h]])^{\mathcal{F}}$ is generated by the relations:
\begin{gather}\label{kM1}
[x^0,x^k]_\star=\imath h x^k,\qquad [x^k,x^j]_\star=0  ,\qquad k,j=1,\ldots ,n-1,
\end{gather}
and constitutes a covariant algebra over deformed $\mathfrak{igl}(n)$.\footnote{Notice that $h$ is the formal parameter. So (\ref{kM1}) is not, strictly speaking, a Lie algebra.} Although the result of star multiplication of two generators  $x^\mu\star x^\nu$ is explicitly twist dependent, the generating relations (\ref{kM1}) are twist independent. Therefore all algebras
$(\mathfrak{X}^{n}[[h]])^{\mathcal{F}}$ are mutually isomorphic to each other. They provide a covariant deformation quantization of the $\kappa$-Minkowski Poisson structure represented by the linear Poisson bivector
\begin{gather}\label{kM2}
    \theta_{\kappa M}=x^k\partial_k\wedge \partial_0,\qquad k=1,\ldots, n-1
\end{gather}
on $\mathbb{R}^{n}$ (cf.~\cite{Z}). The corresponding Poisson tensor $\theta^{\mu\nu}(x)=a^\mu x^\nu - a^\nu x^\mu$, where $a^\mu=(1,0,0,0)$ is degenerated ($\det[\theta^{\mu\nu}(x)]=0$), therefore, the associated symplectic form $[\theta^{\mu\nu}(x)]^{-1}$ does not exist.
\begin{remark}\label{Poisson2}
More generally, there is one-to-one correspondence between linear Poisson structures
$\theta=\theta^{\mu\nu}_\lambda x^\lambda\partial_\mu\wedge\partial_\nu$ on $\mathbb{R}^{n}$ and $n$-dimensional Lie algebras $\mathfrak{g}\equiv\mathfrak{g}(\theta)$
\begin{gather}\label{envelop0}
    [X^\mu, X^\nu]=\imath \theta^{\mu\nu}_\lambda X^\lambda
\end{gather}
with the constants $\theta^{\mu\nu}_\lambda$ playing the role of Lie algebra structure constants.
Therefore, they are also called  Lie--Poisson structures. Following Kontsevich we shall call the corresponding universal enveloping algebra $\U_{\mathfrak{g}}$ a canonical quantization of $(\mathfrak{g}^*,  \theta)$~\cite{Kon}. The classif\/ication of  Lie--Poisson structures on
$\mathbb{R}^{4}$ has been recently presented  in~\cite{Sheng}.

Let us consider the following modif\/ication of the universal enveloping algebra construction~(\ref{UEnvelop})
\begin{gather*}
    \U_h(\mathfrak{g})=\frac{T\mathfrak{g}[[h]]}{J_{h}},
\end{gather*}
where $J_{h}$ denotes an ideal generated by elements
$\langle X\otimes Y-Y\otimes X-h[X, Y]\rangle $ and is closed in $h$-adic topology. In other words the algebra $\U_h(\mathfrak{g})$ is $h$-adic unital algebra generated by $h$-shifted relations
\begin{gather*}
    [\X^\mu, \X^\nu]=\imath h\theta^{\mu\nu}_\lambda \X^\lambda
\end{gather*}
imitating  the Lie algebraic ones~(\ref{envelop0}).
This algebra provides the so-called universal quantization of $(\mathfrak{X}^{n}, \theta)$~\cite{Kat}.
Moreover, due to universal (quotient) construction there is a $\mathbb{C}[[h]]$-algebra epimorphism from $\U_h(\mathfrak{g})$ onto $(\mathfrak{X}^{n}[[h]])^{\mathcal{F}}$ for a suitable twist $\mathcal{F}$ (cf.~(\ref{tcr})). In fact, $\U_h(\mathfrak{g})$ can be identif\/ied with $\mathbb{C}[[h]]$-subalgebra in $\U_{\mathfrak{g}}[[h]]$ generated by $h$-shifted generators $\X^\mu=hX^\mu$.  $\U_{\mathfrak{g}}[[h]]$~is by construction a topological free $\mathbb{C}[[h]]$-module (cf.~Appendix \ref{app3}). For the case of  $\U_h(\mathfrak{g})$ this question is open.
\end{remark}

\begin{example}\label{example8}
In the the case of $\theta_{\kappa M}$ (\ref{kM2}) the corresponding Lie algebra is a solvable one. Following~\cite{Oriti} we shall denote it as $\mathfrak{an}^{n-1}$. The universal $\kappa$-Minkowski  spacetime algebra $\U_h(\mathfrak{an}^{n-1})$ has been introduced in~\cite{Z}, while its Lie algebraic counterpart the canonical $\kappa$-Minkowski  spacetime algebra $\U_{\mathfrak{an}^{n-1}}$ in~\cite{MR}. We shall consider both algebras in more details later on.
\end{example}

\subsection*{Abelian family of twists providing $\boldsymbol{\kappa}$-Minkowski spacetime}

The simplest possible twist is Abelian one.
$\kappa $-Minkowski spacetime can be  implemented by the one-parameter
family of Abelian twists \cite{BP} with $s$ being a numerical
parameter labeling dif\/ferent twisting tensors (see also~\cite{Bu, MSSKG}):
\begin{gather}  \label{At}
  \mathfrak{A}_{s}=\exp \left[ \imath h \left(sP_{0}\otimes
D-\left( 1-s\right) D\otimes
P_{0}\right)\right],
\end{gather}
where $D=\sum\limits_{\mu=0}^{n-1}L^\mu_\mu-L^0_0$.
All twists correspond to the same classical $r$-matrix\footnote{Since in the Heisenberg realization the space dilatation generator $D=x^k\partial_k$, $\mathfrak{r}$ coincides with the Poisson bivector~(\ref{kM2}).}:
\begin{gather*}
\mathfrak{r}=D\wedge P_0
\end{gather*}
and they have the same universal quantum r-matrix which is  of exponential form:
\begin{gather*}
\mathcal{R}=\mathfrak{A}_s^{21}\mathfrak{A}_s^{-1}=e^{\imath D\wedge P_0}.
\end{gather*}

\begin{remark}
This implies that corresponding Hopf algebra deformations of $\U_{\mathfrak{igl}(n)}^\mathcal{F}$ for dif\/ferent values of parameter $s$ are isomorphic. Indeed, $\mathfrak{A}_{s}$ are related by trivial twist:  $\mathfrak{A}_{s_2}=\mathfrak{A}_{s_1}{\mathcal F}_{W_{12}}$, where $W_{12}=\exp{(\imath(s_1-s_2) a DP_0)}$ (cf.\ Remark~(\ref{trivialTwist})). Consequently $ \Delta^{\rm op}_{(s=0)}=\Delta_{(s=1)}$ and
$\Delta^{\rm op}=\mathcal{R}\Delta\mathcal{R}^{-1}$.
We will see later on that dif\/ferent values of $s$ lead to dif\/ferent Heisenberg realizations and describe dif\/ferent physical models (cf.~\cite{BP}).
\end{remark}

\begin{remark}
The smallest subalgebra generated by $D$, $P_0$ and the Lorentz generators (\ref{isog3}) turns out to be entire $\mathfrak{igl}(n)$ algebra.
Therefore, as a consequence,  the deformation induced by twists (\ref{At}) cannot be restricted to
the inhomogeneous orthogonal transformations (\ref{isog1})--(\ref{isog2}). One can say that Abelian twists~(\ref{At}) are genuine $\mathfrak{igl}(n)$-twists.
\end{remark}

The deformed coproducts read as follows (cf.~\cite{Bu,BP}):
\begin{gather*}
\Delta _{s}\left( P_{0}\right) =1\otimes P_{0}+P_{0}\otimes 1,
\qquad
\Delta _{s}\left( P_{k}\right) =e^{-hsP_{0}}\otimes P_{k}+P_{k}\otimes
e^{h(1-s)P_{0}},
\\
\Delta _{s}\left( L_{k}^{m}\right) =1\otimes L_{k}^{m}+L_{k}^{m}\otimes 1,
\qquad
\Delta _{s}\left( L_{0}^{k}\right) =e^{hsP_{0}}\otimes
L_{0}^{k}+L_{0}^{k}\otimes e^{-h(1-s)P_{0}},
\\
\Delta _{s}\left( L_{k}^{0}\right) =e^{-hsP_{0}}\otimes
L_{k}^{0}+L_{k}^{0}\otimes e^{h(1-s)P_{0}}+hsP_{k}\otimes De^{h(1-s)P_{0}}
-h(1-s)D\otimes P_{k},
\\
\Delta _{s}\left( L_{0}^{0}\right) =1\otimes L_{0}^{0}+L_{0}^{0}\otimes 1
+hsP_{0}\otimes D-h(1-s)D\otimes P_{0},
\end{gather*}
and the antipodes are:
\begin{gather*}
S_s\left( P_{0}\right) =-P_{0}, \qquad S_s\left( P_{k }\right) =-P_{k
}e^{hP_0},
\\
S_s\left( L_{k}^{m}\right) =-L_{k}^{m}, \qquad S_s\left( L_{0}^{k}\right)
=-L_{0}^{k} e^{-hP_0},
\qquad
S_s\left( L_{0}^{0}\right) =-L_{0}^{0}-h(1-2s)DP_0,
\\
S_s\left( L_{k}^{0}\right) =-e^{hsP_0}L_{k}^{0}e^{-
h(1-s)P_0}+h\big[sP_kDe^{hsP_0}+(1-s)DP_k e^{h(1+s)P_0}\big].
\end{gather*}
The above relations are particularly simple for $s=0, 1, {\frac{1}{2}}$.
Smash product construction (for given $\Delta_s$) together with classical action~(\ref{action1}) leads to the following crossed commutators:
\begin{gather*}
\left[ \hat{x}^{\mu },P_{0}\right]_s =\imath \delta^\mu_0,\qquad\left[ \hat{x}^\mu,P_{k}\right]_s =\imath \delta^\mu_ke^{h(1-s)P_{0}}-\imath hs\delta^\mu_0P_{k},
\\
\left[ \hat{x}^\mu,L_{k}^{m}\right]_s =\imath \delta^\mu_k\hat{x}^{m},\qquad \big[ \hat{x}^\mu,L_{0}^{k}\big]_s =\imath \delta^\mu_0\hat{x}^{k}e^{-h(1-s)P_{0}}-\imath hs\delta^\mu_0L_{0}^{k},
\\
\left[ \hat{x}^{\mu }, L_{k}^{0}\right]_s =\imath
\delta_k^\mu\hat{x}^{0} e^{h(1-s)P_{0}}-\imath hs\delta^\mu_0L_{k}^{0}
+\imath hs\delta_k^\mu D-\imath h(1-s)\delta^\mu_l \hat{x}^{l}P_{k},
\\
\left[\hat{x}^{\mu },L_{0}^{0}\right]_s =\imath \delta_0^\mu\hat{x}^{0}
+\imath hs\delta_{0}^{\mu }D-\imath h(1-s)\delta_k^\mu \hat{x}^kP_{0}.
\end{gather*}
Supplemented with (\ref{igl}) and $\kappa$-Minkowski spacetime commutators:
\begin{gather}\label{kMhat}
  [\hat{x}^0,\hat{x}^k]=\imath h \hat{x}^k, \qquad [\hat{x}^k,\hat{x}^j]=0  ,\qquad k,j=1,\ldots, n-1
\end{gather} they form the algebra:
$\mathfrak{X}^n[[h]]^\mathcal{F}\rtimes \U_{\mathfrak{igl}(n)}^\mathcal{F}$.
 The change of generators $(L_\mu^\nu,P_\rho,x^\lambda)\shortrightarrow(L_\mu^\nu,P_\rho,\hat{x}^\lambda_s)$, where
 \begin{gather*}  
\hat x^{i}_{s}=x^{i}e^{\left( 1-s\right)hP_0} , \qquad \hat x^{0}_{s} =
x^{0}- h s D
\end{gather*}
 implies the isomorphism from Proposition~\ref{prop2}. Similarly, the change of generators $(P_\rho,x^\lambda)\shortrightarrow(P_\rho,\hat{x}^\lambda_s)$, where
  \begin{gather}  \label{Atlh2}
\hat x^{i}_{s}=x^{i}e^{\left( 1-s\right)hP_0} , \qquad \hat x^{0}_{s} =
x^{0}- hsx^{k}P_{k}
\end{gather}%
illustrates Proposition~\ref{prop3} and gives rise to Heisenberg representation acting in the vector space $\mathfrak{X}^n[[h]]$ as well as its Hilbert space extension acting in ${\cal L}^2(\mathbb{R}^n, dx^n)[[h]]$ provided that $P_k=-\imath\partial_k$.

\begin{remark}
The coordinate's realization shown above (\ref{Atlh2}) is the left-handed representation (Hermitian for $s=0$) in the sense introduced in Chapter 1 (cf. \ref{lhr-rhr1}). 

The right-handed representation for Abelian twist (Hermitian for $\ s=1$) is the following:
\begin{equation}  \label{Atrh}
\hat x^{i}_{R,s}=x^{i}e^{-shA} , \qquad \hat x^{0}_{R,s}=x^{0}+h\left(
1-s\right) x^{k}P _{k}
\end{equation}
This implies
\begin{equation}
[x^i, f]_{\star,s}= x^{i}\left(e^{(1-s)hP_0}-e^{-shP_0}\right)f ,\qquad [x^0,
f]_{\star,s}= \imath h x^k\partial_k f
\end{equation}
which is different from (\ref{c1}) for any value of the parameter $s$.
\end{remark}

\subsection*{Jordanian family of twists providing $\boldsymbol{\kappa}$-Minkowski spacetime}

Jordanian twists have the following form \cite{BP}:
\begin{gather*}  
 \mathfrak{J}_{r }=\exp \left(J_{r}\otimes \sigma_r \right),
\end{gather*}
where $J_{r}=\imath(\frac{1}{r}D-L^0_0)$ with a numerical factor $r\neq 0$ labeling dif\/ferent twists and $\sigma_r =\ln (1-h rP_0)$.
Jordanian twists are related with
Borel subalgebras $\mathfrak{b}^{2}(r)=\{J_r,P_{0}\}\subset \mathfrak{igl}(n,\mathbb{R})$
which, as a matter of fact, are isomorphic to the 2-dimensional solvable Lie algebra $\mathfrak{an}^1$: 
\begin{equation}\label{borel}
[J_r, P_{0}]=P_0
\end{equation}
Direct calculations show that, regardless of the value of $r$, twisted
commutation relations (\ref{tcr}) take the form of that for $\kappa$-Minkowski spacetime
(\ref{kM1}).
The corresponding classical $r$-matrices are the following\footnote{Now, for dif\/ferent values of the parameter $r$ classical $r$-matrices are not the same.}:
\begin{gather}\label{jt2}
\mathfrak{r}_J=\mathfrak{J}_{r}\wedge P_0={1\over r}D\wedge P_0-L^0_0\wedge P_0.
\end{gather}
Since in the generic case we are dealing with $\mathfrak{igl}(n)$-twists (see~\cite{BP} for exceptions), we shall write deformed coproducts and antipodes in terms of generators $\{L_{\nu }^{\mu },P_{\mu }\}$. The deformed coproducts read as follows~\cite{BP}:
\begin{gather*}
\Delta _{r}\left( P_{0}\right) =1\otimes P_{0}+P_{0}\otimes e^{\sigma _{r}},\qquad
 \Delta _{r}\left( P_{k}\right) =1\otimes P_{k}+P_{k}\otimes e^{-\frac{1}{r}\sigma _{r}},
\\
\Delta _{r}\left( L_{k}^{m}\right) =1\otimes L_{k}^{m}+L_{k}^{m}\otimes 1,\qquad
\Delta _{r}\big( L_{0}^{k}\big) =1\otimes
L_{0}^{k}+L_{0}^{k}\otimes e^{\frac{r+1}{r}\sigma _{r}},
\\
 \Delta _{r}\left( L_{k}^{0}\right) =1\otimes L_{k}^{0}+L_{k}^{0}\otimes e^{-%
\frac{r+1}{r}\sigma _{r}}- hrJ_{r}\otimes P_{k}e^{-\sigma _{r}},\\
\Delta _{r}\left( L_{0}^{0}\right) =1\otimes L_{0}^{0}+L_{0}^{0}\otimes 1
-hrJ_{r}\otimes P_{0}e^{-\sigma _{r}},
 \end{gather*}
where
\begin{gather*}
e^{\beta\sigma_r}=\left(1-hrP_0\right)^\beta=\sum_{m=0}^\infty \frac{%
h^m}{m!}\beta^{\underline m}(- rP_0)^m
\end{gather*}
and $\beta^{\underline m}=\beta(\beta-1)\cdots (\beta-m+1)$ denotes the
so-called falling factorial. The antipodes are:
\begin{gather*}
S_r\left( P_{0}\right) =-P_{0}e^{-\sigma_r }, \qquad S_r\left( P_{k }\right)
=-P_{k }e^{\frac{1}{r}\sigma_r },
\\
S_r\big( L_{0}^{k}\big) =-L_{0}^{k}e^{-\frac{r+1}{r}\sigma_r}, \qquad
S_r\left( L_{k}^{0}\right) =-\left(L_{k}^{0}+\imath h rJ_rP_k\right)e^{\frac{r+1}{r}\sigma_r },
\\
S_r\left( L_{0}^{0}\right) =-L_{0}^{0}-\imath hrJ_rP_0, \qquad S_r\left(
L_{k}^{m}\right)=-L_{k}^{m}.
\end{gather*}
Now using twisted coproducts and classical action
one can obtain by smash product construction (for f\/ixed value of the parameter~$r$) of extended
 position-momentum-$\mathfrak{gl}(n)$ algebra with the following crossed commutators:
\begin{gather}\label{WeylJord}
\left[ \hat{x}^{\mu },P_{0}\right]_r =\imath \delta^\mu
_0e^{\sigma _{r}}=\imath \delta^\mu_0(1-hrP_0)
, \qquad\left[ \hat{x}^{\mu },P_{k}\right]_r
=\imath\delta^\mu_k(1-hrP_0)^{-\frac{1}{r}},
\\
\nonumber
\left[ \hat{x}^{\mu },L_{k}^{m}\right]_r =\imath\hat{x}^{m}\delta^\mu_k,\qquad
\big[ \hat{x}^{\mu },L_{0}^{k}\big]_r =\imath\hat{x}
^{k}\delta^\mu_0(1-hrP_0)^{\frac{r+1}{r}},
\\
\nonumber
\left[ \hat{x}^{\mu },L_{k}^{0}\right]_r =\imath\hat{x}
^{0}\delta^\mu_k(1-hrP_0)^{-\frac{r+1}{r}}+\imath h\big(\hat{x}^{p}\delta_p^\mu
-r\hat{x}^{0}\delta^\mu_0\big)P_{k}(1-hrP_0)^{-1},
\\
\nonumber
\left[\hat{x}^{\mu },L_{0}^{0}\right]_r =\imath \hat{x}^{0}\delta^\mu
_0-\imath h\big( -\hat{x}^{k}\delta_k^\mu+r\hat{x}^{0}\delta^\mu
_0\big) P_{0}(1-hrP_0)^{-1}
\end{gather}
supplemented by $\mathfrak{igl}(n)$ (\ref{igl}) and $\kappa$-Minkowski  relations~(\ref{kMhat}).
One can see that relations (\ref{kMhat}), (\ref{WeylJord}) generate $r$-deformed phase space $\mathfrak{W}^n[[h]]^\mathcal{F}$.
Heisenberg realization is now in the following form:
\begin{gather}  \label{Jtlh}
  \hat{x}^{i}_{r}=x^{i}\left( 1-raP_0\right)^{-\frac{1}{r }}  \qquad
\mbox{and}\qquad   \hat{x}^{0}_{r}=x^{0}(1-raP_0).
\end{gather}
It is interesting to notice that above formulas~(\ref{Jtlh}) take the same form before and after Heisenberg realization. Moreover one can notice that commutation relation~(\ref{WeylJord}) can be reached by nonlinear change (\ref{Jtlh}) of generators: $(P_\rho,x^\lambda)\shortrightarrow(P_\rho,\hat{x}^\lambda_r)$. This illustrates Propositions~\ref{prop2},~\ref{prop3} for the Jordanian case. Again in the Heisenberg realization the classical $r$-matrices~(\ref{jt2}) coincide with Poisson bivector~(\ref{kM2}). 

\begin{remark}
Right-handed representations (Hermitian for $r=n-1$) is the following in Jordanian case:
\begin{equation}  \label{Jtrh}
\ \hat{x}^{i}_{R,r}=x^{i} , \qquad \mbox{and}\qquad \hat{x}%
^{0}_{R,r}=x^{0}(1-rh P_0)+h D
\end{equation}%
(or after Heisenberg realization for $\hat{x}^{0}_{R,r}=x^{0}(1-rh P_0)+hx^kP_k$).
Particularly, using (\ref{cc}), we obtain
\begin{equation}
[x^i, f]_{\star,r}= x^{i}\left(( 1-rh P_0)^{-\frac{1}{r }}-1\right)f
,\qquad [x^0, f]_{\star,r}= \imath h x^k\partial_k f
\end{equation}
which is different from (\ref{c1}). However for $r=-1$ one obtains the desired
commutator
\begin{equation}\label{c5}
[x^\mu, f]_\star=\imath\left( a^{\mu}x^\nu-a^{\nu}x^\mu\right)\partial_\nu f
\end{equation}
providing the $\kappa-$deformed Minkowski spacetime, i.e., $%
a^\mu=(h,0,\ldots,0)$.\end{remark}

\begin{remark}

For generic $r\neq 0$, the smallest subalgebra containing
simultaneously the Borel subalgebra (\ref{borel}) and one of the orthogonal
subalgebras $\mathfrak{iso}(g; n)$ [e.g., Poincar\'{e} subalgebra $\mathfrak{iso}%
(n-1, 1)$] is $\mathfrak{igl}(n)$. However, there are three exceptions:
\newline
(A) For $r=n-1$ in $n-$dimensional spacetime, the smallest subalgebra is $%
\mathfrak{isl}(n)$. \newline
(B) For $r=-1$ ($J_{-1}=-L=\imath(-D-L_0^0)$) in an arbitrary dimension, the smallest subalgebra
is Weyl-orthogonal algebra $\mathfrak{iwso}(g; n)$. It contains a central
extension of any orthogonal algebra $\mathfrak{so}(g; n)$ \footnote{%
Signature of the metric $g$ is irrelevant from an algebraic point of view.}. In
this case, the commutation relation (\ref{isog1}-\ref{isog2}) should be
supplemented by
\begin{equation}  \label{iwso}
[M_{\mu\nu}, L]=0, \qquad\qquad [P_\mu, L]= P_\mu
\end{equation}
Of course, for physical applications we will choose the Weyl-Poincar\'{e}
algebra. This minimal one-generator extension of the Poincar\'{e} algebra $\mathfrak{iso}(n-1, 1)$ has
been used in Ref. \cite{Ballesteros} (cf. next subsection).\newline
(C) $r=1$ in $n=2$ dimensions, $J_1=M_{10}$ is a boost generator for
nondiagonal metric $g_{00}=g_{11}=0$, $g_{01}=g_{10}=1$ with the Lorentzian signature. This corresponds
to the so-called light-cone deformation of the Poincar\'{e} algebra $\mathfrak{%
iso}(1, 1)$ \cite{LLM}.\end{remark}

\subsection*{Minimal case: Weyl- Poincar\'{e} algebra}
As mentioned above only for the case $r=-1$ (in physical $n=4$ dimensions) the covariance group can be reduced to one-generator (dilatation) extension of the Poincar\'e algebra~\cite{Ballesteros,BP}. (It is related to research in Ref. \cite{Ballesteros}, where, however, the $\kappa-$Minkowski spacetime has been obtained with different techniques). Below we shall present
coproducts and antipodes for all 11
generators in "physical" basis $(M_{k}, N_k, L, P_\mu)$ of the Poincar\'{e}-Weyl
algebra containing the Lorentz subalgebra of rotation $M_k=-\frac{\imath}{2}%
\epsilon_{klm}M_{lm}$ and boost $N_k=\imath M_{k0}$ generators:
\begin{eqnarray}  \label{sP1}
[M_i,\,M_j ]\ =\ \imath\,\epsilon_{ijk}\,M_k ~,\quad [M_i,\,N_j]\ =
\imath\,\epsilon_{ijk}\,N_k ~,\quad [N_i,\,N_j]\ =-\imath\,
\epsilon_{ijk}\,M_k~
\end{eqnarray}
Abelian four-momenta $P_\mu=-\imath\partial_\mu$ ($\mu=0,\dots,3\ , k=1,2,3$)%
\begin{eqnarray}  
[M_j,\,P_k]\!\!&=\!\!&\imath\,\epsilon_{jkl}\,P_l~,\qquad [M_j,\,P_0]\;=\;0~,\label{P3}
\\[5pt]
[N_j,\,P_k]\!\!&=\!\!&-\imath\,\delta_{jk}\,P_0~,\quad\;\;[N_j,\,P_0]\;=\;-%
\imath\, P_j^{}~.\label{P4}
\end{eqnarray}
and dilatation generator $L=x^\mu\partial_\mu$ as before:
\begin{equation}
[M_k,\,L]=[N_k,\,L]=0 ;\,\qquad [P_\mu,\,L]= P_\mu
\end{equation}
The deformed coproducts are
\begin{equation}
\tilde{\Delta}\left( P_{\mu}\right) =1\otimes P_{\mu}+P_{\mu}\otimes e^{%
\tilde{\sigma} }, \qquad \tilde{\Delta}\left( M_{k}\right) =1\otimes
M_{k}+M_{k}\otimes 1
\end{equation}
\begin{equation}
\tilde{\Delta}\left( N_{k}\right) =1\otimes N_{k}+N_{k}\otimes 1+hL\otimes P_{k}e^{-\tilde{\sigma} }
\end{equation}
\begin{equation}
\tilde{\Delta}\left( L\right) =1\otimes L+L\otimes 1\ +h
L\otimes P_{0}e^{-\tilde{\sigma} }
\end{equation}
Here $e^{-\tilde{\sigma}}=\left(1+hP_0\right)^{-1}$ and $%
L=-J_{-1}$ . The antipodes are
\begin{equation}
\tilde{S}\left( P_{\mu}\right) =-P_{\mu}e^{-\tilde{\sigma} }, \qquad \tilde{S%
}\left( M_{k}\right) =-M_{k}
\end{equation}
\begin{equation}
\tilde{S}\left( N_{k}\right) =-N_{k}+hLP_k\ , \quad \tilde{S}%
\left( L\right) =-L +hLP_0
\end{equation}

\chapter{The $\kappa$-Poincare quantum group}

The first deformations of Poincar\'{e} symmetry appeared in the early 1990's \cite{Luk1}-\cite{MR},\cite{WorPod},\cite{Zumino}.
$\kappa$ -Poincar\'{e} Hopf algebra has been originally discovered in the so-called
standard basis \cite{Luk1} inherited from the anty-de Sitter basis by the contraction procedure. In this basis only the rotational sector remains algebraically undeformed. Introducing  bicrossproduct basis allows to leave the lorentzian generators undeformed.  This basis has
 easier form of the $\kappa$ -Poincar\'{e} algebra and was postulated in \cite{MR}. In this form, the Lorentz subalgebra of the $\kappa$ -Poincar\'{e} algebra, generated by rotations and boosts is not deformed and the difference is only in the way the boosts act on momenta. There is also a change in co-algebraic sector, the coproducts are no longer trivial, which has a consequence: the underlying spacetime is noncommutative. In fact the Majid - Ruegg  bicrossproduct basis is the most popular one and considered by many authors for physical applications. 
However in this chapter we shall focus on $\kappa$ -Poincar\'{e} (Hopf) algebra in its classical Poincar\'{e} Lie algebra basis. The constructions of such basis were previously also investigated in several papers \cite{KosLuk}, \cite{Kos}.

But first let us extend the comment on why we can not use Drinfeld's twisting technique from previous chapter for twisting $\U_{\mathfrak{igl}(n)}$ in the case of Poincar\'{e} algebra. As it was already mentioned quantum deformations of the Lie algebra are controlled by classical $r$-matrices satisfying classical Yang--Baxter (YB) equation.
In the case of $r$-matrices satisfying homogeneous YB equations the co-algebraic sector is
twist-deformed while algebraic one remains classical~\cite{Drinfeld}.
Additionally, one can also apply existing twist tensors to relate Hopf module-algebras in order to obtain quantized, e.g., spacetimes (see \cite{MSSKG,BP,Meljanac0702215} as discussed in the previous section for quantizing Minkowski spacetime). For inhomogeneous $r$-matrices one applies Drinfeld--Jimbo (the so-called standard) quantization instead.
 
\begin{remark}\label{DJ}
Drinfeld--Jimbo quantization algorithm relies on
simultaneous deformations of the algebraic and coalgebraic sectors and
applies to semisimple Lie algebras \cite{Drinfeld,Jimbo}. In particular, it
implies existence of classical basis for Drinfeld--Jimbo quantized algebras.
Strictly speaking, the Drinfeld--Jimbo procedure cannot be applied to the
Poincar\'{e} non-semisimple algebra which has been obtained by the
contraction procedure from the Drinfeld--Jimbo deformation of the
anti-de~Sitter (simple) Lie algebra $\mathfrak{so}(3,2)$. Nevertheless,
quantum $\kappa$-Poincar\'{e} group shares many properties of the original
Drinfeld--Jimbo quantization. These include existence of classical basis, the
square of antipode and solution to specialization problem~\cite{BP2}. There is no cocycle twist related with the Drinfeld--Jimbo deformation. The Drinfeld--Jimbo quantization has many non-isomorphic forms (see e.g.~\cite{Klimyk,Chiari}).
\end{remark}

\begin{remark}
Drinfeld--Jimbo quantization can be considered from the more general point of view, in the so-called quasi bialgebras framework. In this case more general cochain twist instead of ordinary cocycle twist can be used together with a coassociator in order to perform quantization. Cochain twists may lead to non-associative star multiplications~\cite{non-ass}.
\end{remark}
In the case of relativistic  symmetries, such classification (complete for the Lorentz and almost-complete for Poincar\'{e} algebras) has been performed in Ref. \cite{Zakrzewski} (see also \cite{Lya} where this classification scheme has been extended). Particularly, $r$ -matrix which corresponds to $\kappa$ -deformation of Poincar\'{e} algebra is given by
 \be r=
 N_i\wedge P^i \ee
and  it satisfies the inhomogeneous (modified) Yang-Baxter equation:
\be  [[r,r]]=
M_{\mu\nu}\wedge P^\mu\wedge P^\nu \ee
where $[[\cdot,\cdot]]$ is the so-called Schouten bracket (a.k.a. Schouten - Nijenhuis bracket) and it is defined\footnote{The skew symmetric Schouten - Nijenhuis bracket \cite{Schouten} is the unique extension of the Lie bracket of vector fields to a graded bracket on the space of alternating multivector fields that makes the alternating multivector fields into a Gerstenhaber algebra. Here we consider only the case of twovector fields, of which a r - matrix is an example.}
as follows $\forall a,b,c,d \in\mathfrak{g}$ ($\mathfrak{g}$ - Lie algebra) : 
\be 
[[a\wedge b, c\wedge d]]=[[c\wedge d,a\wedge b]]=
[a, d]\wedge b\wedge c - [a,c]\wedge b\wedge d +[b,c]\wedge a\wedge d -[b,d]\wedge a\wedge c
\ee
Therefore, one does not expect to obtain the $\kappa$ Poincar\'{e} coproduct by twist.
However, most of the items on that list contain homogeneous $r-$matrices. Explicit twists for them have been provided in Ref. \cite{Varna} (for superization see \cite{VNT1,VNT2}); the corresponding quantization has been carried out in \cite{BLT2}.

Our purpose in this section is to formulate $\kappa-$Poincar\'{e} Hopf algebra in classical Poincar\'{e} basis. Possibility of def\/ining $\kappa$-Poincar\'{e} algebra in a classical basis has been under debate for a long time, see, e.g.,  \cite{KosLuk,Kos, GNbazy, Meljanac0702215, BP2, Group21}. 
\section{$h$-adic $\kappa$-Poincar\'{e} quantum group in classical\\ basis}
We take the Poincar\'{e} Lie algebra $\mathfrak{io}(1,3)$ provided with a convenient choice of ``physical'' generators $(M_i,N_i,P_\mu)$\footnote{From now one we shall work with physical dimensions $n= 4$, however generalization to arbitrary dimension~$n$ is obvious.}and Lorentzian metric $\eta_{\mu\nu}=diag(-,+,+,+)$
for rising and lowering indices.
\begin{gather}  \label{L1}
[ M_{i},M_{j}]=\imath \epsilon _{ijk}M_{k},\quad [ M_{i},N_{j}]=
\imath \epsilon _{ijk} N_{k},\quad [N_{i},N_{j}]=-\imath \epsilon _{ijk} M_{k},\\
[P_\mu,P_\nu]=\lbrack M_{j},P_{0}]=0,\qquad [M_{j},P_{k}]=\imath \epsilon _{jkl}P_{l}, \label{L2}\\
[N_{j},P_{k}] =- \imath \delta _{jk}P_{0},\qquad [N_{j},P_{0}]=-\imath P_{j}. \label{L3} 
\end{gather}
The algebra $(M_i,N_i,P_\mu)$ can be
extended in the standard way to a Hopf algebra by defining on the universal
enveloping algebra $\U_{\mathfrak{io}(1,3)}$ the coproduct $\Delta_0$, the counit
$\epsilon$, and the antipode $S_0$, where the nondeformed - primitive coproduct, the
antipode and the counit are given by \be\label{undef} \Delta_0(X)=X\otimes
1+1\otimes X , \quad S_0(X)=-X , \quad \epsilon(X)=0 \ee
for $X\in\mathfrak{io}(1,3)$. In addition $\Delta_0(1)=1\otimes 1$, $%
S_0(1)=1$ and $\epsilon(1)=1$. For the purpose of deformation one has to extend
further this Hopf algebra by considering formal power series in $h$, and correspondingly considering the Hopf algebra ($\U_{\mathfrak{io}(1,3)}[[h]],\cdotp%
,\Delta_0,S_0,\epsilon)$ as a topological Hopf algebra with the so-called "h-adic" topology (see also Appendix A, \cite{Klimyk,Chiari}).

The structure of the Hopf algebra can be def\/ined on $\U_{\mathfrak{io}(1,3)}[[h]]$
by establishing deformed coproducts of the generators~\cite{BP2} and the antipodes leaving algebraic sector classical (untouched) like in the case of twisted deformation, it is as follows:
\begin{gather}\label{copM}
\Delta _{\kappa }\left( M_{i}\right) =\Delta _{0}\left( M_{i}\right)=M_i\otimes 1+1\otimes M_i,\\
\Delta _{\kappa }\left( N_{i}\right) =N_{i}\otimes
1+\left(h P_{0}+\sqrt{1-h^{2}
P^{2}}\right)^{-1}\!\otimes N_{i}-h\epsilon _{ijm}
P_{j}\left(h P_{0}+\sqrt{1-h^{2}
P^{2}}\right)^{-1}\!\otimes M_{m},\\
\Delta _{\kappa }\left( P_{i}\right) =P_{i}\otimes
\left(hP_{0}+\sqrt{1-h^{2}
P^{2}}\right)+1\otimes P_{i},  \\
\Delta _{\kappa }\left( P_{0}\right) =P_{0}\otimes
\left(h P_{0}+\sqrt{1-h^{2}
P^{2}}\right)+\left(hP_{0}+\sqrt{1-h^{2}
P^{2}}\right)^{-1}\otimes P_{0}\\
\hphantom{\Delta _{\kappa }\left( P_{0}\right) =}{}
+hP _{m}\left(hP_{0}+\sqrt{1-h^{2}
P^{2}}\right)^{-1}\otimes P^{m},\label{copP0}
\end{gather}
and the antipodes
\begin{gather}\label{SM}
S_{\kappa }(M_{i})=-M_{i},\qquad S_\kappa(N
_{i})=-\left(hP_{0}+\sqrt{1-h^{2}
P^{2}}\right)N_{i}-h\epsilon _{ijm}P_{j}%
M_{m},
\\
S_\kappa(P_{i})=-P_{i}\left(hP_{0}+\sqrt{1-h^{2}%
P^{2}}\right)^{-1},\quad S_\kappa(P
_{0})=-P_{0}+h\vec{P}^{2}\left(hP_{0}+\sqrt{1-h^{2}
P^{2}}\right)^{-1},\label{SP}
\end{gather}
where $P^{2}\doteq P_{\mu }P^{\mu
}\equiv\vec{P}^{2}-P_0^2$, and $\vec{P}^{2}=P
_{i}P^{i}$. One sees that above expressions are  formal
 power series in the formal parameter $h$, 

\begin{equation}\label{sqrt}
\sqrt{1-hP^{2}}=\sum_{n\geq 0}(-1)^n h^{2n}\,\binom{1/2}{n}
 \,[P^{2}]^n
\end{equation}
 where $\binom{1/2}{n}=\frac{1/2(1/2-1)\ldots (1/2-n+1)}{n!}$ are binomial coefficients. Introducing the following shortcuts: 
\begin{equation}  \label{Pi1}
\Xi\doteq hP_{0}+\sqrt{1-h^{2}%
P^{2}}
\quad \mbox{and}\quad
\Xi^{-1}\doteq\frac{\sqrt{1-h^{2}P^{2}}-hP_{0}}{1-h^{2}\vec{P}^{2}}%
\end{equation}
in the formulas above one can calculate that
\begin{equation}  \label{Pi3}
\Delta_\kappa (\Xi)=\Xi\otimes \Xi,\quad \Delta_\kappa
(\Xi^{-1})=\Xi^{-1}\otimes \Xi^{-1}, \qquad S_\kappa(\Xi)=\Xi^{-1}
\end{equation}
as well as
\begin{equation}
\ \Delta _{\kappa }\left(\sqrt{1-h^{2}P^{2}}\right) =\sqrt{%
1-h^{2}P^{2}}\otimes \Xi -h%
\Xi^{-1}\otimes P_{0}-h^2P%
_{m}\Xi^{-1}\otimes P^{m}
\end{equation}%
To complete the definition one leaves the counit $\epsilon$ undeformed. Let us observe that $\epsilon(\Xi)=\epsilon(\Xi^{-1})=1$. It is also worth noticing that the square of the antipode (\ref{SM}-\ref{SP}) is given by a similarity transformation \footnote{In the case of twisted deformation the antipode itself is given by the similarity transformation.}, i.e. $$S_\kappa^2(X)=\Xi X\Xi^{-1}$$
Substituting now\be\label{change_bicross}
\P_0\doteq h^{-1}\ln\Xi , \qquad \P_i\doteq P _i\Xi^{-1}\quad \Rightarrow
\quad\Xi=e^{h\P_0} \ee
one gets the deformed coproducts of the form
\begin{eqnarray}
\Delta_\kappa \left( \P_{0}\right) &=&1\otimes \P_{0}+\P_{0}\otimes 1,\quad
\Delta_\kappa \left( \P_{k}\right) =e^{-h\P_{0}}\otimes
\P_{k}+\P_{k}\otimes 1  \label{kP1} \\
\Delta_\kappa \left( N_{i}\right) &=&N_{i}\otimes 1+e^{-h\P_{0}}\otimes N_{i}-h\epsilon _{ijm}\P_{j}\otimes N_{m}
\label{kP2}
\end{eqnarray}%
Similarly the commutators of new generators can be obtained as
\begin{equation}
\left[ N_{i}, \P_{j}\right] =-\frac{\imath}{2} \delta
_{ij}\left(h^{-1}\left( 1-e^{-2h\P_{0}}\right) +h\vec{P}^{2}\right) + \imath hP_{i}P_{j}
\end{equation}%
with the remaining one being the same as for Poincar\'{e} Lie
algebra (\ref{L1})-(\ref{L3}). This proves that our deformed Hopf algebra (%
\ref{L1})-(\ref{L2}), (\ref{copM})-(\ref{copP0}) is Hopf isomorphic to the
$\kappa $ -Poincar\'{e} Hopf algebra \cite{Luk1} written in its bicrossproduct
basis $(M_i, N_i, \P_\mu)$ \cite{MR} and determines celebrated $\kappa$-Poincar\'{e} quantum group\footnote{We shall follow traditional terminology calling $\U_{\mathfrak{io}(1,3)}[[h]]$ $\kappa$-Poincar\'{e} with parameter $h$ of $[{\rm lenght}]=[{\rm mass}]^{-1}$ dimension.} in a classical basis as a Drinfeld--Jimbo deformation equipped with $h$-adic topology \cite{BP2} and from now on we shall denote it as $\U_{\mathfrak{io}(1,3)}[[h]]^{\rm DJ}$.
 This version of $\kappa$-Poincar\'{e} group is ``$h$-adic'' type and is considered as a traditional approach. The price we have to pay for it is that the deformation parameter cannot be determined, must stay formal, which means that it cannot be related with any constant of Nature, like, e.g., Planck mass or more general quantum gravity scale. In spite of not clear physical interpretation this version of $\kappa$-Poincar\'{e} Hopf algebra has been widely studied since its f\/irst discovery in~\cite{Luk1}.

\begin{remark}
The following immediate comments are now in order:\newline
i) Substituting $N_i=M_{0i}$ and $\epsilon_{ijk}M_k=M_{ij}$ the above
result easily generalizes to the case of \\
$\kappa $ -Poincar\'{e} Hopf
algebra in an arbitrary spacetime dimension $n$ (with the Lorentzian signature).
\newline
ii) Although $\mathfrak{io}(1,n-1)$ is Lie subalgebra of $\mathfrak{io}(1,n)
$ the corresponding Hopf algebra $\mathcal{U}_{\mathfrak{io}(1,n-1)}[[h]]$ (with $%
\kappa-$deformed coproduct) is not Hopf subalgebra of $\mathcal{U}_{\mathfrak{io}
(1,n)}[[h]]$.\newline
iii) Changing the generators by a similarity transformation $X\rightarrow
SXS^{-1}$ for $X\in (M_i, N_i, P _\mu)$ leaves the algebraic sector (\ref%
{L1})-(\ref{L3}) unchanged but in general it changes coproducts (\ref{copM}%
)-(\ref{copP0}). Here $S$ is assumed to be an invertible element in $\U_{\mathfrak{io}(1,3)}[[h]]^{\rm DJ}$. The both commutators and coproducts (\ref%
{copM})-(\ref{copP0}) are preserved provided that $S$ is group-like, i.e. $%
\Delta_\kappa (S)=S\otimes S$, e.g., $\Xi$. For physical applications it
might be also useful to consider another (nonlinear) changes of basis, e.g.
in the translational sector. Therefore the algebra $\U_{\mathfrak{io}(1,3)}[[h]]^{\rm DJ}$ is a convenient playground for developing
Magueijo-Smolin type DSR theories \cite{Smolin}, \cite{Gosh} (DSR2) even if
we do not intend to take into account coproducts. But the coproducts are
there and can be used, e.g., in order to introduce an additional law for
four-momenta. In this situation the $\kappa$-deformed coproducts are not
necessary the privileged ones and the additional law can be determined by,
e.g., the twisted coproducts \cite{BLT2}.
However $\kappa$-deformed coproducts are consistent with  $\kappa$-Minkowski commutation relations and give  $\kappa$-Minkowski spacetime module algebra structure \cite{BP}, \cite{Meljanac0702215}, \cite{BP4}.
\end{remark}

\section{Different algebraic form of $\protect\kappa$-Poincar\'{e} Hopf\\ algebra}
The mathematical formalism of quantum groups requires us to deal with formal
power series. Therefore the parameter $h$ has to stay formal, i.e.,
undetermined. Particularly, we can not assign any specific numerical value
to it and consequently any fundamental constant of nature, like, e.g., the
Planck mass cannot be related with it. Nevertheless there exists a method to
remedy this situation and allow $h$ to admit constant value. Therefore the
$\kappa$-Poincar\'{e} quantum group with $h$-adic topology as described above, is not the only possible version.\\
The idea is to reformulate a Hopf algebra in a way which allows to hide inf\/inite series on the abstract level. In the traditional Drinfeld-Jimbo
approach this is always possible  by using the so-called specialization method
or q-deformation and introduce the so-called ``$q$-analog'' version\footnote{In some physically motivated papers a phrase ``$q$-deformation'' is considered as an equivalent of Drinfeld--Jimbo deformation. In this section we shall, following general terminology of~\cite{Klimyk,Chiari}, distinguish between ``$h$-adic'' and ``$q$-analog'' Drinfeld--Jimbo deformations since they are not isomorphic.}, which allows us to f\/ix value of the parameter $h=\kappa$. Afterwards they become usual complex algebras without the $h$-adic topology. In the case of Drinfeld--Jimbo deformation this is always possible.
 As a result one obtains a one-parameter family of isomorphic Hopf algebras enumerated by numerical parameter $\kappa$.
We shall describe this procedure for the case of $\kappa$-Poincar\'{e} (cf.~\cite{BP2}\footnote{Similar construction has been performed in~\cite{Stachura} in the bicrossproduct basis (see also~\cite{Z}).}).

The main idea behind $q$-deformation is to introduce two mutually inverse group-like elements hiding inf\/inite power series. Here we shall denote them as $\Pi_0$, $\Pi_0^{-1}$: $\Pi_0^{-1}\Pi_0=1$. Fix~\mbox{$\kappa\in \mathbb{C}^*$}. Denote by $\U_\kappa(\mathfrak{io}(1,3))$ a universal, unital and
associative algebra over complex numbers generated by eleven generators $(M_i, N_i, P_i,
\Pi_0, \Pi_0^{-1})$ with the following set of def\/ining relations:
\begin{gather}
\nonumber \Pi_0^{-1}\Pi_0=\Pi_0\Pi_0^{-1}=1,\qquad
[P _i, \Pi_0]=[M_j, \Pi_0]=0 ,\qquad [N_i, \Pi_0]=-\frac{\imath}{\kappa}P_i, 
\\
 [N_i, P_j]=-\frac{\imath}{2}\delta_{ij}\left(\kappa(\Pi_0-\Pi_0^{-1})+\frac{1}{\kappa}\vec{P}^2 \Pi_0^{-1}\right),
\label{11alg}
\end{gather}
where remaining relations between $(M_i,N_i,P_i)$ are the same as in the Poincar\'{e} Lie algeb\-ra~(\ref{L1}) and $[P_i,P_j]=0,\qquad [M_{j},P_{k}]=\imath \epsilon _{jkl}P_{l}, $.
Commutators with $\Pi_0^{-1}$ can be easily calculated from (\ref{11alg}), e.g., $[N_i,\Pi_0^{-1}]=\frac{i}{\kappa}P_i\Pi_0^{-2}$.
Notice that all formulas contain only f\/inite powers of the numerical parameter~$\kappa$.
The  quantum algebra structure $\U_\kappa(\mathfrak{io}(1,3))$ is provided by def\/ining coproduct, antipode and counit, i.e.\ a Hopf algebra structure.
We set
\begin{gather*}
 \Delta_\kappa\left( M_{i}\right) =M_{i}\otimes 1+1\otimes M_{i},
\nonumber\\
 \Delta_\kappa\left( N_{i}\right) =N_{i}\otimes
1+\Pi_0^{-1}\otimes N_{i}-\frac{1}{\kappa}\epsilon _{ijm}
P_{j}\Pi_0^{-1}\otimes M_{m},
\\
\Delta_\kappa\left(P_{i}\right) =P_{i}\otimes
\Pi_0+1\otimes P_{i},  \qquad
\Delta_\kappa(\Pi_0)=\Pi_0\otimes \Pi_0,\qquad \Delta_\kappa
(\Pi_0^{-1})=\Pi_0^{-1}\otimes \Pi_0^{-1},\nonumber
\end{gather*}
and the antipodes
\begin{gather*}
S_\kappa(M_{i})=-M_{i},\qquad S_\kappa(N
_{i})=-\Pi_0N_{i}-\frac{1}{\kappa}\epsilon _{ijm}P_{j}M_{m},\qquad S_\kappa(P_{i})=-P_{i}\Pi_0^{-1},\\
S_\kappa(\Pi_0)=\Pi_0^{-1}, \qquad S_\kappa(\Pi_0^{-1})=\Pi_0.
\end{gather*}

To complete the def\/inition one leaves counit $\epsilon$ undeformed, i.e., $\epsilon(A)=0$ for $A\in(M_i, N_i, P _i)$ and
$\epsilon(\Pi_0)=\epsilon(\Pi_0^{-1})=1$.

For the purpose of physical interpretation, specialization of the parameter $\kappa$ makes possible its identif\/ication with some physical constant of Nature: typically it is quantum gravity scale~$M_Q$. However we do not assume a priori that this is Planck mass (for discussion see Part II, and e.g.,~\cite{BGMP}). Particularly, the value of $\kappa$ depends on a system of units we are working with. For example, one should be able to use natural (Planck) system of units, $\hbar=c=\kappa=1$, without changing mathematical properties of the underlying physical model.
From mathematical point of view, this means that the numerical value of parameter $\kappa$ is irrelevant. And this is exactly the case we are dealing with.
For dif\/ferent numerical values of $\kappa\neq 0$ the Hopf algebras $\U_\kappa(\mathfrak{io}(1,3))$ are isomorphic, i.e.\ $\U_\kappa(\mathfrak{io}(1,3))\cong\U_1(\mathfrak{io}(1,3))$. One can see that by re-scaling $P_\mu\mapsto \frac{1}{\kappa} P_\mu$.

One can also introduce the generator $P _0$ as
\be\label{qP_0}
P _0\equiv P_0(\kappa)\doteq
\frac{\kappa}{2}\left(\Pi_0-\Pi_0^{-1}(1-\frac{1}{\kappa^2}\vec{\P }%
^2)\right) \ee
Thus subalgebra generated by elements $(M_i, N_i, P _i,P_0)$ is, of course, isomorphic to
the universal envelope of the Poincar\'{e} Lie algebra, i.e. $\U_{\mathfrak{io}(1,3)}\subset\U_\kappa(\mathfrak{io}(1,3))$. But this is not a Hopf subalgebra.
Therefore, the original (classical) Casimir element $\mathcal{C}\equiv-P^2
=P _0^2-\vec{P }^2$ has, in terms of  the generators $(\Pi_0, \Pi_0^{-1}, \vec{P})$,
rather complicated form.  We can adopt to our disposal a simpler (central) element instead:
\be
\mathcal{C}_\kappa\doteq\kappa^2(\Pi_0^{}+\Pi_0^{-1}-2)-\vec{P }%
^2 \Pi_0^{-1} \ee
which is responsible for deformed dispersion relations \cite{BP}. For comparison see, e.g., \cite{Casimir}.
Both elements are related by
\begin{equation}
\mathcal{C} =\mathcal{C}_\kappa\left(1+\frac{1}{4\kappa^2}\mathcal{C}_\kappa\right) \quad\mbox{and}\quad
\sqrt{1+\frac{1}{\kappa^2}{\mathcal{C}}}=1+\frac{1}{2\kappa^2}\mathcal{C}_\kappa
\end{equation}
We will concentrate on consequences of this result later on in Part II.

%
%
%


\chapter{$\kappa$-Minkowski as covariant quantum spacetime and DSR algebras}

In the classical case the physical symmetry group of Minkowski spacetime is Poincar\'{e} group and in deformed case analogously quantum $\kappa$-Poincar\'{e}  group should be desired symmetry group for quantum $\kappa$-Minkowski spacetime \cite{MR,Z}. However $\kappa$-Minkowski module algebra studied in the Chapter 2 has been obtained as covariant space over the $\U_{\mathfrak{igl}}^\mathcal{F}[[h]]$ Hopf algebra. Moreover, presented twist constructions
do not apply to Poincar\'{e} subalgebra, what is due to the fact that $\kappa$-Poincar\'{e} algebra is a quantum deformation of Drinfeld--Jimbo type corresponding to inhomogeneous classical r-matrix as it was already explained before. But since one does not expect to obtain $\kappa$-Poincar\'{e} coproduct by cocycle twist, $\kappa$-Minkowski spacetime  has to be introduced as a covariant quantum space in a way without using twist and twisted star product, as it will be shown in this Chapter.
Having defined $\kappa$-Minkowski spacetime as Hopf module algebra, one can extend $\kappa$-Poincar\'{e} algebra by $\kappa$-Minkowski using crossed (smash) product construction. In particular, this contains deformation of Weyl subalgebra, which is crossed-product of $\kappa$-Minkowski algebra with algebra of four momenta. One should notice that there have been other constructions of $\kappa$-Minkowski algebra ($\kappa$-Poincar\'{e} algebra) extension, by introducing $\kappa$-deformed phase space, e.g., in~\cite{LukNow,ncphasespace,AF} (see also~\cite{JKG}, the name ``DSR algebra'' was f\/irstly proposed here and it comes from the following interpretation: the different realizations of this algebra lead to different doubly (or deformed) special relativity models with different physics encoded in deformed dispersion relations.). In most of the cases it appeared  as Heisenberg double construction however we would like to point out that it is not the only way to obtain DSR algebra, we can obtain it by smash- product construction as well. Particularly this will lead to a deformation generated by momenta and coordinates of Weyl (Heisenberg) subalgebra. One of the advantages of using smash-product construction is that we leave an open geometrical interpretation of $\kappa$-Minkowski spacetime.
 Examples of $\kappa$-deformed phase space obtained by using crossproduct construction can be found also in~\cite{Kos,Nowicki}, however only for one specif\/ic realization related with the bicrossproduct \mbox{basis}.\\
This chapter comprises dif\/ferent versions of quantum Minkowski spacetime algebra.
We point out the main dif\/ferences between those versions, i.e. $h$-adic and $q$-analog, and we will construct DSR algebras for both of them which will contain deformed phase spaces as subalgebras.  We show explicitly that DSR models depend upon various Weyl algebra realizations  of $\kappa$-Minkowski spacetime one uses. The most interesting one from physical point of view seems to be version of $\kappa$-Minkowski spacetime with f\/ixed value of parameter $\kappa$ ($q$-analog). In this case $\kappa$-Minkowski algebra is an universal envelope of solvable Lie algebra without $h$-adic topology. This version allows us to connect the parameter $\kappa$ with some physical constant, like, e.g., quantum gravity scale or Planck mass and all the realizations might have physical interpretation. Also this version is used in some Group Field Theories \cite{Oriti} which might be connected with loop quantum gravity and spin foams approach.

In this chapter we will start with
``$h$-adic'' version of $\kappa$-Minkowski spacetime. Then using smash product construction, described in Chapter 1, we will obtain ``$h$-adic'' DSR algebra from
 ``$h$-adic'' universal $\kappa$-Minkowski module algebra with ``$h$-adic'' $\kappa$-Poincar\'{e} (with its classical action). Subsequently (in Section 4.2) we perform similar construction for q-analog case. 


\section{``$h$-adic'' universal $\kappa$-Minkowski spacetime\\ and ``$h$-adic'' DSR algebra.}
Since $\kappa$-Poincar\'{e} Hopf algebra presented in Section 3.1 is not obtained by the twist deformation one needs a~new construction of $\kappa$-Minkowski spacetime as a $U_{\mathfrak{io}(1,3)}[[h]]^{\rm DJ}$ (Hopf) module algebra.

Following Remark~\ref{Poisson2} and Example~\ref{example8}, we are in position to introduce $\kappa$-Minkowski spacetime as a universal $h$-adic  algebra
$\U_h(\mathfrak{an}^3)$ with def\/ining relations\footnote{We take Lorentzian metric $\eta_{\mu\nu}={\rm diag}(-,+,+,+)$ for rising and lowering indices, e.g., $\X_\lambda=\eta_{\lambda\nu}\X^\nu$.}:
\begin{gather}\label{kMcal}
[\X_0, \X_i]=-\imath h \X_i,\qquad [\X_j, \X_k]=0
\end{gather}
Due to universal construction there is a $\mathbb{C}[[h]]$-algebra epimorphism of $\U_h(\mathfrak{an}^3)$ onto $\mathfrak{X}^4[[h]]^\mathcal{F}$ for any $\kappa$-Minkowski twist~$\mathcal{F}$. Before applying smash product construction one has to assure that  $\U_h(\mathfrak{an}^3)$ is $\U_{\mathfrak{io}(1,3)}[[h]]^{\rm DJ}$ ($\kappa$-Poincar\'{e}) covariant algebra. It can be easily done by exploiting  the classical action
(see also Example~\ref{ortho})
\begin{gather}
\label{ClassAction1}
P_{\mu }\triangleright \X^{\nu }=-\imath\delta _{\mu }^{\nu },\qquad
M_{\mu \nu }\triangleright \X^{\rho }= \imath \X_{\nu }\delta _{\mu }^{\rho }-
\imath \X_{\mu }\delta _{\nu }^{\rho}.
\end{gather}
and by checking out consistency conditions similar to those introduced in Example~\ref{example2}, equation~(\ref{smash2}).
\begin{remark}\label{shifted} As it was already noticed, the algebra $\U_h(\mathfrak{an}^{3})$ is dif\/ferent than $\U_{\mathfrak{an}^3}[[h]]$ (see Remark \ref{Poisson2}). Assuming $P_\mu \triangleright \X^\nu=\imath a \delta^\nu_\mu$ and $[\X^0, \X^k]=\imath b \X^k$ one gets from (\ref{smash2}) and the $\kappa$-Poincar\'{e} coproduct the following relation: $b=-ah$. Particulary, our choice $a=-1$ (cf.~(\ref{ClassAction1})) does imply $b=h$. In contrast $b=1$ requires $a=h^{-1}$ what is not possible for formal parameter~$h$.
This explains why the classical action cannot be extended to the unshifted generators $X^\nu$ and
to the entire algebra $\U_{\mathfrak{an}^3}[[h]]$. The last one seemed to be the most natural candidate for $\kappa$-Minkowski spacetime algebra in the $h$-adic case.
\end{remark}

With this in mind one can def\/ine now DSR algebra as a crossed product extension of $\kappa$-Minkowski and $\kappa$-Poincar\'{e} algebras (\ref{L1})--(\ref{L3}), (\ref{kMcal}). It is determined  by the following $\U_h(\mathfrak{an}^3)\rtimes \U_{\mathfrak{io}(1,3)}[[h]]^{\rm DJ}$ cross-commutation relations:
\begin{gather}\label{L5}
[ M_{i},\X_{0}]=0\quad \lbrack N_{i},%
\X_{0}]=-\imath \X_{i}-\imath h N_{i},
\\
\label{L6}
[ M_{i},\X_{j}]=\imath \epsilon _{ijk}\X
_{k},\qquad [ N_{i},\X_{j}]=-\imath \delta _{ij}%
\X_{0}+\imath h\epsilon _{ijk}M_{k},
\\ \label{L7}
[P_k, \X_0]=0 , \qquad [P_k, \X_j]=-\imath\delta_{jk}\left(hP_0
+\sqrt{1-h^{2}P^2}\right),
\\  \label{L8}
[P_0, \X_j]=-\imath h P_j , \qquad [P_0, \X_0]=\imath\sqrt{1-h^{2}P^2}.
\end{gather}

\begin{remark}
Relations (\ref{L5}), (\ref{L6}) can be rewritten in a covariant form:
\begin{gather*}  
 \left[ M_{\mu\nu}, \X_\lambda\right] = \imath \eta_{\mu\lambda}\X_\mu
-\imath \eta_{\nu\lambda}\X_{\mu} -\imath  h a_\mu M_{\nu\lambda}+\imath
ha_\nu M_{\mu\lambda},
\end{gather*}
where $a_\mu=\eta_{\mu\nu}a^\nu$, $(a^\nu)=(1, 0, 0, 0)$.\\
This form is suitable for generalization to arbitrary spacetime dimension $n$.
\end{remark}

\begin{remark}
The following change of generators:
\[
\tilde{\X}_0=\X_0,\qquad \tilde{\X}_j=\X_j+hN_j
\] provides Snyder type of commutation relation and leads to the algebra which looks like a central extension obtained in~\cite{Chryss}
\begin{gather*}\nonumber
[ P_{\mu },\tilde{\X}_{\nu }] =-i\eta_{\mu\nu }M,\qquad
[ P_{\mu },P_{\nu }] =0,\qquad
[\tilde{\X}_\mu ,\tilde{\X}_{\nu }] =\imath h^2M_{\mu \nu },
\\
[ P_{\mu },M] =0,\qquad
[\tilde{X}_\mu,M] =-\imath h^2P_{\mu },
\end{gather*}
where $M=\sqrt{1-h^2P^{2}}$ plays the role of central element and $M_{\mu\nu}$ are Lorentz generators.
\end{remark}

The center of the algebra $\U_{\mathfrak{io}(1,3)}[[h]]$ is an algebra over $\mathbb{C}[[h]]$.
Therefore one can consider a~deformation of the Poincar\'e Casimir operator $P^2$. In fact, for any power series of two variables $f(s,t)$ the element: $\mathcal{C}_f=f(P^2,h)$ belongs to the center as well. The reason of using deformed Casimir instead of the standard one is that
the standard one fails, due to noncommutativity of~$\X_\mu$, to satisfy the relation $ [P^{2},\X^{\mu }]=2P^{\mu }$. Considering deformed Casimir one has freedom to choose the form of the function~$f$. The choice
\[
\mathcal{C}_{h}=2h^{-2}\left(\sqrt{1+h^2P^2}-1\right)
\]
allows to preserve the classical properties:
 $
[M_{\mu \nu },\mathcal{C}_h]=\lbrack \mathcal{C}_h,P_{\mu }]=0$, $[\mathcal{C}_h,\X_{\mu
}]=2P_{\mu }$.

The standard Poincar\'{e} Casimir  gives rise to undeformed dispersion relation:
\begin{gather}
P^{2}+m_{ph}^{2}=0,  \label{nondef}
\end{gather}
where $m_{ph}$ is mass parameter (physical mass, which comes from representation of the
Poincar\'{e} algebra). The second Casimir operator leads to deformed dispersion relations
\begin{gather}
\mathcal{C}_h+m_{h}^{2}=0  \label{def}
\end{gather}
with the deformed mass parameter $m_h$. Relation between this two mass parameters has the following form \cite{BGMP,s1}:
\[
 m^{2}_{ph}=m_h^{2}\left(1-\frac{h^2}{4}m_{h }^{2}\right).
\]
Clearly, DSR algebra $\U_h(\mathfrak{an}^3)\rtimes \U_{\mathfrak{io}(1,3)}[[h]]^{\rm DJ}$ as introduced above is a deformation of the algebra (\ref{isog1})--(\ref{isog3}) from Example (\ref{ortho}) for the Lorentzian (when $g_{\mu\nu}=\eta_{\mu\nu}$) signature; i.e., the latter can be obtained as a limit of the former when $h\rightarrow 0$. Moreover,  $\U_h(\mathfrak{an}^3)\rtimes \U_{\mathfrak{io}(1,3)}[[h]]^{\rm DJ}$ turns out to be a quasi-deformation of $\mathfrak{X}^4\rtimes \U_{\mathfrak{io}(1,3)}[[h]]$ in a sense of Propositions~\ref{prop1},~\ref{prop2}\footnote{Now we cannot make use of the twisted tensor technique from the proof of Proposition~\ref{prop1}. However, we believe, that analog of Proposition~\ref{prop1} is valid for Drinfeld--Jimbo type deformations as well.}. It means that the deformed algebra can be realized  by  nonlinear change of generators in the undeformed one. To this aim we shall use explicit change of generators inspired by covariant Heisenberg realizations proposed f\/irstly in \cite{Meljanac0702215}:
\begin{gather}\label{naturalrealization}
\X^{\mu}=x^{\mu}\left( hp_0 +\sqrt{1-h^{2}p^2}\right)-hx^0p^\mu,\qquad M_{\mu \nu }=M_{\mu\nu},\qquad P_\mu=p_\mu
  \end{gather}%
in terms of undeformed Weyl--Poincar\'e algebra (\ref{isog1})--(\ref{isog3}) satisfying the canonical commutation relations:
\begin{gather}
\label{unWeyl}
\left[ p_{\mu },x_{\nu }\right]=-\imath \eta _{\mu \nu }, \qquad
\left[ x_{\mu },x_{\nu }\right]=
\left[ p_{\mu },p_{\nu }\right]=0.
\end{gather}
\begin{proof}The transformation $(x^\mu,M_{\mu\nu},p_\mu)\longrightarrow(\X^\mu,M_{\mu\nu},P_\mu)$ is invertible. It is subject of easy calculation that generators (\ref{naturalrealization}) satisfy DSR algebra (\ref{L5})--(\ref{L8}) provided that $(x^\mu,M_{\mu\nu},p_\mu)$ satisfy Weyl extended Poincar\'{e} algebra  (\ref{isog1})--(\ref{isog3}). This f\/inishes the proof. Note that in contrast to~\cite{Meljanac0702215} we do not require Heisenberg realization for $M_{\mu\nu}$.\footnote{We recall that in Heisenberg realization $M_{\mu\nu}=x_\mu p_\nu-x_\nu p_\mu $.} Particulary, deformed and undeformed Weyl algebras are $h$-adic isomorphic.
\end{proof}
Above statement resembles, in a sense, the celebrated Coleman--Mandula theorem~\cite{CM}: there is no room for non-trivial combination of Poincar\'{e} and Minkowski spacetime coordinates.

However there exist a huge amount of other Heisenberg realizations of the $\U_h(\mathfrak{an}^3)\rtimes \U_{\mathfrak{io}(1,3)}[[h]]^{\rm DJ}$ DSR algebra. They lead to Heisenberg representations. Particularly important for further applications is the so-called non-covariant family of realizations which is labeled by two arbitrary (analytic) functions~$\psi$,~$\phi$. We shall write explicit form of all DSR algebra generators in terms of undeformed Weyl algebra $\mathfrak{W}^4[[h]]$-generators $(x^\mu, p_\nu)$. Before proceeding further let us introduce a convenient notation: for a given (analytic) function $f(t)=\sum f_n t^n \in\mathfrak{W}^4$ of one variable we shall denote by
\begin{gather}
\label{h-analytic} \tilde{f}=f(-hp_0)=\sum f_n (-1)^np_0^n h^{n}
\end{gather}
the corresponding element $\tilde f\in\mathfrak{W}^4[[h]]$.
Also we will use the following shortcuts:
\begin{gather}\label{Psi}
\Psi(t)=\exp \left(
\int_{0}^{t}\frac{dt'}{\psi (t')}\right) ,\qquad \Gamma(t) =\exp \left(
\int_{0}^{t}\frac{\gamma (t') dt'}{\psi (t')}\right)
\end{gather}
for an arbitrary choice of $\psi$, $\phi$ such that $\psi (0)=\phi (0)=1$ and $\gamma=1+(ln\phi)'\psi$ with notation $f'=\frac{df}{d(-hp_0)}$.

Now generators of  deformed Weyl algebra $\U_h(\mathfrak{an}^3)\rtimes \mathfrak{T}[[h]]^{\rm DJ}$ admit the following Heisenberg realization:
\begin{gather}\label{realization}
\X^{i}=x^{i}\tilde{\Gamma}\tilde{\Psi} ^{-1} ,\qquad \X^{0}=x^{0}\tilde{\psi}  -hx^{k}p_{k}\tilde{\gamma}
\end{gather}
together with
\begin{gather}\label{P_0}
P_{i}=p_{i}\tilde{\Gamma}^{-1} , \qquad
P_{0} =h^{-1}\frac{\tilde{\Psi}^{-1}-\tilde{\Psi}}{2} +\frac{1}{2}h\vec{p}\,{}^2 \tilde{\Psi}\tilde{\Gamma}^{-2}.
\end{gather}
Hermiticity of $\X^{\mu }$ requires $\psi ^{\prime }=-\frac{1}{3}\gamma $, i.e.~$\Gamma =\psi ^{-\frac{1}{3}}$.
The momenta (\ref{P_0}) are also called Dirac derivatives \cite{DJMTWW,MSSKG,BP,Meljanac0702215}. 

The remaining $\U_h(\mathfrak{an}^3)\rtimes \U_{\mathfrak{io}(1,3)}[[h]]^{\rm DJ}$ generators are
just rotations and Lorentzian  boosts:
\begin{gather}\label{M_i}
M_i = \imath\epsilon_{ijk}x_jp_k=\imath\epsilon_{ijk}\X_jP_k\tilde{\Psi},
\\
N_i =  x_{i}\tilde{\Gamma}\frac{\tilde{\Psi}^{-1}-\tilde{\Psi}}{2h}-x_{0}p _{i}\tilde{\psi}\tilde{\Psi}\tilde{\Gamma}^{-1} +
\frac{1}{2}hx_{i}\vec{p}\,{}^2\tilde{\Psi}\tilde{\Gamma}^{-1}-hx^{k}p_{k}p_{i}
\tilde{\gamma}\tilde{\Psi}\tilde{\Gamma}^{-1}\nonumber\\
\phantom{N_i}{} =(\X_iP_0-\X_0P_i)\tilde{\Psi}.\label{N_i}
\end{gather}
The deformed Casimir operator reads as\footnote{Notice that $\mathcal{C}_h=\sum\limits_{k=0}^\infty c_kh^k$ is a well-def\/ined formal power series with entries $c_k\in\U_{\mathfrak{io}(1,3)}$.}:
\begin{gather}\label{C}
\mathcal{C}_h=h^{-2}\big(\tilde{\Psi}^{-1}+\tilde{\Psi}-2\big)-\vec{p}\,{}^2\tilde{\Psi}\tilde{\Gamma}^{-2}.
\end{gather}

\begin{remark}
It is worth to notice that besides $\kappa$-Minkowski coordinates $\X^\mu$ one can also introduce usual (commuting) Minkowski - like coordinates $\tilde x^\mu\doteq \X^\mu \Psi$ which dif\/fer from $x^\mu$. The rotation and boost generators expressed above take then a familiar form:
 \begin{gather*}
  M_i = \imath\epsilon_{ijk}\tilde{x}_jP_k
\qquad\mbox{and}\qquad  N_i =(\tilde{x}_iP_0-\tilde{x}_0P_i).
\end{gather*}
\end{remark}

One can show that above realization (\ref{realization})--(\ref{N_i}) has proper classical limit:\\ $\X^{\mu } =x^\mu$, $M_{\mu \nu }=x_\mu p_\nu-x_\nu p_\mu$, $P_\mu=p_\mu$ as $h \rightarrow 0$ in terms of the canonical momentum and position $(x^\mu, p_\nu)$ generators (\ref{unWeyl}). In contrast the variables $(\tilde x^\mu, P_\nu)$ are not canonical. Moreover, phy\-si\-cally measurable frame is provided by  the canonical variables (\ref{unWeyl}), which is important in DSR theories interpretation. For discussion of the DSR phenomenology, see, e.g.,~\cite{ACdsrmyth,Liberati} and references therein.

\begin{remark}
It is important to notice that all twisted realizations (\ref{Atlh2}), (\ref{Jtlh}) from the previous section are special case of the one above (\ref{realization}) for special choice of the functions~$\psi$ and $\phi$ (or equivalently $\gamma$). More exactly, Abelian realization (\ref{Atlh2}) one gets taking constant functions $\psi=1$ and $\gamma=s$ (equivalently $\phi=e^{-h(s-1)p_0}$) and Hermiticity of $\hat{x}^0$ forces $\gamma=0$.
Jordanian realization~(\ref{Jtlh}) requires $\psi=1-hrp_0$ and $\gamma=0$ ($\phi=e^{hp_0}$) (cf.~\cite{BP} for details).
\end{remark}

\textbf{Representations versus realizations}\\
As an example of the above realizations let us consider representation of the Poincar\'{e}
Lie algebra in a Hilbert space $\mathfrak{h}$. In fact one has realization of the enveloping algebra $\U_{\mathfrak{io}(1,3)}$  in  space of linear operators over $\mathfrak{h}$, i.e. $\mathscr{L}(\mathfrak{h})$. This leads to homomorphism $\rho$ of the corresponding h-adic extension $\rho^h$:
 $\U_{\mathfrak{io}(1,3)}[[h]]^{\rm DJ}\shortrightarrow \mathscr{L}(%
\mathfrak{h})[[h]]$  (cf. Remark \ref{h-repr}). 
We speak about representations when linear space is determined and about realizations when we determine homomorphism from one algebra to another. Below we shall list few examples of Weyl algebra realizations. We will see in next section that true Hilbert space representation requires specification of the parameter $h$.
\begin{example}
Let us choose: $\Psi=\exp(-hp_0)$, $\Gamma=\exp(-hp_0)$ in the formulae (\ref{M_i}-\ref{C}). Then
the realization of the Poincar\'{e} Lie algebra has the form:
\begin{gather*}
\nonumber
M_{i}=\frac{1}{2}\epsilon _{ijm}(x_{j}p _{m}-x_{m}p_{j}) ,
\quad N_{i}={%
\frac{1}{2h}}x_{i} \big( e^{-2hp_{0}}-1\big) +
x_{0}p _{i}-\frac{\imath h}{2}x_{i}\vec{p}\,{}^2+\imath h x^{k}p_{k} p_{i},\\
P_{i}=p_{i}e^{hp_{0}}, \qquad P_{0}=h^{-1} \sinh (hp_{0})+\frac{h}{2}\vec{p}\,{}^{2}e^{hp_{0}}.
\end{gather*}
Moreover, one can easily see that the operators $(M_i, N_i, p_\mu)$ constitute the bicrossproduct basis \cite{MR}. Therefore, dispersion relation expressed in terms of the canonical momenta $p_\mu$ recovers the standard version of doubly special relativity theory (cf.\ formulae in \cite{AC,BACKG})
\begin{gather*}
\mathcal{C}_h=h^{-2}\big(e^{-{\frac{1}{2}}hp_{0}}-e^{{
\frac{1}{2}}hp_{0}}\big)^{2}-\vec{p}\,{}^2{e}^{hp_{0}}
\end{gather*}
and implies
\begin{gather*}
m^{2}=\left[2h^{-1}
\sinh \left(\frac{hp_{0}}{2}\right)\right]^{2}-\vec{p}\,{}^{2}e^{hp_{0}}.
\end{gather*}
\end{example}
One can notice that the above boost generators (\ref{N_i}) in this realization are not Hermitian. So we can consider another choice of $\Psi,\Gamma$.
\begin{example}
The Hermitian realization of the $\kappa$ -Poincar\'{e} algebra can be
determined as: \bea
M_{i}=\frac{1}{2}\epsilon _{ijm}(x_{j}p _{m}-x_{m}p_{j}) ,
\quad N_{i}={%
\frac{1 }{2h}}x_{i}\,\left( e^{-hp_{0}}-e^{hp_{0}}\right) +\imath x_{0}p _{i}e^{-hp_{0}}+ x_{i}\vec{p}\,{}^{2}
\,e^{-hp_{0}} 
\eea\bea
P_{i}=p_{i}, \qquad P_{0}=h^{-1} \sinh (hp_{0}
)+\frac{h}{2}\vec{p}^{2}e^{-hp_{0}} \label{P0}
\eea
and the dispersion relation is%
\begin{equation}
m_{0}^{2}=[2h^{-1} \sinh (\frac{hp_{0}}{2 })]^{2}-\vec{p}^{2}e^{-h
p_{0}}
\end{equation}
\end{example}
\begin{example}
In the minimal case, connected with Weyl- Poincar\'{e} algebra, in physical $n=4$  dimensions \cite{BP} the representation of the Poincar\'{e} algebra $%
(M_{i},N_{i},\P_{\mu })$ reads as \begin{equation}
M_{i}=-\frac{\imath }{2}\epsilon _{ijm}(x_{j}\partial _{m}-x_{m}\partial
_{j})
\end{equation}%
\begin{eqnarray}
N_{i} &=&-x_{i}\left[{\frac{p_{0}}{2}}(2+
hp_{0})-\vec{p}\,{}^{2}\right] \,\left( 1+hp_{0}\right)
^{-1} + x_{0}p_{i}
\end{eqnarray}
The 
deformed Casimir operator
\begin{equation}
\mathcal{C}_h=\frac{p_{0}^{2}-\vec{p}^{2}}{ 1+hp_{0}}
\end{equation}
leads to following dispersion relation
\begin{equation}
m_{0}^{2}\left( 1+hp_{0}\right) =p_{0}^{2}-\vec{p}^{2}
\end{equation}
which is not deformed for (free) massless particles.
\end{example}

\section{$\kappa$-Minkowski spacetime in q-analog version\\ and canonical DSR algebra}

In section 3.2 different from standard form of $\kappa$-Poincare algebra was introduced, and we have called it q-analog (canonical) version. In this section we shall f\/ind its $\kappa$-Minkowski counterpart. The corresponding (the so-called canonical) version of $\kappa$-Minkowski spacetime algebra seems to be very interesting from physical point of view.
In this model  $\kappa$-Minkowski algebra is a universal enveloping algebra of a solvable Lie algebra. This algebra has been also used in some recently postulated approaches to quantum gravity, e.g., Group Field Theory (see \cite{Oriti} and references therein).

In this version $\kappa$-Minkowski spacetime commutation relations have the Lie algebra form:
\be\label{kM}
[X^0,X^i]=\frac{\imath}{\kappa}X^i,\qquad [X^i,X^j]=0
\ee
with unshifted generators (see Remark \ref{shifted}): again the numerical factor $\kappa$ can be put to $1$ (the value of $\kappa$ is unessential) by re-scaling, i.e.\ the change of basis in the Lie algebra\footnote{It was not possible in $h$-adic case with the formal parameter $h=\kappa^{-1}$.}
$X^{0}\mapsto \kappa X^{0}$.
Relations (\ref{kM}) def\/ine four-dimensional solvable Lie algebra $\mathfrak{an}^3$ of rank~3, cf.\ Example~\ref{example8}.

\begin{remark}
The Lie algebra (\ref{kM}) turns out to be a Borel (i.e.\ maximal, solvable) subalgebra in de Sitter Lie algebra $\mathfrak{o}(1,4)$. This fact can  be seen via realization:
\be\nonumber
 X^0\AC M^{04}, \qquad X^i\AC M^{0i}-M^{i4}.
 \ee
and it has been explored in \cite{Oriti,Iwasawa}.
\end{remark}
Now we are in position to introduce the canonical (in terminology of~\cite{Kon}, cf.\ Remark~\ref{Poisson2}) $\kappa$-Minkowski spacetime $\M^4_\kappa=\mathcal{U}_{\mathfrak{an}^3}$ as a universal enveloping algebra of the solvable Lie algebra~$\mathfrak{an}^3$ studied before in~\cite{Sitarz,Oriti,Andrea}. In order to assure that $\M^4_\kappa$ is  $\U_\kappa(\mathfrak{io}(1,3))$-module algebra one has to check the consistency conditions (cf.\ Example~\ref{example2} and Remark~\ref{shifted}). Then utilizing the crossed product construction one obtains canonical DSR algebra $\M^4_\kappa\rtimes\U_\kappa(\mathfrak{io}(1,3))$.

\subsection*{Canonical DSR algebra}

The Hopf algebra $\U_\kappa(\mathfrak{io}(1,3))$ acts covariantly on $\M^4_\kappa$ by means of the classical action on the generators:
\begin{gather}
\label{ClassAction}
P_{i }\triangleright X^{\nu }=-\imath\delta _{i }^{\nu },\qquad
M_{\mu \nu }\triangleright X^{\rho }= \imath X_{\nu }\delta _{\mu }^{\rho }-
\imath X_{\mu }\delta _{\nu }^{\rho},\qquad
\Pi_0^{\pm 1}\triangleright X^\mu=X^\mu\mp\imath \kappa^{-1}\delta_0^\mu.
\end{gather}
As a result  cross-commutation relations determining the canonical   DSR algebra $\M^4_\kappa\rtimes\U_\kappa(\mathfrak{io}(1,3))$ take the following form:
\begin{equation}\label{MkM}
[ M_{i},X_{0}]=0,\quad [ N_{i},%
X_{0}]=-\imath X_{i}-\frac{\imath }{\kappa }N
_{i},  \quad
[ M_{i},X_{j}]=\imath \epsilon _{ijl}X
_{l},\quad
 [ N_{i},X_{j}]=-\imath \delta _{ij}
X_{0}+\frac{\imath }{\kappa}\epsilon _{ijl}M_{l},
\end{equation}
\begin{equation}\label{phsp1}
[ P_{k},X_{0}]=0,\qquad [ P_{k},
X_{j}]=-\imath \delta _{jk}\Pi_0,\qquad
[X_i, \Pi_0]=0 ,\qquad [X_0, \Pi_0]=-\frac{\imath}{\kappa}\Pi_0.
\end{equation}
It is easy to see again that dif\/ferent values of the parameter $\kappa$ give rise to isomorphic al\-geb\-ras. One should notice that the above ``$q$-analog'' version of DSR algebra cannot be pseudo-deformation type as the one introduced in the ``$h$-adic'' case.
 A subalgebra of special interest is, of course,
  a canonical $\kappa$-Weyl algebra $\mathfrak{W}^4_\kappa$ (canonical $\kappa$-phase space) determined by the relations~(\ref{kM}) and (\ref{phsp1}).
Heisenberg like realization of the Lorentz generators in terms of $\mathfrak{W}^4_\kappa$-generators can be found as well:
\begin{gather*}
 M_{i} =\epsilon _{ijk}\left[ X^{j}\Pi _{0}^{-1}+2\kappa^{-1}X^{0}P^{j}\big( \Pi _{0}^{2}+1-\kappa^{-2}\vec{P}^{2}\big) ^{-1}\right] P_{k},
\\
N_{i}=X^{0}P^{i}\frac{3-\Pi _{0}^{-2}\big(1-\kappa^{-2}
\vec{P}^{2}\big)}{\big( \Pi _{0}+\Pi _{0}^{-1}\big(1-
\kappa^{-2}\vec{P}^{2}\big)\big) }+\frac{\kappa}{2}X^{i}\big(
1-\Pi _{0}^{-2}\big(1-\kappa^{-2}\vec{P}^{2}\big)\big).
\end{gather*}
The deformed Casimir operator reads
\begin{gather}\label{qCas}
\mathcal{C}_{\kappa} =\kappa^{2}\big(\Pi _{0}+\Pi
_{0}^{-1}-2\big)-\vec{P}^{2}\Pi _{0}^{-1}.
\end{gather}
Moreover, one can notice that undeformed Weyl algebra $\mathfrak{W}^4$: $(x^\mu,p_\nu)$ is embedded in $\mathfrak{W}^4_\kappa$ via:
\begin{gather*}
  x^0=2X^0\big(\Pi_0+\Pi_0^{-1}\big(1-\kappa^{-2}\vec{P}^2\big)\big)^{-1},
\qquad
 x^i=X^i\Pi_0^{-1} +2\kappa^{-1}X^0\big(\Pi_0^2+1-\kappa^{-2}\vec{P}^2\big)^{-1}P_i,
\\
p_{0} =\frac{\kappa}{2}\big( \Pi _{0}-\Pi _{0}^{-1}\big(1-\kappa^{-2}\vec{P}^{2}\big)\big), \qquad p_i=P_i.
\end{gather*}
Finally we are in position to introduce Heisenberg representation of the canonical DSR algebra in the Hilbert space ${\cal L}^2(\mathbb{R}^4, dx^4)$. It is given by similar formulae as in the $h$-adic case (\ref{realization}), (\ref{M_i})--(\ref{C}). However here instead of $h$-adic extension $\tilde{f}$ of an analytic function $f$ (see (\ref{h-analytic})) one needs its Hilbert space operator realization
\[
\check{f}=\int f(t) dE_{(\kappa)}(t)
\]
 via spectral integral with the spectral measure $dE_\kappa (t)$ corresponding to the self-adjoint ope\-ra\-tor\footnote{Notice that $E_{(1)}=p_0$.}$E_{(\kappa)}=-\frac{\imath}{\kappa}\partial_0$, where $\kappa\in\mathbb{R}^*$.
Thus Hilbert space  representation of the canonical DSR algebra generators rewrites as (cf.~(\ref{realization}), (\ref{M_i}), (\ref{N_i})):
\begin{gather}\label{checkX}
X^{i}=x^{i}\check{\Gamma}\check{\Psi} ^{-1} ,\qquad X^{0}=x^{0}\check{\psi}  -hx^{k}p_{k}\check{\gamma},
\\
\label{checkP}
 \Pi_0=\check\Psi^{-1},\qquad
 P_{i}=p_{i}\check{\Gamma}^{-1},
\\
\label{checkM}
M_i = \imath\epsilon_{ijk}x_jp_k,
\\
\label{checkN}
N_i = \kappa x_{i}\check{\Gamma}\frac{\check{\Psi}^{-1}-\check{\Psi}}{2}-x_{0}p _{i}\check{\psi}\check{\Psi}\check{\Gamma}^{-1} +
\frac{\imath}{2\kappa}x_{i}\vec{p}\,{}^2\check{\Psi}\check{\Gamma}^{-1}
-\frac{1}{\kappa}x^{k}p_{k}p_{i}
\check{\gamma}\check{\Psi}\check{\Gamma}^{-1}
\end{gather}
with the Casimir operator:
\begin{gather}\label{checkC}
\mathcal{C}_\kappa=\kappa^2\big(\check{\Psi}^{-1}+\check{\Psi}-2\big)-\vec{p}\,{}^2\check{\Psi}
\check{\Gamma}^{-2}.
\end{gather}
Here $p_\mu=-\imath\partial_\mu$ and $x^\nu$ are self-adjoint operators acting in the Hilbert space ${\cal L}^2(\mathbb{R}^4, dx^4)$.
This leads to the St\"{u}ckelberg version of relativistic Quantum Mechanics (cf.~\cite{Menski,BP}).

\begin{remark}\label{PhaseSp}
    Alternatively, one can consider relativistic symplectic structure (cf.~(\ref{unWeyl}))
    \begin{gather}\label{RSS}
        \{x^\mu, x^\nu\}=\{p_\mu, p_\nu\}=0, \qquad \{x^\mu, p_\nu\}=\delta^\mu_\nu
    \end{gather}
 determined by the symplectic two-form $\omega=dx^\mu\wedge d p_\mu$ on the phase space $\mathbb{R}^4\times\mathbb{R}^4$. Now we can interpret formulas (\ref{M_i}), (\ref{N_i}) for the f\/ixed value $h={1\over\kappa}$ as a non-canonical transformation (change of variables) in the phase space. Thus in this new variables one gets $\kappa$-deformed phase space \cite{ncphasespace} with deformed Poisson brackets replacing the commutators in formulas~(\ref{L7}),~(\ref{L8}): $\{\;,\; \}_\kappa={1\over\imath}[\;,\; ]$. It corresponds to the so-called ``dequantization'' procedure~\cite{AF}. Conversely, the operators (\ref{checkX})--(\ref{checkC}) stand for true (Hilbert space) quantization of this deformed symplectic structure.
\end{remark}

\begin{example}
As a yet another example let us consider deformed phase space of Magueijo--Smolin model \cite{Smolin} (see also~\cite{ncphasespace}):
\begin{gather*}
 \{X^\mu, X^\nu\}=\frac{1}{\kappa}(a^\mu X^\nu-a^\nu X^\mu),\qquad
\{P_\mu, P_\nu\}=0, \qquad \{X^\mu, P_\nu\}=\delta^\mu_\nu+\frac{1}{\kappa}a^\mu P_\nu.
\end{gather*}
It corresponds to  the following change  of variable in the phase space~(\ref{RSS}):
\[
X^\mu=x^\mu-\frac{a^\mu}{\kappa}x^\nu p_\nu, \qquad P_\mu=p_\mu.
\]
 However we do not know twist realization for this algebra.
\end{example}

\textbf{Summary of Part I}\\
The first part of this thesis was concentrated on mathematical details of deformations for Minkowski spacetime with its symmetry groups.
 Firstly underlining how it may be obtained
by twist-deformation with $\U_{\mathfrak{igl}(n)}$ symmetry group, then insisting on two distinct versions ($h$-adic topology and $q$-analog) as a Hopf module to $\kappa$-Poincar\'{e} quantum groups. 
We started from reminding standard def\/initions and illustrative examples on crossed (smash) product construction between Hopf algebra and its module. We have also reminded few facts on twisted deformation and provided $\kappa$-Minkowski spacetime as well as smash algebra, with deformed phase space subalgebras, as a result of twists: Abelian and Jordanian. Moreover it was shown that the deformed (twisted) algebra is a pseudo-deformation of an undeformed one with above mentioned Jordanian and Abelian cases as explicit examples of this fact (Proposition~\ref{prop1}). The statement that the deformed smash (crossed) algebra can be obtained by a non-linear change of the generators from the undeformed one is an important result from the physical point of view, because it provides that one can always choose a physically measurable frame related with canonical commutation relations.

In the case of $\kappa$-deformation we have used smash product construction to obtain the so called DSR algebra, which uniquely unif\/ies $\kappa$-Minkowski spacetime coordinates with Poincar\'{e} generators.  Nevertheless it is only possible after either $h$-adic extension of universal enveloping algebra in $h$-adic case or introducing additional generator ($\Pi_0$) in $q$-anolog case. In our approach dif\/ferent DSR algebras have dif\/ferent physical consequences due to dif\/ferent realizations of $\kappa$-Minkowski spacetime. 

We have also introduced some realizations of deformed phase spaces (deformed Weyl algebra): $r$-deformed and $s$-deformed in Jordanian and Abelian cases respectively, and also $h$-adic and $q$-analog versions as well. Heisenberg representations in Hilbert space are also provided in all above cases. What is important in our approach that it is always possible to choose physical frame (physically measurable momenta and position) by undeformed Weyl algebra which makes clear physical interpretation~\cite{Liberati}. This implies that various realizations of DSR algebras are written in terms of the standard (undeformed) Weyl--Heisenberg algebra  which opens the way for quantum mechanical interpretation DSR theories in a more similar way to (proper-time) relativistic  (Stuckelberg version) Quantum Mechanics instead (in Hilbert space representations contexts\footnote{We believe that this work can be also helpful for building up a proper operator algebra formalism.}).

Now we are ready to discuss the possible physical interpretations of this framework. The dif\/ferent realizations of $\kappa$ -Minkowski spacetime shown in this Part of the thesis are responsible for dif\/ferent physical phenomena. In other words we obtain dif\/ferent physical predictions, e.g., on quantum gravity scale or time delays of high energy photons coming from gamma ray bursts depending on given realization. This point will be shown explicitly in the following.
%

\newpage
\thispagestyle{empty}


\newpage
 \thispagestyle{empty}
\vspace{2cm}

\begin{center}
    {\LARGE PART II}
\end{center}

\vspace{1.5cm}
\begin{center}
    
\textsc{ \Huge Applications in Planck scale physics }

\end{center}
\vspace{8.0cm}

%

\hspace{\stretch{1}} \begin{flushright}
    \emph{"There are three principal means of acquiring knowledge\\ available to us: observation of nature, reflection, and experimentation.\\ Observation collects facts; reflection combines them;\\ experimentation verifies the result of that combination.\\ Our observation of nature must be diligent,\\ our reflection profound, and our experiments exact.\\ We rarely see these three means combined;\\ and for this reason, creative geniuses are not common."}
\end{flushright}
\begin{flushright}
    Denis Diderot
\end{flushright}

\newpage
\newpage
\thispagestyle{empty}

\chapter{Possible experimental predictions from\\ $\kappa$ -Minkowski spacetime} 
Together with the development of the mathematical formalism standing behind noncommutative spacetimes, physical theories providing experimentally falsifiable features appeared. In 2001, it was proposed \cite{AC,BACKG} to generalize Special Relativity and add another invariant parameter besides speed of light. This generalization is now called Doubly (or deformed) Special Relativity (DSR). A different model of DSR was proposed later by Joao Magueijo and Lee Smolin. There exists strong motivation behind DSR theories because in the theory of Quantum Gravity a fundamental role is played by the Planck Mass (equivalently Planck Energy, Length), since this is the scale at which Quantum Gravity effects start to be comparable to the ones of the other interactions. Therefore all independent observers should observe this at the same scale. There is no  established mathematical framework behind the DSR idea \cite{ACdsrmyth} yet, however $\kappa$-Minkowski spacetime and $\kappa$-Poincare is one of the possibilities, and probably one of the most successful so far. In $\kappa$-deformation the parameter 'kappa' stays invariant, and $\kappa$-Minkowski spacetime becomes naturally connected with the DSR theory interpretation. Moreover given assumption that DSR is phenomenological limit of quantum gravity theory, $\kappa$-Minkowski spacetime formalism appears even more interesting. On top of that noncommutative quantum structure of spacetime as a result gives uncertainty relations for coordinates and preferred frame appears as property hence we have to deal with Lorentz-invariance violation (LIV), as in many other approaches to quantum gravity.

$\kappa$-Poincar\'{e} and $\kappa$-Minkowski algebras as one of the possible formalisms for DSR theories have been widely studied [10,11,13,14,18-20,70-72,78,82,84,88,90,92,94,102,103].
One can see that in this formalism the second invariant scale in DSR appears naturally from noncommutativity of coordinates, and has the meaning analogous to speed of light in Special Relativity. Nevertheless together with growth of popularity of this approach in DSR theories many critical remarks have appeared  \cite{DSRdebate,DSRdebateI,DSRkrytyka}. Some of the authors, using $\kappa$-Poincar\'{e} algebra formalism, have argued its equivalence to Special Relativity \cite{DSRkrytyka}. 
Deformed Special Relativity is not a complete theory yet, with many open problems.
Also recently the problem of nonlocality in varying speed of light theories appeared \cite{DSRdebate,DSRdebateI}. Fortunately, it has been shown that nonlocality problem is inapplicable to DSR framework based on $\kappa$-Poincar\'e~\cite{DSRnonloc}. Because of this, it seems to be very timely and interesting to deal with Hopf algebras and noncommutative spacetimes associated with them. Nevertheless here, in this thesis, we only mention DSR interpretation as one of the possibilities for phenomenology of $\kappa$-Minkowski spacetime, which is interesting and promising itself. Hence we think it is important to clarify and investigate in detail examples of noncommutative spacetimes from mathematical point of view (as it was done in Part I of this thesis) because various important technical aspects of $\kappa$-Minkowski spacetime were not always considered in the physical literature.
Nonetheless we would like to analyse physical predictions which follow from $\kappa$-Minkowski spacetime framework and this will constitute the second part of the thesis.

\section{Constraints on the quantum gravity scale}
Understanding the properties of matter and spacetime at the Planck scale
remains a major challenge in theoretical physics. In spite of several
theoretical candidates, the corresponding experimental data are very limited.
Among the data that is presently available, the difference in arrival time
of photons with different energies from GRB's, as observed by the Fermi
gamma ray telescope \cite{Fermi,ref3,recent}, may contain information about the
structure of spacetime at the Planck scale \cite{ellis1,ellis2}, see \cite%
{smolin1} for a recent review. There have been attempts to analyze this data
from the GRB's based on the framework of doubly special relativity (DSR)
\cite{AC}, where it was argued that the dispersion relation following
from DSR is consistent with the difference in arrival time of photons with
different energies. Similar dispersion relations have also been proposed in
other scenarios, see e.g. \cite{Smolin,mattingly,girelli}, most of which lead to Lorentz symmetry
violation.

In a related development it was found that the dispersion relations
following from the $\kappa$-Minkowski spacetime can be used to analyze the astrophysical data from the GRB \cite{majid1}. In addition, one possible description of the symmetry structure of the DSR is through the Hopf-algebra of symmetry of the $\kappa$-Minkowski spacetime
\cite{AC,Bu,BP2},\cite{glik1}-\cite{ref1}, although alternative descriptions are also possible (see e.g. \cite{ref1,ref2}).
The noncommutative geometry defined by the $\kappa$-Minkowski spacetime can also
be obtained from the combined analysis of special relativity and the quantum
uncertainty principle \cite{Dop1,Dop2} (the so-called DFR model). These observations indicate
that the $\kappa$-Minkowski spacetime could be a possible candidate to
describe certain aspects of physics at the Planck scale.

Recently it has been emphasized that the dispersion relation compatible with
the GRB data is likely to arise from a DSR model where the transformation
laws are changed but the Lorentz symmetry is kept undeformed \cite{smolin1}. Assuming that the DSR formalism can be described by the $\kappa$-Minkowski spacetime, this implies that the corresponding $\kappa$-Poincar\'e Hopf algebra should have an undeformed Lorentz algebra but may have a deformed co-algebra. Therefore we identify deformation parameter $\kappa$ with quantum gravity scale since it appears  as invariant scale in DSR theory interpretation, dispersion relations become deformed and quantum gravity corrections appears.

In this Part we shall quickly remind the description of the deformed dispersion
relations arising from $\kappa$-Minkowski spacetime. 
We consider two equivalent types of dispersion relations which are characterized by a set of parameters and 
can be determined by the choice of the realization. 
 Although the parameters in the dispersion
relations can be calculated for any given choice of the realization, here we
take the point of view that they should be determined using the empirical
data from the astrophysical sources. Identifying deformation parameter with quantum gravity scale, experimental data 
\cite{Fermi,ref3,recent} gives the relation $\frac{M_Q}{M_{Pl}}>1.2$, which indicates the possibility of truly Planckian effects. In our framework however it is possible to even get
$M_Q=M_{Pl}$, due to proper choice of proportionality coefficient. Besides leading term - proportional to
$\frac{1}{M_{Pl}}$, we consider also quadratic corrections, which might provide better fit to astronomical data, and have higher accuracy on bounds of quantum gravity scale.

Later, we will discuss the generalized dispersion relations that follow
from our analysis of the $\kappa$-Minkowski algebra and we
introduce all the realizations for $\kappa$-Minkowski spacetime coordinates and provide comparison with recent experimental data which allows us to obtain bounds on quantum gravity scale or model parameters.

\section{Dispersion relations}

Considering the $\kappa$-deformed Minkowski spacetime provided with noncommutative
coordinates $\X_\mu$ (or $X_{\mu}$ consistenlty with the notation from previous chapter)\footnote{We will keep refering to two versions from previous chapter $h$-adic and q-analog ones, even if we have argued that only that last one should be considered in physical theories. Nevertheless mainly it was $h$-adic version which was used by many authors for that purpose.} ($\mu = 0,1,..,n-1 $), which satisfy Lie-algebra type of commutation relations (\ref{kMcal}) (or equivalently (\ref{kM})),
one can notice that, due to Jacobi identity, canonical commutation relations: $
[P_{\mu },X_{\nu }]=\eta _{\mu \nu }$ are no longer satisfied. 
In fact one has instead deformed Weyl algebra as it was shown in previous chapter (\ref{L7}-\ref{L8})  ((\ref{phsp1}) respectively).


These relations modify  Heisenberg uncertainty  relations (see e.g. \cite{Liberati} and references therein). It appears that interesting class of representations of the DSR algebra
(\ref{L1}--\ref{L3}) and (\ref{kMcal}), (\ref{L5}--\ref{L8}); (equivalently (\ref{kM}), (\ref{MkM}--\ref{phsp1})) can be induced from representations of the deformed spacetime algebra itself (\ref{kMcal}) (or (\ref{kM})) (see \cite{s1}, \cite{MS2}, \cite{s5} for details).
As it was mentioned before the name 'DSR algebra' appears because of it's interpretation.
Different realizations lead to different doubly (or deformed) special relativity models with
different physics encoded in deformed dispersion relations. 
Let us remind this point in more detail. In terms of momenta $P_{\mu }$, we have
undeformed dispersion relations given by standard Poincar\'{e} Casimir
operator (\ref{nondef}).
But since the standard Casimir operator didn't satisfy all the needed conditions
we have found another invariant
operator - deformed Casimir operator $\mathcal{C}_h$ (\ref{C})  ($\mathcal{C_{\kappa }}$ (\ref{qCas})) which lead to  deformed dispersion relation
with the deformed mass parameter $m_{\kappa }$. It turns out that in the most
concrete realizations the interrelation between this two invariants has the
following form \cite{BP,BGMP,s1}
:\be P^{2}=\mathcal{C_{\kappa }}(1+\frac{a^{2}}{4}\mathcal{C_{\kappa }})\ee
Particularly \be m_{ph}^{2}=m_{\kappa }^{2}(1+\frac{a^{2}}{4}m_{\kappa }^{2})\ee
Therefore for photons: $m_{ph}=m_{\kappa }=0$ and as a consequence dispersion relations obtained from (\ref{nondef}) and (\ref{def}) are identical.
One can however see that in general both expressions  have
the same classical limit $a_\mu \rightarrow 0$ but differ by order as
polynomials in $a$. It happens that both dispersion relations (\ref{nondef},\ref{def}) can
be rearranged into convenient form:
\begin{equation}  \label{disp1}
E \simeq |{\vec{p}}|\frac{G_{1}}{G_{2}}+\frac{m^{2}}{2|{\vec{p}}|}
\end{equation}
\begin{equation}  \label{disp2}
\frac{G_{1}}{G_{2}} \simeq 1+c_{1}\frac{|{\vec{p}}|}{M_{Q}}+c_{2}\frac{|{%
\vec{p}}|^{2}}{M_{Q}^{2}}
\end{equation}
if one restricts oneself only to the second order accuracy and introduces locally measurable momenta $(E,\vec{p})$.
Here $G_{1}$ and $G_{2}$ are model-dependent functions of $\frac{|\vec{p}|}{M_{Q}%
}$ and $\frac{m^{2}}{M_{Q}^{2}}$ satisfying the conditions $%
G_{1}(0)=G_{2}(0)=1$, and $c_{1}$ and $c_{2}$ are model-dependent constants.
Taking $p_{0}=E$ and $a^{\mu}=(M_{Q}^{-1},0,0,0)$ as timelike vector so that $a^{2}=-M_{Q}^{-2}<0$ where $M_{Q}$ is the quantum gravity scale and assuming that $m<<E<<M_{Q}$, we have energy dependence up to second order in $\frac{1}{M_{Q}^{2}}$.
Eqs. (\ref{disp1}) and (\ref{disp2}) provide the general dispersion
relations within the $\kappa$-Minkowski algebra framework and up to the
second order in $\frac{1}{M_Q^2}$. These constants, $c_1$ and $c_2$, are
additional parameters appearing in the dispersion relation and time delay
formula. It has been observed in \cite{smolin1} that a proper analysis of
the GRB data using dispersion relations may require more than just the
parameter given by the quantum gravity scale $M_Q$. In fact, in order to diversify between different DSR models (see below) one has to assume some numerical value for the parameter $M_Q$ (e.g. $M_Q=M_{Pl}=1,2\times 10^{19} GeV$) and fit the values for linear $c_1$ as well as quadratic $c_2$ corrections \cite{Albert,Aharonian}. Otherwise one can choose parameters $c_1$ and $c_2$ from quantum gravity model and experimentaly fit quantum gravity scale $M_Q$.

 In our formalism, the
parameters $c_1$ and $c_2$ arise naturally in the dispersion relations,
which can be obtained by fitting with the empirical data. They can also be
calculated theoretically once the choice of the realization of the $\kappa$%
-Minkowski spacetime is fixed. In this sense, the parameters $c_1$ and $%
c_2$ incorporate the information of the quantum gravity vacuum, analogous to
certain fuzzy descriptions of quantum gravity \cite{f1,f2}.\newline
Alternatively one can consider dispersion relation in the form:
\begin{equation}  \label{disp3}
|\vec{p}|\simeq E \left(1-b_1\frac{E}{M_Q}+b_2\frac{E^2}{M^2_Q}\right)
\end{equation}
which coincides with (\ref{disp1}) and (\ref{disp2}) up to second order accuracy
provided $c_1=b_1$ and $c_2=2b_1^2-b_2$, i.e.\be
E\simeq |\vec{p}| \left(1+b_1\frac{|{\vec{p}}|}{M_Q}+(2b_1^2-b_2)\frac{|{%
\vec{p}}|^2}{M^2_Q}\right)\ee
This gives time delay formula as: 
\be  \label{td}
\Delta t\simeq -\frac{l}{c}\frac{|\vec{p}|}{M_Q}\left(B_1+\frac{|\vec{p}|}{M_Q%
}B_2\right) = -\frac{l}{c}\frac{E}{M_Q}\left(2b_1-3b_2\frac{E}{M_Q}\right)
\ee
where $B_1=2c_1=2b_1,\quad B_2=c_2-4c_1^2=2b_1^2-3b_2$ and $l$ is a distance
from the source of high energy photons. This last equation might  be
more suitable for our purposes since photon energy is well measurable
physical quantity. The case $M_Q=M_{Pl}$ is not excluded, thus natural relation between quantum gravity scale and Planck scale arises here (see below for details).
 In general, due to Lorenz Invariance Violation (LIV): $|\vec{p}|\neq E$, the second order contributions to "momentum" and "energy" versions of the time delay formula (\ref{td}) are different. However, for $b_1=0$, one has $B_2=-3b_2$ and both contributions are the same.
The equation (\ref{td}) describes absolute time delay between
deformed and undeformed cases. What one really needs, in order to compare against experimental data, is the relative time delay i.e. time-lag between two photons with different energies $E_{l}<E_{h}; \delta E=E_{l}-E_{h}; \delta E^2=E^2_{l}-E^2_{h}$ (with lower energy photons $E_l$ arriving
earlier).
\begin{equation}
\Delta \delta t\simeq - \frac{l}{c}\left( 2b_{1}\frac{\delta E}{M_{Q}}-3b_{2}%
\frac{\delta E^{2}}{M_{Q}^{2}}\right)
\end{equation}%
Following the lines of thoughts presented in \cite{jacob,xiao} one can also take into account
the universe cosmological expansion: for the photons coming from a redshift $z$ one gets
\begin{equation}  \label{cosm}
\Delta \delta t\simeq -2b_1\frac{\delta E_e}{M_{Q}}%
\int_{0}^{z}\frac{1+z^{\prime }}{h(z^{\prime })}dz^{\prime }+3b_2\frac{%
\delta E_e^{2}}{M_{Q}^{2}}\int_{0}^{z}\frac{(1+z^{\prime })^{2}}{%
h(z^{\prime })}dz^{\prime }
\end{equation}%
where: $h(z^{\prime })=H_{0}\sqrt{\Omega _{\Lambda }+\Omega _{M}(1+z)^{3}}$
and $H_{0}=71/s/Mpc$ is Hubble parameter, $\Omega _{M}=0.27$ - the matter density and $%
\Omega _{\Lambda }=0.73$ - vacuum energy density are cosmological parameters represented by their present day values:
$E_e$ denotes redshifted photon's energy as measured on earth.

\section{Time delay formulae for different realizations of DSR algebra}

The DSR algebra (\ref{L1}--\ref{L3}) and (\ref{kMcal}), (\ref{L5}--\ref{L8}) (equivalently (\ref{kM}), (\ref{MkM}--\ref{phsp1})) admits a wide range of realizations, which can be written in the general form as:\be\label{XMPreal}
\X_{\mu } =x^{\alpha }\phi _{\alpha \mu }(p; M_Q ), \quad
M_{\mu \nu } =x_{\alpha }\Gamma _{\mu \nu }^{\alpha }(p; M_Q ),\nonumber\ee
\be P_{\mu } =\Lambda _{\mu }(p; M_Q )\ee
\footnote{Please note that in the case of q-analog we are able to use compact notation $P_\mu$ since $P_0$ is expressed by generators $\Pi_0$ and $\Pi_0^{-1}$, defined by the formula (\ref{qP_0}).}
in terms of (undeformed) Weyl algebra \footnote{For details on relation between deformed and undeformed Weyl algebra see Chapter 1 (\cite{BP4}), particularly Proposition \ref{prop1} and \ref{prop3}} satisfying canonical commutation relations:\be
\left[ p_{\mu },x_{\nu }\right]=-i\eta _{\mu \nu }, \quad
\left[ x_{\mu },x_{\nu }\right]=0, \quad
\left[ p_{\mu },p_{\nu }\right]=0.
\ee
with proper classical limit: $\X_{\mu } =x_\mu,\quad M_{\mu \nu }=x_\mu p_\nu-x_\nu p_\mu,\quad P_\mu=p_\mu$ as $M_Q\rightarrow\infty$.
In these realizations, $x_{\mu }$
and $p_{\nu }$ do not transform as vectors under the action of $%
M_{\mu \nu }$ (\ref{XMPreal}). They provide (commuting) position $x^\mu$ operators and local, physically measurable, momentum $p_{\mu }=-i\partial_{\mu }$ which define measurable frame  for DSR theories, for discussion of DSR phenomenology, see \cite{Liberati} and references therein.\newline

\textbf{Noncovariant realizations} (\ref{realization}) cover a huge family of DSR-type  models.
One can rewrite them introducing quantum gravity scale $M_Q$ in the following form:
\be \X^{i}=x^{i}\phi (A),\quad \X^{0}=x^{0}\psi (A)+\imath ax^{k}\partial _{k}\gamma (A)\ee
where $A=-ap=-\frac{E}{M_{Q}}$ and $\phi (A)=\exp \left( \int_{0}^{A}\frac{%
(\gamma (A^{\prime })-1)dA^{\prime }}{\psi (A^{\prime })}\right)$; $\psi (0)=1$, keeping notation from Section 4.1. for  $\psi ,\phi $ as two arbitrary analytic functions (cf.\ref{Psi}) \cite{BP,s1}.
Photon's dispersion relations with respect to both undeformed and deformed Casimir operators (\ref{nondef}, \ref{def})
can be recast as follows:
\begin{equation}\label{a2}
|\vec{p}|=E\frac{1-\exp \left( -\int_{0}^{A}\frac{%
dA^{\prime }}{\psi (A^{\prime })}\right)}{A}
\exp \left( \int_{0}^{A}\frac{
\gamma (A^{\prime })dA^{\prime }}{\psi (A^{\prime })}\right)
\end{equation}
Recent experimental data disfavour models with only linear accuracy.
In order to calculate second order contribution
stemming from (\ref{a2}) one needs the following expansion:
\begin{equation}
\psi(A) =1+\psi _{1}A+\psi _{2}A^{2}+o(A^{3});\quad \gamma(A) =\gamma _{0}+\gamma
_{1}A+o(A^{2})
\end{equation}%
This provides general formulae for the coefficients $b_{1},b_{2}$ in (\ref%
{disp3}) and $B_{1},B_{2}$ in (\ref{td}):
\be b_{1}=\frac{1}{2}(2\gamma _{0}-1-\psi _{1})\ee
\be b_{2} =\frac{1}{6}(1+3\psi _{1}+2\psi _{1}^{2}-\psi _{2}+3\gamma
_{0}^{2}-3\gamma _{0}+3\gamma _{1}-6\gamma _{0}\psi _{1})\ee
\be B_{1} =2\gamma _{0}-1-\psi _{1};\ee
\be B_{2} =\frac{1}{2}(\gamma
_{0}^{2}-\psi _{1}^{2}-\gamma _{0}+2\psi _{1}\gamma _{0}-\psi _{1}+\psi
_{2}-3\gamma _{1})\ee
General case of Hermitian (Hilbert space) realization \cite{BP} requires (in physical dimension four): $\psi ^{\prime
}+3\gamma=0$ which correspond to: $\psi _{1}=-3\gamma _{0};\psi _{2}=-\frac{3%
}{2}\gamma _{1}$ and give rise to formulas:
\be
b_{1}=\frac{1}{2}(5\gamma _{0}-1)\ee
\be b_{2}=\frac{1}{6}(1+39\gamma _{0}^{2}-12\gamma _{0}+\frac{9}{2}\gamma _{1})\ee
\be B_{1}=5\gamma _{0}-1;\quad B_{2}=\frac{1}{2}(4\gamma _{0}^{2}+2\gamma _{0}-%
\frac{9}{2}\gamma _{1})\ee
Particularly for $\gamma _{0}=\frac{1}{5}$ one can reach $b_{1}=0$ and
$ b_{2}=\frac{3}{4}\gamma _{1}+\frac{2}{75}$.\\
a) Non-covariant realizations contain representations generated by
\textbf{Jordanian one-parameter family of Drinfeld twists} (for details see
\cite{BP}): $\psi =1+rA, \gamma=0$, hence $\psi _{1}=r,\psi_2 =\gamma
_{0}=\gamma _{1}=0$.\newline
The corresponding time delay coefficients for photons are:\be
b_{1}=-\frac{1}{2}(1+r);\quad b_{2}=\frac{1}{6}(1+3r+2r^{2});\quad
B_{2}=-\frac{1}{2}r(r+1)\ee
so one gets upper-bound $B_2\leq -\frac{3}{8} $. However in (\ref{disp3}) one gets $b_2 \geq -\frac{1}{4}$ which provides lower-bound.
Particularly, for $r=-1$ one gets $b_{1}=0=b_2$. This
corresponds to Poincar\'{e}-Weyl algebra (Chapter 2, \cite{BP}) and provides no time delay for photons.
A particular Jordanian Hermitian  case in $n=4$ dimensions requires $r=3$ and gives \be
\Delta t=\frac{l}{c}\frac{|\vec{p}|}{M_{Q}}\left( 4+3\frac{|\vec{p}|}{M_{Q}}%
\right)\ee  and leads to "time advance".
Another case of dispersion relation worth to consider:
\be |\vec{p}|=\frac{E}{1-\frac{E}{M_Q}} \ee
is recovered here for parameter $r=1$, and one gets $b_1=-1; b_2=1$ in time delay formulas. This case differs by sign from dispersion relation considered in \cite{Smolin}.
Nevertheless original formulae proposed there:
\be \frac{E^2}{(1+\frac{E}{M_Q})^2}-|\vec{p}|^2=m^2\ee
 can be also reconstruct within this formalism for the following choice:\\ $\psi=1-3A+2A^2;\gamma=0$ and one gets $b_1=b_2=1$ in time delay for photons. However we still do not know if this realization is possible to obtain by Drinfeld twist.\\
b) $\kappa -$Minkowski spacetime can be also implemented by
one-parameter family of \textbf{Abelian twists} \cite{Bu,BP,lee2}.
Abelian twists give rise to : $\psi =1, \gamma =s= \gamma _{0}, \gamma _{1}=\psi_1=\psi_2=0$ and
\be b_{1}=\frac{1}{2}(2s-1);\quad b_{2}=\frac{1}{6}(3s^{2}-3s+1)\ee
\be \Delta t=-\frac{l}{c}\frac{|\vec{p}|}{M_{Q}}\left( 2s-1+\frac{|\vec{p}|}{%
2M_{Q}}s(s-1)\right)\ee
so $B_2\geq \frac{1}{8}$ provides bounds from below ($s=\frac{1}{2}$). And analogously one obtains an upper-bound for $b_2\leq \frac{1}{6}$ in (\ref{disp3}).
The case $b_{1}=0$ gives  $\Delta t\simeq -\frac{l}{c}\frac{E^{2}}{8M_{Q}^{2}}$. Moreover Hermitian realization requires $s=0$ and provides $\Delta t=+\frac{l}{c}\frac{|\vec{p}|}{M_Q}$, "time advance" instead. (The corresponding dispersion relation has been also found in \cite{LRZ}.)
As it has been shown in \cite{BP} the case $s=1$
reproduces the standard DSR theory \cite{AC}-\cite{Smolin}, \cite{glik1}-\cite{glik3} which is related with the so-called bicrossproduct basis
\cite{MR}. The time delay formula taking into account the second order contribution reads now as
\be\label{dsr}
\Delta t=-\frac{l}{c}\frac{|\vec{p}|}{M_{Q}}=-\frac{l}{c}\frac{E}{M_{Q}}\left(1-\frac{E}{2M_{Q}}\right)\ee
Assuming $M_Q=M_{Pl}$, the leading term of (\ref{td}) is $\frac{2b_1}{M_{Pl}}$, which compared with recent results: $\frac{M_{Q}}{M_{Pl}}>1.2$, gives $b_1<0,417$. Particularly for Jordanian realizations one obtains lower bound for parameter $r>-1.208$, analogously for Abelian realizations $s<0.604$ which provides upper-bound.
Finally, it might be of some interest in physics to study models with first order corrections
vanishing, i.e. $b_{1}=0$. This is due to the fact that it is unlikely to observe LIV at linear order of $\frac{E}{M_{Q}}$.
This yields: $ \psi _{1}=2\gamma _{0}-1$\be
b_{2}=\frac{1}{6}(\gamma _{0} -\psi _{2}+3\gamma _{1}-\gamma _{0}^{2})\ee
\be B_{2}=-3b_{2}=-\frac{1}{2}(\gamma _{0} -\psi _{2}+3\gamma _{1}-\gamma
_{0}^{2})\ee
\textbf{General covariant realization} \cite{MS2}.
One can distinguish interesting case of covariant realizations of
noncommutative coordinates:\newline
\be \X_{\mu }=x_{\mu }\phi +i(ax)\partial _{\mu }+i(x\partial )(a_{\mu
}\gamma _{1}+ia^{2}\partial _{\mu }\gamma _{2})\ee
with dispersion relation \be
\ m_{ph}^{2}=\frac{E^{2}-\vec{p}^{2}}{(\phi -\frac{E}{M_Q})^{2}-\frac{%
\vec{p}^{2}}{M^2_Q}}\ee
with respect to (\ref{nondef}), which does not yield time delay for photons.
It is worth noticing however that for special 
choice of $\phi =1$ the last formula recovers Magueijo-Smolin type of covariant dispersion relations (DSR2),
see Ref. \cite{Smolin} formula (3).\newline
One can consider second Casimir operator in this realization:
$$\mathcal{C_{\kappa }}=\frac{2}{M_Q^{2}}\left( \sqrt{1+M_Q^{2}\frac{p^{2}%
}{\left( \phi +\frac{p_0}{M_Q}\right) ^{2}+\frac{p^{2}}{M^2_Q}}}-1\right)$$
but there is no time delay either.\newline

\textbf{Comparison against experimental data}\\
All the above results can be directly translated and compared with recent numerical data obtained by MAGIC, FERMI or other space or ground experiments. In the following we shall compare data given in \cite{Albert,Aharonian} within our framework.
Our formula for time delay (\ref{td}) coincides with the one introduced in \cite{Aharonian} for change of speed of light due to quantum gravity corrections.
\be\label{c}
c^{\prime }=c\left( 1+\xi \frac{E}{M_{Q}}+\zeta \frac{E^{2}}{M_{Q}^{2}}%
\right)=c\left( 1+2b_{1}\frac{E}{M_{Q}}+(4b_{1}^{2}-3b_{2}) \frac{E^{2}}{M_{Q}^{2}}
\right)\ee

First let us consider model-independent limits on the quantum gravity scale based on astrophysical data. Let us assume, as in \cite{Aharonian}, (in general it is not necessary in our models as mentioned in previous section) that quantum gravity scale is Planck mass: $M_{Q}=E_{P}=1.22\times 10^{19}GeV$. For linear corrections in speed of photons (with quadratic corrections vanishing $\zeta=0$, or equivalently $b_2=\frac{4}{3}b_1^2$ in (\ref{c})) one obtains the following bounds for coefficients $b_1,b_2$ in time delay formula (\ref{td}):\\
- for GRB's $|b_{1}|<35-75$;\\
-for active galaxies (Whipple collaboration during flare at 1996) $|b_{1}|<100$;\\
-MAGIC (2005) $|b_{1}|\approx 15$;
with limits $ |b_{1}|<30$;\\
-and if $c^{\prime }$ is helicity dependent $ |b_{1}|<0.5\times 10^{-7}$.\\
It should be noticed that results simultaneously constraining linear and quadratic corrections are very rough. e.g. $|\xi|<60$ and $|\zeta|<2.2\times 10^{17}$ \cite{Albert,Aharonian} gives rise to $|b_1|<30$ and $1200-0.73\times 10^{17}<|b_2|<1200+0.73\times 10^{17}$.
More experimental data one can find \cite{Fermi,Albert,Aharonian} and in references therein.

Moreover, one can consider model-dependent limits on the quantum gravity scale, since coefficients $b_1,b_2$ are connected with above introduced, twist realizations of $\kappa$ -Minkowski spacetime. We can choose the type of model (e.g. Hermitian Jordanian, Abelian) and provide the bounds for quantum gravity scale without $M_Q=M_P$ assumption.\\ Assuming
\textbf{Jordanian} realization in Hermitian case (with $\zeta=0$) the quantum gravity scale bounds are:\\
$M_{Q} >28\times 10^{17} GeV$ (for the limit obtained by MCCF method); \\
$M_{Q}>20,8\times 10^{17}GeV$ (from wavelet analysis). \\
Taking into account only quadratic corrections in this case we do not obtain any bound for $M_Q$ scale due to $r=-1$ case discussed in previous section ($\xi=0\Rightarrow b_1=0\Rightarrow b_2=0$).
Analogously one can do this analysis for
\textbf{Abelian twists}.
For linear correction in
                                                 \textbf{Abelian Hermitian} realization we obtain the following bounds on quantum gravity scale: $M_{Q}>7.2\times 10^{17} GeV $ (the limit obtained by MCCF method). Limit obtained from wavelet analysis: $M_{Q}>5.2\times 10^{17} GeV$.
                                                                                                                                                                                                                                                                                                    \textbf{Abelian DSR} realization provides exactly the same bounds.
However considering only quadratic corrections in Abelian, no longer Hermitian neither DSR case, ($\xi=0\Leftrightarrow s=\frac{1}{2}$) one obtains: $ M_{Q}>1.31\times 10^{9} GeV$ (within the MCCF method).
One can notice that all bounds on quantum gravity scale are lower bounds.

\textbf{Summary of Part II}\\
We have shown that quadratic approximation of dispersion relations arising from our analysis of the $\kappa$-Minkowski spacetime contain multiple parameters, which depend on the choice of the realization and deformation parameter which is identified with quantum gravity scale. While the pure linear suppression by the quantum gravity scale is theoretically allowed, its exact form
is not yet established from the recent analysis of the GRB data \cite{Fermi,ref3,recent}. However, our analysis predicts a more general form of the in vacuo dispersion relations (\ref{disp3}), which typically contain terms with both linearly and quadratically suppressed by the quantum gravity scale $M_Q$. Moreover, the dispersion relations obtained here contain parameters which depend on the choice of the realizations of the $\kappa$-Minkowski algebra. A priori there is no basis to prefer one realization over another, which should be an empirical issue. However one natural choice appears to be hermitian ones. Nevertheless standard, based on Doubly Special Relativity approach \cite{BACKG,LukDSR,smolin1,DSRnonloc,JKG} to this problem prefers just one specific realization and the traditional form of dispersion relations related to bicrossproduct one \cite{MR}. We shall show that also classical basis satisfies bicrossproduct construction (see Apendix \ref{app2}).

It would therefore be best to obtain all the parameters in the dispersion relations from a multi parameter fit of the GRB data. We believe that this should be possible with the increased availability of the astrophysical data. However we obtain some bounds on quantum gravity scale following from given realizations. And it was shown that our framework makes possible to retain that Planck scale is quantum gravity scale.
The above arguments suggest that the $\kappa$-Minkowski algebra related with $\kappa$ -Poincar\'{e} Hopf algebra might be able to capture certain aspects of the physics at the Planck scale, which is compatible with the
claim that noncommutative geometry arises from a combined analysis of
special relativity and quantum uncertainty principle \cite{Dop1,Dop2}.
It is well known that the $\kappa$-Minkowski spacetime leads to modification
of particle statistics and leads to deformed oscillator algebras \cite{kappaSt,s5}.

\chapter*{Conclusions and perspectives}
\addcontentsline{toc}{chapter}{Conclusions and perspectives}
\fancyhf{} 
\fancyhead[LE,RO]{\bfseries\thepage}
\fancyhead[RE]{Conclusions}
In the past years, our knowledge on the universe has become more and more detailed thanks to new advanced technologies which provide us with many precise experimental astrophysical data, which in turn require understanding and explanation. Furthermore they place strong bounds on existing theoretical models or they need a theoretical explanation in the first place. One of the most intriguing problems is the very origin of our Universe, which is connected with Planck scale physics.

Noncommutative spacetimes are considered as one of the candidates for the description of Planck scale physics. The topic of this thesis was mainly focused on the example of such noncommutative spacetime which is $\kappa$- Minkowski spacetime. The deformation of spacetime requires a generalisation of its symmetry group and in this case one deals with the $\kappa$-Poincare Quantum Group as well as with twist-deformations of inhomogeneous linear group. In the result quantum noncommutative spacetime stays invariant under the Quantum Group of transformations in analogy to the classical case. 
The deformation parameter 'kappa' can be understood as 'quantum gravity' scale which was the main point in application of $\kappa$- Minkowski spacetime in possible physical predictions.
The parameter stays invariant and therefore makes $\kappa$-Minkowski spacetime naturally connected to Deformed (doubly) Special Relativity theories (DSR).

\textbf{Summary of results}\\
The mathematical techniques of noncommutative geometry such as deformation theory and Hopf algebra formalism involved in the investigation of $\kappa$- Minkowski spacetime properties were gathered in Part I of this thesis. 
Different constructions of the $\kappa$-Minkowski noncommutative spacetime were introduced as an example of Hopf module algebra with $\kappa$-Poincare as its quantum symmetry group (Hopf algebra). It is worth to mention that the $\kappa$-Minkowski noncommutative spacetime can be extended to the full DFR (Doplicher-Fredenhagen-Roberts) formalism \cite{LD-GP}.\\
In this thesis two families of twists, providing $\kappa$-Minkowski spacetime, were considered: Abelian and Jordanian. The module algebra of functions over the universal enveloping algebra of inhomogeneous general linear algebra is equipped in star-product and can be represented by star-commutators of noncommuting ($\kappa$-Minkowski) coordinates. This approach is connected with the Drinfeld twist deformation technique applied to inhomogeneous general linear algebra from one hand. On the other hand it is related with Kontsevich idea of deformation quantization of Poisson manifolds, the so-called star product quantization. The deformation of the minimal Weyl-Poincare algebra was introduced explicitly. In contrast, Drinfeld-Jimbo deformed Poincare algebra was also described, together with the deformed d'Alembert operator, which implied a deformation on dispersion relations. The corresponding family of realizations of $\kappa$-Minkowski spacetime as quantum spaces over deformed Poincare algebra was analyzed (the so-called noncovariant realizations). Moreover, two families of deformed dispersion relations (Abelian and Jordanian) were considered in detail. In the minimal Weyl-Poincare algebra case the correction to mass is linear in the energy. These deformed dispersion relations were a starting point to calculate "time delays" for high energy photons arriving from gamma ray bursts (GRB's) as example of possibility of quantum gravity effects measurement. As one of the mathematical frameworks for the observed "time delay" one can consider for example spacetime noncommutativity. The general framework describing modified dispersion relations and time delays with respect to different noncommutative $\kappa$-Minkowski spacetime realizations was proposed.\\
The other area of the research  introduce in this thesis was more focused on the mathematical properties of the quantum $\kappa$-Poincare group. We pushed forward the idea that the $\kappa$-Poincare Hopf algebra can be seen as quantum deformation of the Drinfeld-Jimbo type, since its classical r-matrix satisfies the modified Yang-Baxter equation. Two points of view were considered: a traditional one, with h-adic topology, and a new one, with a reformulation of the algebraic sector of the corresponding Hopf algebra (the so-called q-analog version, which is no longer of the form of universal enveloping algebra). Correspondingly one can consider two non-isomorphic versions of $\kappa$-Minkowski spacetime (i.e. h-adic and q-analog ones), which in both cases are covariant quantum spaces (i.e. Hopf module algebra over $\kappa$-Poincare Hopf algebra). From the physical point of view the most interesting model is $\kappa$-Minkowski as universal enveloping algebra of solvable Lie algebra without h-adic topology.
Furthermore by using smash (a.k.a. crossed) product construction the so called DSR algebra was obtained (the name DSR comes from the Doubly (Deformed) Special Relativity Theory interpretation). Its construction uniquely unified $\kappa$-Minkowski spacetime coordinates with $\kappa$-Poincare generators. It was also proved that this DSR algebra can be obtained by a non-linear change of the generators from the undeformed one. It means that DSR algebra is insensitive on deformation which is kind of no-go theorem. Therefore a physical interpretation should be related with with a choice of generators \cite{Liberati}.

 A general result establishing isomorphism between deformed (twisted) and undeformed smash product algebras was proved. The Deformed Phase Space (Weyl) algebra was also described as a subalgebra of the DSR algebra, together with momenta and coordinates realizations. Again, it turned out to be obtained by change of generators in the undeformed one. The construction of the DSR algebra, in general, depends on the mathematical model of $\kappa$-Minkowski spacetime and $\kappa$-Poincare one uses. In our approach the basis of Poincare sector of this algebra stays classical. However their different representations in Hilbert space are responsible for realizations of noncommutative spacetime coordinates and at the same time for different physical properties of DSR-type models. Moreover the twist realization of $\kappa$-Minkowski spacetime was described as a quantum covariant algebra determining a deformation quantization of the corresponding linear Poisson structure.

\textbf{Perspectives}\\
In the view of future research it could be interesting to investigate in more details other possible applications of $\kappa$-deformation or noncommutative spacetimes in general in physical theories.  For example noncommutative gravity is considered as a toy model of its quantisation. However there exist many different formulations of noncommutative gravity. Among them are the twist deformation of Riemannian geometry, the spectral action approach as model of gravity coupled to matter, deformed gravity as gauge theory of deformed Poincare symmetry, as well as gravity considered on noncommutative spacetime with Poisson structure. One of the widely investigated so far is the deformation of Riemannian geometry using a twist element, together with theta-deformation of Einstein equations and their solutions. However it would be interesting to connect some of the above mentioned approaches (the twist-deformed gravity and gravity on Poisson manifold or gravity as gauge theory of deformed Poincare symmetry) using $\kappa$-Minkowski spacetime. 
As starting point for the research on kappa-deformed gravity might be twist deformation of Riemannian geometry already proposed for theta-deformation, i.e. to obtain and analyse the exact solutions of noncommutative gravity basing on kappa-deformation, and to investigate generic properties of the solutions of deformed Einstein equations. Using twists for $\kappa$-Minkowski spacetime (Abelian and Jordanian) one could construct Einstein equations following the same procedure. It was proven that in the $\theta$- deformation case , by using the Drinfeld twist built up from Killing vectors, that solutions of undeformed gravity are noncommutative solutions as well. Therefore it would be interesting to find all the possible metrics for which Killing vectors coincide with vectors constituting , e.g. the Abelian twist for $\kappa$-Minkowski spacetime. Moreover having an appropriate twist element one can extend the star-product to all tensors and see which 'kappa' (quantum gravity scale) corrections appear in, e.g. curvature and torsion tensors, the covariant derivative, Christoffel symbols etc. Particularly deformed covariant tensors are connected with deformed differential calculus, which also could be developed in more detail. 
Of particular interest, for example, would be a study on geodesics, the Newtonian limit or corrections to gravitational waves solutions. Moreover it would be also interesting to consider if such corrections are able to address the cosmological constant problem, possibly explaining if such term in Einstein's equations has indeed a quantum origin.

Moreover it would another point of view could be given by approach to noncommutative gravity using Poisson structure, since $\kappa$-Minkowski spacetime can be also seen as deformation quantization (due to the Kontsevich theorem) of the corresponding linear Poisson structure. The comparison of the results from those a priori different descriptions of quantisation, connected with the noncommutative spacetime example of $\kappa$-Minkowski spacetime is interesting and could inspire development of some general connection between them.\\
One could also consider a supersymmetric extension of the above noncommutative (kappa-deformed) theory of gravity using quantum symmetries in the super-Hopf-algebraic formalism. It is possible to obtain twists describing supersymmetric quantum deformations from super-r-matrices of Abelian and Jordanian type. One could use this method of supersymmetrisation and construct super-extensions of the twists corresponding to kappa-Minkowski spacetime mentioned above and consider the supersymmetrised twist deformation of Riemannian geometry. The supersymmetric extensions of standard symmetries were and still are a very fruitful framework and it is interensting to consider them together with quantum symmetries and quantum deformations as a toy model of another way of quantising gravity on the ground of $\kappa$-Minkowski super-spacetime with noncommuting spacetime variables and non-anti-commuting (spinorial) grassman sector.\\

The next step could be to examine some of the consequences of noncommutative gravity, e.g. constructing kappa-deformed cosmology with Hamiltonian formulation of deformed Einstein's equations and asymptotic evolution of cosmological models with deformed motion in superspace. Moreover kappa-deformed black hole physics could be discussed involving $\kappa$-deformed Schwarzschild and Kerr-Newman solutions together with $\kappa$- black holes thermodynamics.  More strictly, starting from $\kappa$-deformed phase spaces one can consider $\kappa$-deformed cosmology, using dynamical system approach to Cosmological Models. Since the evolution of the Universe can be described as an Hamiltonian system it is planned to use methods developed for the classical case but take the deformed phase space in the first step. Moreover using Hamiltonian methods one can easily analyse the generic asymptotic behaviour of cosmological models near the singularity and at late times. Also due to canonical formulation one could interpret the evolution of the model as a point particle in a space.  The space in which the point particle (evolution) moves is called superspace. Therefore one can also consider deformed motion of such particle in superspace and minisuperspace models which one usually considers in the classical case.
The $\kappa$-deformed phase space, together with new uncertainty relations, is also a good starting point for the research on $\kappa$-deformed black holes. Corrections to solutions in the Schwarzschild and Kerr-Newman cases, were previously considered using the techniques of effective field theory, i.e. a different approach to quantum gravity. It would be interesting to investigate and compare the order of corrections using the $\kappa$-deformation approach.
This last part of the research  might be the most interesting from the point of view  of the experimental search of quantum gravity effects. It could shed light on some of the most challenging issues of modern theoretical astrophysics and experimental cosmology. 

The above mentioned solutions could be also understood within the DSR framework, that has been suggested to be the phenomenological limit of any theory of Quantum Gravity. Moreover  noncommutative spacetimes seem also to be connected with spin foam formalism via Group Field Theory (GFT) on mathematical level. GFT as mathematical formalism describing the most general approach to quantum gravity theory provides a link between field theory on group manifolds and noncommutative field theories on noncommutative spacetimes. As one of the examples only one realisation of $\kappa$-Minkowski spacetime was used (standard DSR type): this point could be generalised.  Moreover four-dimensional  $\kappa$-Minkowski spacetime can be described as subalgebra of de Sitter algebra and can be realized using Iwasawa decomposition. Therefore one can introduce generalized Iwasawa-type decompositions connected with the above described class of realizations. Furthermore one can try to transfer this results to the inhomogeneus de Sitter group. 
Moreover the $\kappa$-Poincare quantum group has two Casimir operators, connected with mass and spin. The deformed quadratic Casimir of mass leads to deformed dispersion relations, which turns out to depend on the Weyl algebra realization of $\kappa$-Minkowski spacetime. The construction of a realization-dependent form of deformed second, Pauli-Lubanski, Casimir operator could provide new properties of deformed spin.\\

Therefore $\kappa$-Minkowski spacetime is placed between a strictly mathematically theoretical example of noncommutative algebra and a physical model which might describe nature at the Planck scale. With its strong connection with Doubly Special Relativity Theory it becomes very promising in providing formalism of experimental tests of quantum gravity models (as, e.g. the above mentioned time delays which naturally appear from noncommutativity). Also the mathematical framework necessary in the investigation of  the properties of $\kappa$-Minkowski spacetime is highly not trivial and interesting itself. Noncommutative spacetimes as a one of the approach to a quantum theory of gravity has become very fruitful in inspiring developments in theoretical physics as well as in pure mathematics, with a continuously growing interest. Moreover the proposed application of the kappa-deformed gravity in cosmology and black hole physics might be
the most interesting part from the point of view  of the experimental search of quantum gravity effects. It could new shed light on some of the most challenging issues of modern theoretical astrophysics and experimental cosmology.

\appendix

\chapter{$h$-adic topology}\label{app3}
\fancyhf{} 
\fancyhead[LE,RO]{\bfseries\thepage}
\fancyhead[RE]{Appendixes}

We use the notion of $h$-adic (Hopf) algebras and $h$-adic  modules, i.e.\ (Hopf) algebras and modules dressed in $h$-addic topology. Therefore, we would like to collect basic facts concerning $h$-adic topology which is required by the concept of deformation quantization (for more details see \cite[Chapter~1.2.10]{Klimyk}  and \cite[Chapter~XVI]{Kassel}). For example, in the case of quantum enveloping algebra, $h$-adic topology provides invertibility of twisting elements and enables the quantization.

Let us start from the commutative ring of formal power series $\mathbb{C}[[h]]$: it is a (ring) extension of the f\/ield of complex numbers $\mathbb{C}$ with elements of the form:
\[
\mathbb{C}[[h]]\ni a=\sum_{n=0}^\infty a_n h^n,
\]
where $a_n$ are complex coef\/f\/icients and $h$ is undetermined. One can also see
this ring as $\mathbb{C}[[h]]=\times_{n=0}^\infty\mathbb{C}$
which elements are (inf\/inite) sequences of complex numbers $(a_0,a_1,\dots,a_n,\dots)$ with powers of $h$ just ``enumerating'' the position of the coef\/f\/icient. Thus, in fact, $\mathbb{C}[[h]]$ consists of all inf\/inite complex valued sequences, both convergent and divergent in a sense of standard topology on $\mathbb{C}$. A subring of polynomial
functions $\mathbb{C}[h]$ can be identif\/ied with the set of all f\/inite sequences. Another important subring is provided by analytic functions $\mathcal{A}(\mathbb{C})$, obviously: $\mathbb{C}\subset\mathbb{C}[h]\subset\mathcal{A}(\mathbb{C})\subset\mathbb{C}[[h]]$. Slightly dif\/ferent variant of sequence construction can be applied to obtain the real out of the rational numbers. Similarly, the ring $\mathbb{C}[[h]]$ constitute substantial extension of the f\/ield $\mathbb{C}$ and specialization of the indeterminant $h$ to take some numerical value does not make sense, strictly speaking.
The ring structure is determined by addition and multiplication laws:
\[
a+b:=\sum_{n=0}^\infty(a_n+b_n)h^n,\qquad a\cdot b:= \sum_{n=0}^\infty\left(\sum_{r+s=n}a_r b_s\right)h^n.
\]
This is why the power series notation is only a convenient tool for encoding the multiplication (the so-Cauchy multiplication).
Let us give few comments on topology with which it is equipped, the so-called ``$h$-adic'' topology.
This topology is determined ``ultra-norm'' $||\cdot||_{\rm ad}$ which is def\/ined~by:
\[
\left\|\sum_{n=0}^\infty a_n h^n \right\|_{\rm ad}=2^{-n(a)},
\]
where $n(a)$ is the smallest integer such that $a_n\neq 0$  (for $a\equiv 0$ one sets $n(a)=\infty$ and therefore $||0||_{\rm ad}=0$).
It has the following properties:
\begin{gather*} 0\leq ||a||_{\rm ad}\leq 1, \qquad
||a+b||_{\rm ad}\leq \max(||a||_{\rm ad}, ||b||_{\rm ad}), \\ ||a\cdot b||_{\rm ad}=||a||_{\rm ad}||b||_{\rm ad},
\qquad ||h^k||_{\rm ad}=2^{-k}.
\end{gather*}
It is worth to notice that the above def\/ined norm is discrete (with values in inverse powers of~2).
Important property is that the element $a\in\mathbb{C}[[h]]$ is invertible if an only if $||a||_{\rm ad}=1$.\footnote{Particularly, all nonzero complex numbers are of unital ultra-norm.} Above topology makes
the formalism of formal power series self-consistent in the following sense: all formal power series becomes convergent (non-formal) in
the norm $||\cdot||_{\rm ad}$. More exactly, if $\mathbb{C}[[h]]\ni a=\sum\limits_{n=0}^\infty a_n h^n$ is a formal power series then
$||a-A_N||_{\rm ad}\rightarrow 0$, with $A_N=\sum\limits_{n=0}^{n=N} a_n h^n$ being the sequence of partial sums.
Moreover, $\mathbb{C}[[h]]$ is a topological ring, complete in $h$-adic topology;
in other words the addition and the  multiplication are continuous operations and $h$-adic Cauchy sequences are convergent to the limit which belongs to the ring.

Furthermore one can extend analogously other algebraic objects, such as vector spaces, algebras, Hopf algebras, etc. and equip them in $h$-adic topology.
Considering $V$ as a complex vector space the set $V[[h]]$ contains all formal power series $v=\sum\limits_{n=0}^\infty v_n h^n$ with coef\/f\/icients $v_n\in V$. Therefore $V[[h]]$ is a $\mathbb{C}[[h]]$-module.
More generally in  the deformation theory we are forced to work with the category of topological $\mathbb{C}[[h]]$-modules, see~\cite[Chapter~XVI]{Kassel}.
$V[[h]]$ provides an example of topologically free modules. Particularly if $V$ is f\/inite dimensional  it is also free module. Any basis $(e_1,\dots,e_N)$ in $V$ serves as a system of free generators in $V[[h]]$. More exactly
\[
\sum_{k=0}^\infty v_kh^k=\sum_{a=1}^N x^ae_a , \qquad  v_k=\sum_{a=1}^N x^a_k e_a,
\]
where the  coordinates $x^a=\sum\limits_{n=0}^\infty x^a_nh^n\in \mathbb{C}[[h]]$.
It shows that $V[[h]]$ is canonically isomorphic to~$V\otimes\mathbb{C}[[h]]$. The ultra-norm $||\cdot||_{\rm ad}$ extends to $V[[h]]$ automatically. Particulary, if $V$ is equipped with an algebra structure then the Cauchy multiplication makes $V[[h]]$ a topological algebra.

Intuitively, one can think of the quantized universal enveloping algebras
introduced in the paper as families of Hopf algebras depending
on a f\/ixed numerical parameter $h$. However this does not make sense, for an algebra def\/ined over the ring $\mathbb{C}[[h]]$. To remedy this situation, one has to introduce a new algebra, def\/ined over the f\/ield of complex numbers. However this procedure is not always possible.
One can specialize $h$ to any complex number in the case of Drinfeld--Jimbo deformation but not in the case of twist deformation. For more details and examples see e.g.~\cite[Chapter~9]{Chiari}.

\chapter{Bicrossproduct construction versus Weyl-Heisenberg algebras}\label{app2}
Bicrossproduct construction, originally introduced in \cite{MR_bicross1} (see
also \cite{MR_bicross2}, \cite{M_book} for more details), allows us to construct a new
bialgebra from two given ones. Its applicability to Weyl-Heisenberg algebra is a subject of our study here. In fact, algebraic sector of Weyl-Heisenberg algebra relies on crossed-product construction \cite{Klimyk},\cite{Kassel},\cite{smash},\cite{Majid2}
while the coalgebraic one will be main issue of our investigation here. One can easily show that full bialgebra
structure cannot be determined in this case.  However appropriate weakening of some assumptions   automatically  allows on
bicrossproduct type construction.

We start this note with reviewing the notions of Weyl-Heisenberg algebra and indicating its basic properties.
Then we recall definitions of crossed product algebras, comodule coalgebras, their crossed coproduct and bicrossproduct
construction. We follow with some
examples of bicrossproduct construction for the classical inhomogeneous
orthogonal transformations as well as for  the $\kappa-$deformed case.

\section{Preliminaries and notation: Weyl-Heisenberg algebras}
Let us start with reminding that Weyl-Heisenberg algebra\footnote{Here an algebra means unital, associative algebra over a commutative ring which is assumed to be a field of complex numbers $\mathds{C}$ or its h-adic extensions $\mathds{C}[[h]]$ in the case of deformation.}
$\W(n)$ can be defined as an universal algebra with $2n$ generators $\{x^1\ldots x^n\}\cup\{P_1\ldots P_n\}$ satisfying
the following set of commutation relations (cf. example \ref{example4})
\begin{gather}  \label{Weyl}
P_{\mu }x^{\nu}-x^\nu P_\mu = \delta _{\mu}^\nu\, 1, \qquad  x^{\mu }x^{\nu } - x^{\nu }x^{\mu }
  = P_{\mu }P_{\nu } - P_{\nu }P_{\mu } =0\ .
\end{gather}
for $\mu, \nu =1\ldots n$.

It is worth to underline that the Weyl-Heisenberg algebra as defined above is not an enveloping algebra of some Lie algebra.  More precisely, in contrast to the  Lie algebra case, Weyl-Heisenberg algebra have no finite dimensional (i.e matrix) representations. One can check it by taking the trace of the basic commutation relation $[x,p]=1$ which leads to the contradiction.
Much in the same way one can set
\begin{proposition}
There is no bialgebra structure which is compatible with the commutation relations (\ref{Weyl}).
\end{proposition}
The proof is trivial: applying the counit $\epsilon$ to both sides of the first commutator in (\ref{Weyl}) leads to a contradiction since $\epsilon (1)=1$.

The best known representations are given on the space of (smooth) functions on $\mathds{R}^n$ in terms of multiplication and differentiation operators, i.e. $P_\mu= {\partial\over \partial x^\mu}$. For this reason one can identify Weyl-Heisenberg algebra with an algebra of linear differential operators on $\mathds{R}^n$ with polynomial coefficients. In physics, after taking a suitable real structure, it is known as an algebra of the canonical commutation relations. Hilbert space representations of these algebras play a central role in Quantum Mechanics while their counterpart with infinitely many generators (second quantization) is a basic tool in Quantum Field Theory.

 A possible deformation of Weyl-Heisenberg algebras have been under investigation \cite{Pillin}, and it turns out that there is no non-trivial deformations of the above algebra within a category of algebras. However the so-called q-deformations have been widely investigated, see e.g. \cite{Pillin,Wess,Lavagno}.

Another obstacle is that the standard, in the case of Lie algebras, candidate for undeformed (primitive) coproduct
\be\label{primitive}
\Delta_0(a)=a\otimes 1+ 1\otimes a
\ee
$a\in \{x^1\ldots x^n\}\cup\{P_1\ldots P_n\}$ is also incompatible with (\ref{Weyl}).
It makes additionally impossible to determine a bialgebra structure on the Weyl-Heisenberg algebras.

However one could weaken the notion of bialgebra and consider unital non-counital bialgebras equipped with 'half-primitive' coproducts \footnote{These formulae were announced to us  by S. Meljanac and D. Kovacevic in the context of Weyl-Heisenberg algebra.}, left or right:
\be \label{half_prim}
\Delta^L_0(x)=x\otimes 1;\qquad\Delta^R_0(x)=1\otimes x
\ee
on $\W(n)$. In contrast to (\ref{primitive}) which is valid only on generators, the formulae (\ref{half_prim}) preserve their form for all elements of the algebra.

Moreover, such coproducts turn out to be applicable also to larger class of deformed  coordinate algebras  (quantum spaces \cite{Dop1}-\cite{Dop2}) being, in general, defined by commutation relations of the form
\be\label{q-space}
x^\mu x^\nu - x^\mu x^\nu = \theta^{\mu\nu}+\theta^{\mu\nu}_\lambda x^\lambda+\theta^{\mu\nu}_{\rho\sigma}x^\rho x^\sigma +\ldots
\ee
for constant parameters $\theta^{\mu\nu}, \theta^{\mu\nu}_\lambda,\theta^{\mu\nu}_{\lambda\rho},\ldots$ . Of course, one has to assume that the number of components on the right hand side of (\ref{q-space}) is finite.
\begin{proposition}
The left (right)-primitive coproduct determines a non-counital bialgebra structure on the class of algebras defined by the commutation relations
(\ref{q-space}).
\end{proposition}
\begin{remark}
Such deformed algebra provides a deformation quantization of $\mathds{R}^n$ equipped
with the Poisson structure:
\be
\{x^\mu,x^\nu\}=\theta^{\mu\nu}(x)=\theta^{\mu\nu}+\theta^{\mu\nu}_\lambda x^\lambda+\theta^{\mu\nu}_{\rho\sigma}x^\rho x^\sigma +\ldots
\ee
represented by Poisson bivector $\Theta=\theta^{\mu\nu}(x)\partial_\mu\wedge \partial_\nu$.\end{remark}
Particularly, one can get the so-called theta-deformation:
\be\label{theta}
[x^\mu,x^\nu]=\theta^{\mu\nu}
\ee
which can be obtained via  twisted deformation by means of Poincar\'{e} Abelian twist:
$$\F=exp(\theta^{\mu\nu}P_\mu\wedge P_\nu)$$
The same twist provides also $\theta-$ deformed Poincar\'{e} Hopf algebra as a symmetry group, i.e. the quantum group with respect to which (\ref{theta}) becomes a covariant quantum space \footnote{
Note that the twist deformation requires h-adic extension.} .

Another way to omit counital coalgebra problem for (\ref{Weyl}) relies on introducing the
central element $C$ and replacing the commutation relations (\ref{Weyl}) by the following Lie algebraic ones
\begin{gather}  \label{Heisenberg}
\left[ P_{\mu },x^{\nu }\right]=-\imath \delta _{\mu}^\nu C, \qquad \left[
x^{\mu },x^{\nu }\right]=\left[C ,x^{\nu }\right]=\left[ P_{\mu },P_{\nu }%
\right]=\left[C ,P_{\nu }\right]=0.
\end{gather}
The relations above determine $(2n+1)$-dimensional Lie algebra  of
rank $n+1$ which we shall call Heisenberg-Lie algebra $\mathfrak{hl}(n)$. This algebra can be described as a central extension of the Abelian Lie algebra
$\mathfrak{ab}(x^1,\ldots , x^n, P_1, \ldots, P_n)$. Thus Heisenberg algebra
can be now def\/ined as an enveloping algebra $\U_{\mathfrak{hl}(n)}$ for (\ref{Heisenberg}). There is no problem to introduce Hopf algebra structure with the primitive coproduct (\ref{primitive}) on the generators $\{x^1,\ldots , x^n, P_1, \ldots, P_n, C\}$.
This type of extension  provides a
starting point for Hopf algebraic deformations, e.g. quantum group framework is considered in~\cite{LukMinn}, \cite{Bonechi}, standard and nonstandard deformations are presented e.g. in \cite{tw_Heis} while deformation quantization formalism is developed in \cite{Sorace}. As a trivial example of quantum deformations of the Lie algebra (\ref{Heisenberg}) one can consider the maximal Abelian twist of the form:
\begin{equation}
\F=exp(ih\theta^{\mu\nu}P_\mu\wedge P_\nu)exp(\lambda^\mu P_\mu\wedge C)
\end{equation}
$\theta^{\mu\nu}, \lambda^\mu$-are constants (parameters of deformation).
It seems to us, however, that there are no enough strong physical motivations for studying  deformation problem for such algebras. Therefore we shall focus on possibilities of  relaxing some algebraic conditions in the definition of bicrossproduct bialgebra in order to obey the case of Weyl-Heisenberg algebra as it is defined by (\ref{Weyl}).\\

\section{Crossed product and coproduct}

\textbf{Crossed product algebras} were widely investigated in Section 1 of this thesis, provided with many examples.
However here we shall consider the case of  Weyl-Heisenberg algebra introduced above (\ref{Weyl}) As opposed to example (\ref{example4}) here we consider the case of right module and right action. For this purpose one considers two copies of Abelian $n-$dimensional Lie algebras: $\mathfrak{ab}(P_1,\ldots,P_{n})$, $\mathfrak{ab}(x^1,\ldots, x^{n})$ together with  the corresponding universal enveloping  algebras $\U_{\mathfrak{ab}(P_1,\ldots,P_{n})}$  and $\U_{\mathfrak{ab}(x^1,\ldots,x^{n})}$. Alternatively both algebras are isomorphic to the universal commutative algebras with $n$ generators (polynomial algebras). These two algebras  constitute a dual pair of Hopf algebras. Making use of primitive coproduct on generators of $\U_{\mathfrak{ab}(P_1,\ldots,P_{n})}$  we extend the (right) action implemented by duality map
\begin{gather}  \label{actionP2A2}
x^\nu\triangleleft P_{\mu}=\delta_\mu^\nu, \qquad 1\triangleleft P_{\mu}=0
\end{gather}
to the entire algebra  $\U_{\mathfrak{ab}(x^1,\ldots,x^{n})}$. Thus
 $W(n)=\U_{\mathfrak{ab}(P_1,\ldots,P_{n})}\ltimes\U_{\mathfrak{ab}(x^1,\ldots,x^{n})}$.

Similarly, the Heisenberg-Lie algebra can be obtained in the same way provided slight modifications in the action:
\begin{equation}\label{actionP3A2}
x^\nu\triangleleft P_{\mu}=\delta_\mu^\nu C, \qquad C\triangleleft
P_{\mu}=0
\end{equation}
It gives  $\U_{\mathfrak{hl}(n)}=\U_{\mathfrak{ab}(P_1,\ldots,P_{n})}\ltimes\U_{\mathfrak{ab}(x^1,\ldots,x^{n},C)}$.

\textbf{Crossed coproduct coalgebras} \cite{MR_bicross2},\cite{Klimyk}\newline
The dual concept to the action of an algebra (introduced in def. \ref{module}%
) is the \textit{coaction} of a coalgebra. Let now $\mathcal{A}=\mathcal{A}\left( m_{\mathcal{A}}, \Delta _{\mathcal{A}},\epsilon_{\mathcal{A}}, 1_{\mathcal{A}}\right)$ be a bialgebra and $\mathcal{H}= \mathcal{H}\left( \Delta_{\mathcal{H}}, \epsilon_{\mathcal{H}}\right)$  be a coalgebra. The left coaction of the bialgebra $
\mathcal{A}$ over the coalgebra $\mathcal{H}$ is defined as linear map: $\beta :%
\mathcal{H}\rightarrow \mathcal{A}\otimes \mathcal{H};$ with the following Sweedler type
notation: $\beta \left( L\right) = L^{\left( -1\right) }\otimes
L^{\left( 0\right) }$, where $L^{\left( -1\right) }\in\mathcal{A}$ and $%
L^{\left( 0\right) }\in\mathcal{H}$, $\beta(1_\H)=1_\A\otimes 1_\H$.

\begin{definition}
\label{beta} We say that $\mathcal{H}$ is left $\mathcal{A}$ -\textbf{%
comodule coalgebra} with the structure map $\beta :\mathcal{H}\rightarrow
\mathcal{A}\otimes \mathcal{H}$ if this map satisfies the following two
conditions: $\forall f,g\in \mathcal{A};L,M\in \mathcal{H}$

1)
\begin{equation}  \label{com_1}
\left( id_{\mathcal{A}}\otimes \beta \right)\circ \beta =\left( \Delta_{%
\mathcal{A}} \otimes id_{\mathcal{H}}\right)\circ \beta
\end{equation}
which can be written as: $L^{\left( -1\right) }\otimes (L^{\left( 0\right)}
)^{\left( -1\right) }\otimes (L^{\left( 0\right)} )^{\left( 0\right)
}=\left( L^{\left( -1\right) }\right) _{\left( 1\right) }\otimes \left(
L^{\left( -1\right) }\right) _{\left( 2\right) }\otimes L^{\left( 0\right) }$%
\newline
and $(\epsilon_{\mathcal{A}}\otimes id_{\mathcal{H}})\circ\beta=id_{\mathcal{%
H}}$ which reads as: $\epsilon _{\mathcal{A}}\left( L^{\left( -1\right)
}\right) L^{\left( 0\right) }= L;$

2) Additionally it satisfies comodule coaction structure (comodule coalgebra
conditions):
\begin{equation}  \label{com_2}
L^{\left( -1\right) } \epsilon _{\mathcal{H}}\left( L^{\left( 0\right)
}\right) =1_{\mathcal{A}}\epsilon _{\mathcal{H}}\left( L\right)
\end{equation}
\begin{equation}
L^{\left( -1\right) }\otimes \left( L^{\left( 0\right) }\right) _{\left(
1\right) }\otimes \left( L^{\left( 0\right) }\right) _{\left( 2\right)
}=\left( L_{\left( 1\right) }\right) ^{\left( -1\right) }\left( L_{\left(
2\right) }\right) ^{\left( -1\right) }\otimes \left( L_{\left( 1\right)
}\right) ^{\left( 0\right) }\otimes \left( L_{\left( 2\right) }\right)
^{\left( 0\right) }
\end{equation}
\end{definition}

Left $\mathcal{A}$-comodule coalgebra is a bialgebra $\mathcal{H}$ which is
left $\mathcal{A}$-comodule such that $\Delta _{\mathcal{H}}$ and $\epsilon
_{\mathcal{H}}$ are comodule maps from definition \ref{beta}. 

For such a left $\mathcal{A}$ - comodule coalgebra $\mathcal{H}$, the vector
space $\mathcal{H}\otimes \mathcal{A}$ becomes a (counital) coalgebra with the
comultiplication and counit defined by:
\begin{equation}  \label{copcom}
\Delta_\beta \left( L\otimes f\right) =\sum L_{\left( 1\right) }\otimes \left(
L_{\left( 2\right) }\right) ^{\left( -1\right) }f_{\left( 1\right) }\otimes
\left( L_{\left( 2\right) }\right) ^{\left( 0\right) }\otimes f_{\left(
2\right) }
\end{equation}
\begin{equation}  \label{epsiloncom}
\epsilon \left( L\otimes f\right) =\epsilon _{\mathcal{H}}\left( L\right)
\epsilon _{\mathcal{A}}\left( f\right)
\end{equation}
$L\in \mathcal{H};f\in \mathcal{A}.$

This coalgebra is called the \textbf{left crossed product coalgebra} and it
is denoted by $\mathcal{H}\rtimes ^{\beta }\mathcal{A}$ or $\mathcal{H}%
\rtimes \mathcal{A}.$ One should notice that:
\begin{equation*}
\Delta_\beta(L\otimes 1_{\mathcal{A}})=\left(L_{(1)}\otimes
(L_{(2)})^{(-1)}\right)\otimes\left((L_{(2)})^{(0)}\otimes 1_{\mathcal{A}%
}\right)=L_{(1)}\otimes \beta(L_{(2)})\otimes 1_{\mathcal{A}}
\end{equation*}
and
\begin{equation*}
\Delta_\beta(1_{\mathcal{H}}\otimes g)=\left(1_{\mathcal{H}}\otimes
g_{(1)}\right)\otimes\left(1_{\mathcal{H}}\otimes g_{(2)}\right)
\end{equation*}
i.e. $\Delta_\beta(\tilde f) =\tilde f_{(1)}\otimes\tilde f_{(2)}$, where $\tilde f= 1_\H\otimes f$. Moreover for the trivial choice
\be\label{trivial}
\beta_{trivial} (M)=1_\A\otimes M
\ee
one also gets
\be
\Delta_\beta(\tilde M)=\tilde M_{(1)}\otimes\tilde M_{(2)}
\ee
where $\tilde M=M\otimes 1_\A$. This implies that both coalgebras are subcoalgebras in $\mathcal{H}\rtimes \mathcal{A}.$
\begin{remark}\label{r7}
Let us assume for a moment that the coalgebra $\H$ has no counit. Leaving remaining assumptions in the same form and skipping ones containing $\epsilon_\H$ we can conclude that the resulting coalgebra  $\mathcal{H}\rtimes ^{\beta }\mathcal{A}$ has no counit (\ref{epsiloncom}) as well. In other words all other elements of the construction work perfectly well.
\end{remark}

\section{Bicrossproduct construction}
Through this section let both $\mathcal{H}$ and $\mathcal{A}$ be
bialgebras. The structure of an action is useful for crossed product algebra construction and a coaction map allows us to consider crossed coalgebras. However considering both of them simultaneously we are able to perform the so-called bicrossproduct construction. %
\begin{theorem} \label{becross}
(S. Majid \cite{MR_bicross2}, Theorem 6.2.3) Let $\mathcal{H}$ and $\mathcal{A}$ be bialgebras and $%
\mathcal{A}$ is right $\mathcal{H}$-module with the structure map $\triangleleft
:\mathcal{A}\otimes \mathcal{H}\rightarrow \mathcal{A}$. And $\mathcal{H}$
is left $\mathcal{A}$-comodule coalgebra with the structure map\newline
$\beta :\mathcal{H}\rightarrow \mathcal{A}\otimes \mathcal{H}$, $\beta
\left( L\right) =L^{\left( -1\right) }\otimes L^{\left( 0\right) }$ (cf.
def. \ref{beta}).\\
Assume further the following compatibility conditions:\newline
(A)
\begin{equation}  \label{bicrossA1}
\Delta_{\mathcal{A}} \left( f\triangleleft L\right) =\sum \left(
f\triangleleft L\right) _{\left( 1\right) }\otimes \left( f\triangleleft
L\right) _{\left( 2\right) }= \left( f_{\left( 1\right) }\triangleleft
L_{\left( 1\right) }\right) \left( L_{\left( 2\right) }\right) ^{\left(
-1\right) }\otimes f_{\left( 2\right) }\triangleleft \left( L_{\left(
2\right) }\right) ^{\left( 0\right) }
\end{equation}
\begin{equation}  \label{bicrossA2}
\epsilon _{\mathcal{A}}\left( f\triangleleft L\right) =\epsilon _{\mathcal{A}%
}\left( f\right) \epsilon _{\mathcal{H}}\left( L\right)
\end{equation}
(B)
\begin{equation}  \label{bicrossB}
\beta \left( LM\right) =\left( LM\right) ^{\left( -1\right) }\otimes \left(
LM\right) ^{\left( 0\right) }=\sum \left( L^{\left( -1\right) }\triangleleft
M_{\left( 1\right) }\right) \left( M_{\left( 2\right) }\right) ^{\left(
-1\right) }\otimes L^{\left( 0\right) }\left( M_{\left( 2\right) }\right)
^{\left( 0\right) }
\end{equation}
\begin{equation}  \label{bicrossB2}
\beta(1_{\mathcal{H}})\equiv\left( 1_{\mathcal{H}}\right) ^{\left( -1\right) }\otimes \left( 1_{\mathcal{%
H}}\right) ^{\left( 0\right) }=1_{\mathcal{A}}\otimes 1_{\mathcal{H}}
\end{equation}
(C)
\begin{equation}  \label{bicrossC}
\left( L_{\left( 1\right) }\right) ^{\left( -1\right) }\left( f\triangleleft
L_{\left( 2\right) }\right) \otimes \left( L_{\left( 1\right) }\right)
^{\left( 0\right) }=\left( f\triangleleft L_{\left( 1\right) }\right) \left(
L_{\left( 2\right) }\right) ^{\left( -1\right) }\otimes \left( L_{\left(
2\right) }\right) ^{\left( 0\right) }
\end{equation}
hold. Then the crossed product algebra $\mathcal{H}\ltimes\mathcal{A}$,
i.e. tensor algebra $\mathcal{H}\otimes\mathcal{A}$ equipped with algebraic:
\begin{equation*}
(L\otimes f)\cdot(M\otimes g)=LM_{(1)}\otimes (f\triangleleft
M_{(2)}g)\qquad\qquad\qquad(product)
\end{equation*}
\begin{equation*}
1_{\mathcal{H}\ltimes\mathcal{A}}=1_\mathcal{H}\otimes 1_\mathcal{A}
\qquad\qquad\qquad\qquad\qquad\qquad\qquad(unity)
\end{equation*} and coalgebraic
\begin{equation}\label{cop_bi}
\Delta_\beta(L\otimes f)=\left(L_{(1)}\otimes (L_{(2)})^{(-1)}
f_{(1)}\right)\otimes\left((L_{(2)})^{(0)}\otimes
f_{(2)}\right)\qquad(coproduct)
\end{equation}
\begin{equation*}
\epsilon(L\otimes
f)=\epsilon_\H(L)\epsilon_\A(f)\qquad\qquad\qquad\qquad\qquad\qquad\qquad(counit)
\end{equation*}
sectors becomes a bialgebra. Following \cite{MR_bicross1,MR_bicross2} one calls it bicrossproduct bialgebra and denotes as $\mathcal{H}%
\Join\mathcal{A}$. Moreover if the initial algebras are Hopf algebras then
introducing the antipode:
\begin{equation*}
S(L\otimes f)=(1_\mathcal{H}\otimes S_\A(L^{(-1)}f))\cdot(S_\H(L^{(0)})\otimes 1_%
\mathcal{A})\qquad\qquad(antipode)
\end{equation*}
it becomes bicrossproduct Hopf algebra $\mathcal{H}\Join\mathcal{A}$ as well.
\end{theorem}
\begin{example}
Primitive Hopf algebra structure on $\U_{\mathfrak{hl}(n)}$ can be obtained via bicrossproduct construction.
Take $\U_{\mathfrak{ab}(P_1,\ldots,P_{n})}$ as left $\U_{\mathfrak{ab}(x^1,\ldots,x^{n},C)}$ comodule algebra with the trivial coaction map: $\beta(P_{\mu})=1\otimes P_{\mu}$. Taking into account the action (\ref{actionP3A2})
all assumptions from the previous theorem are fulfilled.
Thus due  to the formula (\ref{cop_bi})
one obtains the following coalgebraic structure:\\ $\Delta \left( \tilde{P}_{\nu
}\right) =\tilde{P}_{\nu }\otimes 1+1\otimes \tilde{P}_{\nu };\qquad\Delta
\left( \tilde{x}^\nu\right) =\tilde{x}^\nu\otimes
1+1\otimes \tilde{x}^\nu;\qquad\Delta(\tilde{C})=\tilde{C}\otimes 1+1\otimes \tilde{C}$\newline
with canonical Hopf algebra embeddings: $1\otimes P_{\nu }\rightarrow \tilde{P}_{\nu
};\ {x}^\nu\otimes 1\rightarrow \tilde{x}^\nu;\ C\otimes 1\rightarrow \tilde{C}$.
\end{example}

The last example suggests the following more general statement:
\begin{proposition}
Let $\U_{\mathfrak{g}}$ and $\U_{\mathfrak{h}}$ be two enveloping algebras corresponding to two finite dimensional Lie algebras $\mathfrak{g},\mathfrak{h}$, both equipped in the primitive coalgebra structure (i.e. the coproduct $\Delta(x)=x\otimes 1+1\otimes x$ for $x\in {\mathfrak{g}}\cup {\mathfrak{h}}$). Assume that the (right) action of $\U_{\mathfrak{g}}$ on $\U_{\mathfrak{h}}$ is of Lie type, i.e. it is implemented by Lie algebra action: $h_a\triangleleft g_i =c_{ia}^b h_b$ in some basis $g_i$, $h_a$, where $c_{ia}^b$ are numerical constants.
Then one can always define the primitive Hopf algebra structure on $\U_{\mathfrak{g}}\ltimes\U_{\mathfrak{h}}$ by using 
bicrossproduct construction with the trivial co-action map: $\beta_{trivial}(g_i)=1\otimes g_i $. \end{proposition}

\section{Classical basis for $\kappa$ -Poincare Hopf algebra as bicrossproduct basis}
However from our point of view the most interesting case is deformed one.
To this aim let us remind  bicrossproduct construction for $\kappa$-Poincar\'{e} quantum group. In contrast to the original construction presented in \cite{MR} the resulting Hopf algebra structure will be determined in the classical Poincar\'{e} basis. The coaction map which provides $\kappa-$Poincar\'{e}
quantum (Hopf) algebra \cite{LNRT} was firstly proposed in \cite{MR}. In fact, the system of generators
used in the original construction \cite{MR} which preserves Lorentzian sector algebraically undeformed is called "bicrossproduct
basis". It became the most popular and commonly used by many authors in various applications, particularly in doubly special relativity formalism (see e.g. \cite{AC}-\cite{Smolin}) or quantum field theory on noncommutative $\kappa-$Minkowski spacetime (cf. \cite{DJMTWW},\cite{kappaQFT2}-\cite{kappaQFT4}). However bicrossproduct construction itself is  basis independent. Therefore we also  demonstrate that the so-called classical
basis (cf. \cite{BP2}) leaving entire Poincar\'{e} sector algebraically undeformed is consistent with the bicrossproduct construction and can be used instead as well.

We take as the first component enveloping algebra of 4-dimensional Lorentz Lie algebra $\mathfrak{o}\left( 1,3\right)$, closed in  h-adic topology, i.e. $\mathcal{H}=\mathcal{U}_{\mathfrak{o}\left( 1,3\right)}[[h]]$ with the primitive (undeformed) coalgebra structure (\ref{primitive}).
As the second component we assume Hopf algebra of translations $\mathcal{A}=\U_{\mathfrak{ab}(P_1,P_2,P_3,P_4)}[[h]] $ with nontrivial coalgebraic sector:
\begin{equation}
\Delta _{\kappa }\left( P_{i}\right) =P_{i}\otimes \left(hP_{4}+\sqrt{%
1-h^{2} P^{2}}\right)+1\otimes P_{i} \ ,\qquad i=1,2,3
\end{equation}
\begin{equation}
\Delta _{\kappa }\left( P_{4}\right) =P_{4}\otimes \left(h P_{4}+\sqrt{%
1-h^{2} P^{2}}\right)+\left(hP_{4}+\sqrt{1-h^{2} P^{2}}\right)^{-1}\otimes
P_{4} +hP _{m}\left(hP_{4}+\sqrt{1-h^{2} P^{2}}\right)^{-1}\otimes P^{m},
\label{copP0A2}
\end{equation}
here $P^2=P_\mu P^\mu$ and $\mu=1,\ldots ,4$. Observe that one deals here with formal power series in the formal parameter $h$ (cf. \cite{BP4}).
Now $\U_{\mathfrak{ab}(P_1,\ldots,P_4)}[[h]] $ is a right $\mathcal{U}_{\mathfrak{o}\left( 1,3\right)}[[h]]$ module algebra implemented by the  classical (right) action:
\begin{gather}  \label{class_right_action}
 P_{k}\triangleleft M_{j}=\imath \epsilon
_{jkl}P_{l},\qquad P_{4}\triangleleft M_{j}=0, \\
P_{k}\triangleleft N_{j} =- \imath \delta _{jk}P_{4},\qquad
P_{4}\triangleleft N_{j}=-\imath P_{j}  \label{class_right_action2}
\end{gather}
Conversely, $\mathcal{U}_{\mathfrak{o}\left( 1,3\right) }[[h]]$ is a left $\U_{\mathfrak{ab}(P_1,\ldots,P_4)}[[h]] $ - comodule coalgebra with (non-trivial) structure map defined on generators as follows:  
\begin{equation}\label{beta_kap1}
\beta_{\kappa} \left( M_{i}\right) =1\otimes M_{i}
\end{equation}
\begin{equation}\label{beta_kap2}
\beta_{\kappa} \left( N_{i}\right) =\left(h P_{4}+\sqrt{1-h^{2} P^{2}}%
\right)^{-1}\!\otimes N_{i}-h\epsilon _{ijm} P_{j}\left(h P_{4}+\sqrt{%
1-h^{2} P^{2}}\right)^{-1}\!\otimes M_{m}
\end{equation}
and then extended to the whole universal enveloping algebra. Such choice guarantees that all the conditions (\ref{bicrossA1}-\ref{bicrossC}) are fulfilled.
Thus  the structure obtained via bicrossproduct construction  constitutes  Hopf algebra $\mathcal{U}%
_{\mathfrak{o}\left( 1,3\right) }[[h]]\Join\U_{\mathfrak{ab}(P_1,\ldots,P_4)}[[h]] $ which has classical algebraic sector
while coalgebraic one reads as introduced in \cite{BP2,BP4}.

\section{The case of Weyl-Heisenberg algebras} 

Now we are in position to extend remark (\ref{r7}) to the bicrossproduct case. Again we have to neglect counit
on the bialgebra $\H$. As a result one obtains unital and non-counital bialgebra $\mathcal{H}\Join\mathcal{A}$.

As an illustrative example of such  constrution one can consider Weyl-Heisenberg algebra (\ref{Weyl}).
The algebra of translations $\U_{\mathfrak{ab}(P_1,\ldots,P_n)}$ is taken  with primitive coproduct.  Non-counital bialgebra of spacetime (commuting) coordinates $\U_{\mathfrak{ab}(x^1,\ldots,x^{n})}$ is assumed to posses half-primitive coproduct. The action is the same as in (\ref{actionP2A2}) while coaction is assumed to be trivial. As a final result one gets non-counital and non-cocommutative bialgebra structure on $\W(n)$:  $\Delta \left( \tilde{P}_{\nu
}\right) =\tilde{P}_{\nu }\otimes 1+1\otimes \tilde{P}_{\nu };\qquad\Delta
\left( \tilde{x}^\nu\right) =\tilde{x}^\nu\otimes 1$ , where $1\otimes P_{\nu }\rightarrow \tilde{P}_{\nu
};\ {x}^\nu\otimes 1\rightarrow \tilde{x}^\nu$.

It is still an open problem what kind of deformations can be encoded in the bicrossproduct construction.

For example, in the class of twisted deformation we were unable to find a single case obtained by means of such construction.
Nevertheless $\kappa$-deformation of the Poincar\'{e} Lie algebra is one of few  examples of quantization  for which bicrossproduct description works perfectly. More sophisticated examples can be found in \cite{more1}-\cite{more3}.
Moreover, it has been proved in \cite{BP4} that large class of deformations of the Weyl-Heisenberg algebra $\W(n)$ can be obtained as a (non-linear) change of generators in its h-adic extension $W(n)[[h]]$. Therefore our results concerning construction of non-counital bialgebra structure extend automatically to these cases.


\chapter{Heisenberg Doubles of quantized Poincare algebras}\label{app4}
Some elements of smash product of dual pair of Hopf algebras, known under the name of Heisenberg double, was introduced in Section 1. Here we shall consider the twisted situation which in the case of Heisenberg double turns out to be much more complicated. Twisting the coalgebra in $\mathcal{H}$ we deform coproduct
according to (\ref{twcop}) what influences the product in the dual algebra: $%
\mathcal{A}\rightarrow \mathcal{A}_{\mathcal{F}}$. I.e., from one hand this
new product has to take the form of the star-product (\ref{tsp}) in order to
preserve the Leibniz rule. From the other hand (\ref{tsp}) is not longer dual to
the deformed coproduct if one wants to keep pairing unchanged. In other
words around $\mathcal{A}_{\mathcal{F}}$ is a covariant quantum space with
respect to $\mathcal{H}^{\mathcal{F}}$ but it refuse to be dual Hopf algebra
with the same pairing.
We shall postpone more detailed study of this problem for another
publication. Here we concentrate on some explicit example provided by
quantum Poincar\'{e} Hopf algebras. 
Let us illustrate explicitly the case of nonstandard $\theta$-twisted
Poincare Heisenberg double. However later on we will also show how this
works in standard $\kappa$-deformation of Heisenberg double.\\

\textbf{Theta-twisted Poincare symmetry}\\
As above mentioned undeformed Poincare algebra $\mathcal{U}_{\mathrm{iso}%
(1,3)}$ contains $\{P_{\mu },M_{\alpha \beta }\}$ generators. And it also
constitutes a Hopf algebra with: (\ref{copPoincare},\ref{counPoincare}).
Within twist deformation framework we can deform it to new Hopf algebra. The
simplest example of twist deformation comes from the so-called $\theta $%
-twist:
\begin{equation*}
\mathcal{F}\in \mathcal{U}_{\mathrm{iso}(1,3)}\otimes \mathcal{U}_{\mathrm{%
iso}(1,3)}:\mathcal{F=}\exp \left( -\frac{i}{2}\theta ^{\mu \nu }P_{\mu
}\otimes P_{\nu }\right)
\end{equation*}%
After twisting the Hopf algebra structure becomes twisted Poincare algebra:
\begin{equation}
\Delta _{\theta }\left( P_{\mu }\right) =\Delta _{0}\left( P_{\mu }\right)
;\qquad \Delta _{\theta }\left( M_{\mu \nu }\right) =\Delta _{0}(M_{\mu \nu
})-(P\theta )_{\mu }\wedge P_{\nu }+(P\theta )_{\nu }\wedge P_{\mu }
\end{equation}%
where $(P\theta )_{\mu }=P_{\rho }\theta ^{\rho \lambda }\eta _{\lambda \mu }
$. 
The module algebra $\mathcal{F(}\mathrm{ISO}(1,3))$ can also be deformed
accordingly \cite{WorPod}, which will introduce new multiplication $(\ast
_{\theta })$ between its elements: $\mathcal{F}_{\theta }\mathcal{(}\mathrm{%
ISO}(1,3))=\{\left[ x^{\mu },x^{\nu }\right] _{\theta }=i\theta ^{\mu \nu };%
\left[ \Lambda _{\alpha }^{\beta },\Lambda _{\mu }^{\nu }\right] _{\theta
}=0,$ $\Lambda ^{T}\eta \Lambda =\eta \}$ . The coalgebra sector of $%
\mathcal{F}_{\theta }\mathcal{(}\mathrm{ISO}(1,3))$ does not change under $%
\theta -$twist because multiplication in $\mathcal{U}_{\mathrm{iso}(1,3)}$
remains undeformed. Therefore (\ref{copF}) still holds.\newline

$\theta $-twisted Heisenberg double: $\mathfrak{H}\left( \mathcal{U}_{%
\mathrm{iso}(1,3)}^{\theta }\right) $ has deformed crossed commutation
relations:
\begin{eqnarray}
\lbrack M_{\alpha \beta },\Lambda _{\nu }^{\mu }] &=&-<M_{\alpha \beta
},\Lambda _{\rho }^{\mu }>\Lambda _{\nu }^{\rho };\qquad \left[ \Lambda
_{\alpha }^{\beta },P_{\nu }\right] =0;\qquad \left[ P_{\mu },x^{\nu }\right]
=-\imath \delta _{\mu }^{\nu } \\
\lbrack M_{\mu \nu },x^{\lambda }] &=&-<M_{\mu \nu },\Lambda _{\rho
}^{\lambda }>x^{\rho }+\frac{\imath }{2}\left[ \theta _{\mu }^{\lambda
}P_{\nu }-\theta _{\nu }^{\lambda }P_{\mu }-\delta _{\nu }^{\lambda
}{}(P\theta )_{\mu }+\delta _{\mu }^{\lambda }{}(P\theta )_{\nu }\right]
\end{eqnarray}%
which are supplemented by (\ref{isog1},\ref{isog2}) and:
\begin{equation}
\left[ x^{\mu },x^{\nu }\right] _{\theta }=i\theta ^{\mu \nu };\qquad \left[
P_{\mu },P_{\nu }\right] =0;\qquad \left[ \Lambda _{\alpha }^{\beta
},\Lambda _{\mu }^{\nu }\right] =0
\end{equation}
As much as in the previous case the smash
product algebra $\mathfrak{X}_{\theta }^{4}$ $\rtimes \mathcal{U}_{\mathrm{%
iso}(1,3)}^{\theta }$ determined by the classical action is subalgebra of the above Heisenberg double $%
\mathfrak{H}\left( \mathcal{U}_{\mathrm{iso}(1,3)}^{\theta }\right) $.
It has the following cross-commutations:
\begin{equation}
\left[ P_{\mu },x^{\nu }\right] =-\imath \delta _{\mu }^{\nu };\qquad[M_{\mu \nu
},x^{\lambda }]=M_{\mu \nu }\triangleright x^{\lambda }+\frac{\imath }{2}%
\left[ \theta _{\mu }^{\lambda }P_{\nu }-\theta _{\nu }^{\lambda }P_{\mu
}-\delta _{\nu }^{\lambda }{}(P\theta )_{\mu }+\delta _{\mu }^{\lambda
}{}(P\theta )_{\nu }\right]
\end{equation}
However  
\begin{equation*}
\mathfrak{X}^{4}\rtimes \mathcal{U}_{\mathrm{iso}(1,3)}\cong \mathfrak{X}%
_{\theta }^{4}\rtimes \mathcal{U}_{\mathrm{iso}(1,3)}^{\theta }
\end{equation*}%
with the isomorphism implemented by changing generators  \newline
$x^{\mu }\rightarrow (\mathrm{\bar{f}}^{\alpha }\triangleright x^{\mu
})\cdot \bar{\mathrm{f}}_{\alpha }=$ $x^{\mu }+\frac{i}{2}\theta ^{\mu \nu
}P_{\nu }$.
Nevertheless an isomorphism between deformed  $\mathfrak{H}(\mathcal{U}%
_{\mathrm{iso}(1,3)}^{\theta })$ and undeformed $\mathfrak{H}(\mathcal{U}_{\mathrm{iso}%
(1,3)})$ Heisenberg doubles remains an open question due to the problems mentioned above.\\

\textbf{$\kappa$-deformed Poincare symmetry}\\

In the deformed case the $\kappa -$Poincare $\mathcal{U}_{\mathfrak{io}%
(1,3)}{}^{\mathrm{DJ}}$ algebra consists of (\ref{isog1}) and (\ref{isog2})
with Abelian translation sector $\left[ P_{\mu },P_{\nu }\right] =0$ as in
undeformed case but coalgebra structure is no longer primitive:
\begin{equation}
\Delta _{\kappa }\left( M_{i}\right) =\Delta _{0}\left( M_{i}\right)
=M_{i}\otimes 1+1\otimes M_{i}, \\
\Delta _{\kappa }\left( N_{i}\right) =N_{i}\otimes 1+\Pi _{0}^{-1}\!\otimes
N_{i}-\frac{1}{\kappa }\epsilon _{ijm}P_{j}\Pi _{0}^{-1}\!\otimes M_{m}
\end{equation}%
\begin{equation}
\Delta _{\kappa }\left( P_{i}\right) =P_{i}\otimes \Pi _{0}+1\otimes P_{i},
\\
\Delta _{\kappa }\left( P_{0}\right) =P_{0}\otimes \Pi _{0}+\Pi
_{0}^{-1}\otimes P_{0}+\frac{1}{\kappa }P_{m}\Pi _{0}^{-1}\otimes P^{m},
\end{equation}%
where we use the notation: $M_{0i}=iN_{i};\quad M_{ij}=\epsilon
_{ijk}M_{k};\quad \Pi _{0}=\left( \frac{P_{0}}{\kappa }+\sqrt{1-\frac{P_{\mu
}P^{\mu }}{\kappa ^{2}}}\right) $. In the view of Heisenberg double
construction for $\kappa $-Poincare \cite{LukNow} we are interested in : $%
\mathfrak{H}(\mathcal{F}_{\kappa }\left( \mathrm{ISO}(1,3)\right) ).$ In the
algebra dual to $\kappa -$Poincare: $\mathcal{F}_{\kappa }\left( \mathrm{ISO}%
(1,3)\right) $ only multiplication becomes deformed ($\ast _{\kappa }$), and
coalgebra sector stays unchanged. This new (deformed) multiplication leads
to $\kappa $-Minkowski spacetime algebra $\mathcal{U}_{h}(\mathfrak{an}^{3})$
with commutation relations between coordinates:
\begin{equation}
\left[ x^{\mu },x^{\nu }\right] _{\kappa }=\frac{i}{\kappa }\left( \delta
_{0}^{\mu }x^{\nu }-\delta _{0}^{\nu }x^{\mu }\right)
\end{equation}%
and to $\mathcal{F}^{\kappa }\left( \mathrm{SO}(1,3)\right) $ with $\left[
\Lambda _{\alpha }^{\beta },\Lambda _{\mu }^{\nu }\right] _{\kappa }=0$.
Both of them contribute to\newline
$\mathcal{F}^{\kappa }\left( \mathrm{ISO}(1,3)\right) =\mathcal{U}_{h}(%
\mathfrak{an}^{3})\rtimes \mathcal{F}^{\kappa }\left( \mathrm{SO}%
(1,3)\right) $ with crossed relations:
\begin{equation}
\left[ \Lambda _{\nu }^{\mu },x^{\rho }\right] =-\frac{i}{\kappa }\left(
\left( \Lambda _{0}^{\mu }-\delta _{0}^{\mu }\right) \Lambda _{\nu }^{\rho
}+\left( \Lambda _{\nu }^{0}-\delta _{\nu }^{0}\right) \eta ^{\mu \rho
}\right)
\end{equation}%
The Heisenberg double of $\kappa $- Poincare: $\mathcal{U}_{\mathfrak{io}%
(1,3)}{}^{\mathrm{DJ}}$ with its dual $\mathcal{F}_{\kappa }\left( \mathrm{%
ISO}(1,3)\right) $ contains besides the above also the following cross
relations :
\begin{equation}
\lbrack M_{ik},x^{\lambda }]=-<M_{ik},\Lambda _{\rho }^{\lambda }>x^{\rho
};\quad \left[ M_{0i},x^{\lambda }\right] =-<M_{0i}+\frac{1}{\kappa }%
M_{ij}P_{j},\Lambda _{\rho }^{\lambda }>\Pi _{0}^{-1}x^{\rho };
\end{equation}%
\begin{equation}
\lbrack M_{ik},\Lambda _{\beta }^{\alpha }]=-<M_{ik},\Lambda _{\beta }^{\rho
}>\Lambda _{\rho }^{\alpha };\quad \left[ M_{0i},\Lambda _{\beta }^{\alpha }%
\right] =-<M_{0i}+\frac{1}{\kappa }M_{ij}P_{j},\Lambda _{\beta }^{\rho }>\Pi
_{0}^{-1}\Lambda _{\rho }^{\alpha };
\end{equation}%
\begin{equation}
\lbrack P_{0},x^{k}]=-\frac{\imath }{\kappa }P_{k}\Pi _{0}^{-1};\qquad
\lbrack P_{k},x_{0}]=-\frac{\imath }{\kappa }P_{k};
\end{equation}%
\begin{equation}
\lbrack P_{k},x^{j}]=-\imath \delta _{k}^{j};\qquad \lbrack P_{0},x_{0}]=-%
\frac{\imath }{\kappa }P_{0}-\imath \Pi _{0}^{-1}
\end{equation}%
One can notice that the above set of relations differs from the smash
commutation relations of: $\mathcal{U}_{h}(\mathfrak{an}%
^{3})\rtimes \mathcal{U}_{\mathfrak{io}(1,3)}{}^{\mathrm{DJ}}$ which were
studed in \cite{BP4} and called DSR algebra. In this case a possible
isomorphism between deformed and undeformed Heisenberg doubles
\begin{equation*}
\mathfrak{H}(\mathcal{F(}\mathrm{ISO}(1,3))\cong \mathfrak{H}(\mathcal{F}%
_{\kappa }\left( \mathrm{ISO}(1,3)\right) )
\end{equation*}%
is also postponed to further investigation.


%
\end{document}